\catcode`\@=11
\newif\if@fewtab\@fewtabtrue
{\count255=\time\divide\count255 by 60
\xdef\hourmin{\number\count255}
\multiply\count255 by-60\advance\count255 by\time
\xdef\hourmin{\hourmin:\ifnum\count255<10 0\fi\the\count255}}
\def\ps@draft{\let\@mkboth\@gobbletwo
    \def\@oddfoot{\hbox to 7 cm{\tiny \versionno
       \hfil}\hskip -7cm\hfil\rm\thepage \hfil {\tiny\draftdate}}
    \def\@oddhead{}
    \def\@evenhead{}\let\@evenfoot\@oddfoot}
\def\draftdate{\number\month/\number\day/\number\year\ \ \ \hourmin }

\def\citen#1{\if@filesw \immediate\write \@auxout {\string\citation{#1}}\fi%
\@tempcntb\m@ne \let\@h@ld\relax \def\@citea{}%
\@for \@citeb:=#1\do {\@ifundefined {b@\@citeb}%
    {\@h@ld\@citea\@tempcntb\m@ne{\bf ?}%
    \@warning {Citation `\@citeb ' on page \thepage \space undefined}}%
    {\@tempcnta\@tempcntb \advance\@tempcnta\@ne
    \setbox\z@\hbox\bgroup\ifcat0\csname b@\@citeb \endcsname \relax
    \egroup \@tempcntb\number\csname b@\@citeb \endcsname \relax
    \else \egroup \@tempcntb\m@ne \fi \ifnum\@tempcnta=\@tempcntb
    \ifx\@h@ld\relax \edef \@h@ld{\@citea\csname b@\@citeb\endcsname}%
    \else \edef\@h@ld{\hbox{--}\penalty\@highpenalty
    \csname b@\@citeb\endcsname}\fi
    \else \@h@ld\@citea\csname b@\@citeb \endcsname \let\@h@ld\relax \fi}%
\def\@citea{,\penalty\@highpenalty\hskip.13em plus.13em minus.13em}}\@h@ld}
\def\@citex[#1]#2{\@cite{\citen{#2}}{#1}}%
\def\@cite#1#2{\leavevmode\unskip\ifnum\lastpenalty=\z@\penalty\@highpenalty\fi%
  \ [{\multiply\@highpenalty 3 #1%
  \if@tempswa,\penalty\@highpenalty\ #2\fi}]}   %
\makeatother % end of CITE.STY
\catcode`\@=12

\documentclass[12pt]{article}
\usepackage{latexsym, amsmath, amsthm, amsfonts, enumerate, amssymb, bbm, xspace, xypic }
\usepackage{fancybox}
\usepackage[all]{xy}
\usepackage[mathscr]{eucal}
\usepackage{graphicx,color}
\usepackage{rotating}
\usepackage{epstopdf,hyperref}

 \usepackage[T1]{fontenc}
 \usepackage[english]{babel}
 \usepackage[all]{xy}
 \usepackage{verbatim}
 \usepackage{mathrsfs}
 \usepackage{fancyhdr}
 \usepackage{stmaryrd}% To use \llbracket

\newtheorem{thm}{Theorem}

\newtheorem{lemma}[thm]{Lemma}
\newtheorem{cor}[thm]{Corollary}

\newtheorem{proposition}[thm]{Proposition}
\newtheorem{theorem}[thm]{Theorem}
\newtheorem{defi}[thm]{Definition}

\theoremstyle{definition}
\newtheorem{definition}[thm]{Definition}
\newtheorem{remark}[thm]{Remark}

\def\proof {\noindent{{\it Proof.}}\hspace{7pt}}
\def\endofproof {\hfill{$\Box$}\\}

\def\be{\begin{equation}}
\def\ee{\end{equation}}

\newcommand\eqpic[4]{\begin{eqnarray}
                   \begin{picture}(#2,#3){}\end{picture}\nonumber\\
                   \raisebox{-#3pt}{ \begin{picture}(#2,#3) #4 \end{picture} }
                   \label{#1} \\~\nonumber \end{eqnarray} }
\newcommand\Eqpic[4]{\begin{eqnarray}
                   \begin{picture}(#2,#3){}\end{picture}\nonumber\\
                   \raisebox{-#3pt}{ \begin{picture}(#2,#3) #4 \end{picture} }
                   \nonumber \\[3pt]~\label{#1} \end{eqnarray} }

\newcommand\Includepic[1]   {{\begin{picture}(0,0)(0,0)
                            \scalebox{.38}{\includegraphics{imgs/def_fact_#1.eps}}\end{picture}}}
\newcommand\includepic[1]   {{\begin{picture}(0,0)(0,0)
                            \scalebox{.32}{\includegraphics{imgs/def_fact_#1.eps}}\end{picture}}}
\newcommand\INcludepic[1] {{\begin{picture}(0,0)(0,0)
                   \scalebox{.456}{\includegraphics{imgs/def_fact_#1.eps}}\end{picture}}}

\newcommand\Includepicfj[2]   {{\begin{picture}(0,0)(0,0)
                            \scalebox{.#1}{\includegraphics{imgs/def_fact_#2.eps}}\end{picture}}}

\def\assa          {\ensuremath{\mathrm a}}
\def\assb          {\ensuremath{\mathrm x}}
\def\assc          {\ensuremath{\mathrm m}}
\def\assd          {\ensuremath{d}}
\def\assp          {\ensuremath{\mathrm p}}
\def\bearl         {\begin{array}{l}}
\def\bearll        {\begin{array}{ll}}
\def\cir           {\,{\circ}\,}
\def\wcS           {{\cS_\circ}}
\def\D             {\varGamma}
\def\dD            {\assd(\D)}
\def\DX            {\D_{\!\times}}
\def\eear          {\end{array}}
\def\eq            {\,{=}\,}
\def\essential     {intrinsic}
\def\Essential     {Intrinsic}
\def\essentially   {intrinsically}
\def\fusion        {network}
\def\hatt          {\widehat}
\newcommand\Hatwsf[2]{\widehat{\ws~}_{\!\!\!#1#2}}
\def\iN            {\,{\in}\,}
\def\M             {{\mathcal M}}
\def\Mapw          {\mathrm{Map}_{\mathrm w.s.}}
\def\Mapsto        {\,{\mapsto}\,}
\def\MGL           {\M_{G\!\ell}}
\def\MST           {\M_{\varSigma,\torus}}
\def\nregion       {network region}
\def\nxt           {\raisebox{.08em}{\rule{.41em}{.41em}}~\,}
\def\nxT           {\raisebox{.08em}{\rule{.41em}{.41em}}$\;$}
\def\nxx           {\raisebox{.11em}{\rule{.24em}{.24em}}$\;$}
\def\P             {\varPi}
\newcommand\Thet[4]{\Theta^{#1}_{#2#3,#4}}
\def\Times         {\,{\times}\,}
\def\To            {\,{\to}\,}
\def\TOR           {\TT_{\!X}}
\def\torus         {{\ensuremath{{\mathrm T}^2}}}
\def\TORS          {\mathcal T_{\!X}}
\def\TOT           {\TT_{X;pq,\alpha\beta}}
\def\TT            {{\mathcal {\tilde T}}}
\def\tunit         {{\ensuremath{\text{\bf 1}}}}
\def\tws           {{\mathring{\ws}}}
\def\ve            {^{\mathrm v}}
\newcommand\iline[1]{\begin{itemize}\item[\nxT] #1 \end{itemize}~\\[-3.1em]}
\def\vv            {^\pm}
\def\vvi	     {^\mp}
\def\ws            {{\ensuremath{\varSigma}}}
\newcommand\cdef[5]{(c^{\mathrm{def}}_{#1,#2#3})_{#4#5}^{\phantom|}}
\newcommand\cdefinv[5]{{(c^{\mathrm{def}\;-1}_{#1,#2#3})}_{\!#4#5}^{}}

\def\dim{\mathrm{dim}}
\def\dsty{\displaystyle}

\def\eps{\varepsilon}
    \newcommand\eTA[1]  {e_{#1}}

\newcommand\ExtTdef[1]    {\ensuremath{T_{#1,#1}}}

\def\Fiso   {\Phi_{\!A}^{}}
\def\Fisoi  {\Phi_{\!A}^{-1}}

   \def\GH    {\M^{\rm A}_{G\!\ell}}

\newcommand\GLL[2] {G\!\ell_{#1#2}}

\def\Hom{\ensuremath{\mathrm{Hom}}}
\def\HomP{\ensuremath{\mathrm{Hom}^P}}

\newcommand\Homaa[2]{\ensuremath{\mathrm{Hom}_{A|A}(#1,#2)}}
\newcommand\Homab[2]{\ensuremath{\mathrm{Hom}_{A|B}(#1,#2)}}

\def\I{\ensuremath{\mathcal I}}
\def\ia                     {{\ensuremath{\imath}}}
\def\ib                     {\ensuremath{{\bar\imath}}}
\def\id		                {\text{id}}

\newcommand\instord[5]  {\ensuremath{\mathcal{\overline L}^{#1}_{#2,#3,#4,#5}}}

\def\jb{\ensuremath         {{\bar\jmath}}}

\newcommand\K[6]            {{\em #6}, {#1} {#2} ({#3}) {#4} {\tt[#5]} }
\def\kb                     {\ensuremath {{\bar k}}}

\def\nl{\cdot}

\def\oti{{\otimes}}

\newcommand\PA[2] {P_{#1\oti\cdots\oti#2}}
\def\pb{{\bar p}}
\def\proj{\mathrm{P}}
\def\projMF{\mathcal{P}}
\def\pS {\pi}

\def\qb{{\bar q}}

\newcommand\RR[5]   {{\sf R}^{(#1\,#2)#3}_{#4\,#5}}
\newcommand\Rm[5]  {{\sf R}^{-\,(#1\,#2)#3}_{\;#4\,#5}}
    \newcommand\rTA[1]  {r_{#1}}

\def\Scut          {S_{\text{cut}}}
\def\Sp            {S^1}

\newcommand\tpbm[2]    {B^{-}_{#1#2}}

\newcommand\wsf[4]{\ws_{#1#2,#3#4}}
\newcommand\hatwsf[4]{\widehat{\ws_{#1#2,#3#4}}}
\newcommand\wsSD[5]{S^{2}_{#1;#2#3,#4#5}}
\def\wst {{\mathrm A}}
\def\wstc {{\mathrm A}'}
\def\wstcc {{\mathrm A}''}
\def\one {{\mathbf 1}}

\def\defas {\equiv}

\def\TD {\ensuremath{\D_{\!-}}}
\def\DTD {\ensuremath{\D_{\!+}}}
\def\R {\mathbb{R}}
\def\cC {\ensuremath{\mathcal{C}}}
\def\cS {\ensuremath{\mathcal{S}}}
\def\im {\mathrm{Im}}
\def\bl {\ensuremath{\mathcal{H}}}
\def\cH {\ensuremath{\mathcal{H}}}

\numberwithin{equation}{section}
\begin{document}

\def\cir{\,{\circ}\,} %% wg xypic
\numberwithin{equation}{section}
\numberwithin{thm}{section}

\begin{flushright}
    {\sf NITS-PHY-2012001}\\[2mm]
    February 2012
\end{flushright}
 \vskip 3.5em

\begin{center} \Large\bf RCFT WITH DEFECTS: FACTORIZATION \\[2mm]
               \Large\bf AND FUNDAMENTAL WORLD SHEETS
\end{center}

\vskip 2.8em

{\large
\begin{center}
  ~Jens Fjelstad\,$^{\,a}$,~
  ~J\"urgen Fuchs\,$^{\,b}$,~
  ~Carl Stigner\,$^{\,b}$
\end{center}
}

\vskip 12mm

\begin{center}\it$^a$
  Department of Physics, \ Nanjing University\\
  22 Hankou Road, \ Nanjing,\ 210093\, China
\end{center}
\begin{center}\it$^b$
  Teoretisk fysik, \ Karlstads Universitet\\
  Universitetsgatan 21, \ S\,--\,651\,88\, Karlstad
\end{center}

\vskip 5.5em

\noindent{\sc Abstract}
\\[3pt]
It is known that for any full rational conformal field theory,
the correlation functions that are obtained by the TFT construction
satisfy all locality, modular invariance and factorization
	conditions,
 and that there is a small set of fundamental correlators to which all others
are related via factorization -- provided that the world sheets considered
do not contain any non-trivial defect lines. In this paper we generalize
both results to oriented world sheets with an arbitrary network of topological
defect lines.

\newpage

%%%%%%%%%%%%%%%%%%%%%%%%%%%%%%%%%%%%%%%%%%%%%%%%%%%%%%%%%%%%%%%%%%%%%%%%

\section{Introduction}

The correlation functions of a full rational conformal field theory are strongly
constrained by consistency requirements: the locality, modular invariance and
factorization (or sewing) constraints.  Indeed, as has been shown in \cite{fjfrs2},
for any solution to the sewing constraints (with nondegenerate closed state
vacuum and nondegenerate two-point functions of boundary fields on the disk and
of bulk fields on the sphere), all correlators on arbitrary oriented world sheets are
already uniquely determined by the one-, two- and three-point functions on the disk.
The factorization constraints also allow one to obtain all correlators via sewing
from a small number of fundamental correlators \cite{sono2,lewe3,fips}.

A procedure for constructing all correlation functions as elements of the
appropriate spaces of conformal blocks of the corresponding chiral CFT (which,
in turn, are the spaces of solutions to the chiral Ward identities)
has been established in \cite{fuRs4,fuRs10,fjfrs}.
This procedure, called the \emph{TFT construction}, uses as an
input the systems of conformal blocks together with certain Frobenius algebras
in the category of representations of the chiral symmetry algebra. In
\cite{fjfrs} it was demonstrated explicitly that the correlators obtained from
the TFT construction do obey all locality, modular invariance and factorization
constraints.

The statement just made is in need of a further qualification, though. Namely,
the TFT construction gives the correlators on all world sheets,
including not only arbitrary field insertions in the bulk and on the boundary,
and arbitrary boundary conditions preserving the chiral symmetry, but also
arbitrary topological defect lines.
In contrast, the verification of the sewing identities in \cite{fjfrs} has
been carried out only for world sheets without any defect lines.
Or put differently, only factorization across \emph{trivial} defect lines
has been considered. To put the latter statement into context, recall that
in the TFT construction topological defect lines are labeled by bimodules
over the Frobenius algebras that characterize the full CFTs in the adjacent
regions of the world sheet. A trivial defect line separates two regions
characterized by one and the same Frobenius algebra and is labeled by that
Frobenius algebra.
Similarly, the list of fundamental correlators from which all others can be
obtained by sewing will no longer be exhausted by those considered in
\cite{sono2,lewe3,fips,fuRs10} when non-trivial defect lines are admitted.

The purpose of this paper is to close both of these gaps. First, we extend
the proof of factorization to world sheets with an arbitrary configuration
of defect lines. Second, we complete the list of fundamental world sheets
to include also world sheets with defects, taking into account that there are
various ways for world sheets to be equivalent, i.e.\ to have the same correlator.
Concretely, we show that the resulting list is still finite, and in fact not
much more complicated than the one obtained in the absence of defects. It
is worth stressing that it is not a priori guaranteed that the factorization
constraints in the presence of defects still give rise to a finite set of
fundamental world sheets. Indeed one has to allow for world sheets with
arbitrarily complicated defect networks, and one might have feared that the
factorization procedure is not sufficiently flexible for ending up with a
finite list.

We note that topological defects arise naturally e.g.\ in statistical systems
modeling condensed matter systems of physical interest, see e.g.\ \cite{savi,osaf2}.
The behavior of such a system can be very hard to analyze when a complicated defect
network is present. Our results allow one to discuss aspects of such a behavior by
considering instead only a finite set of correlators with simple defect configurations.
More generally, defect lines, as well as defect domain walls, are of much
current interest in other areas of quantum field theory as well (see e.g.\
\cite{kaWi,kaSau3,gamn,kiKon}); this provides an additional incentive for
studying basic properties of quantum field theories in the presence of defects.
In particular, the concise mathematical framework arising in our discussion may
still be suited in situations in which more heuristic concepts from conformal
field theory can no longer be applied.
    
To the best of our knowledge, aspects of factorization in the presence of
defects have so far only been addressed in \cite{petk4}, where the particular
case of crossing relations for four-point correlators on the sphere was
discussed. In our analysis the only restriction is that we take
the world sheets to be oriented, and accordingly the term CFT will be tacitly
understood as oriented CFT. But this restriction is only made for the sake
of brevity; our results can in fact easily be generalized to include
unoriented world sheets as well.
At a technical level our proof of factorization in principle follows the lines
of the corresponding proof in \cite{fjfrs}. We have, however, reformulated the
basic idea of factorization in such a manner that our arguments should be
accessible even without a full familiarity with the considerations in
\cite{fuRs4,fuRs10,fjfrs}.

\medskip

This paper is organized as follows. Section \ref{sec:bulkfac} is devoted to bulk
factorization. We first present, in Sections \ref{BF_outline} and
\ref{sec:gluingcob}, the basic ingredients of bulk factorization as well as
details about one important structure, the gluing cobordism. The precise
factorization statement is formulated in Section \ref{FactorizationIdentity}
(Theorem \ref{thm:bulkfac}, formula \eqref{fact_rel}), and its proof is given
in Section \ref{sec:factorizationproof}. As a preparation for the discussion of
fundamental correlators, Section
\ref{Worldsheets} deals with various properties of world sheets with defect lines.
First we provide in Section \ref{sec:defws} a precise definition of what we mean
by a world sheet with defects. In Section \ref{sec:eqWS} we then list a number
of ways in which world sheets can differ while still having equal correlators.
Section \ref{sec:bdfact} contains a brief discussion of boundary factorization
in the presence of defect lines, with the main result stated in Theorem
\ref{thm:bdfact}. Section \ref{sec:fundcorr} is devoted to fundamental
correlators. In Section \ref{sec:modcov} we state and prove the covariance and
invariance properties of correlators with defects (Theorem \ref{thm:modcov} and
Corollary \ref{cor:modinv}). In Section \ref{sec:fundws} we then present a list of
fundamental world sheets, and show that every correlator can be expressed in terms
of the correlators of these fundamental world sheets (Theorem \ref{thm:cS}). A
number of issues of more technical nature are collected in an Appendix.

%%%%%%%%%%%%%%%%%%%%%%%%%%%%%%%%%%%%%%%%%%%%%%%%%%%%%%%%%%%%%%%%%%%%%%%%

\section{Bulk factorization}\label{sec:bulkfac}

\subsection{Cutting and gluing}\label{BF_outline}

We start with a brief survey of the ingredients needed for the analysis of bulk
factorization. A correlator of a full CFT is associated to a \emph{world sheet}.
This is a surface with a conformal structure and with an embedded graph that
carries information about all field insertions, boundary conditions and
topological defect lines. Here, following \cite{baGa}, by a \emph{topological}
defect line we mean a defect line across which all chiral symmetries are continuous.
A topological defect line is thus in particular totally transmissive for the
stress-energy tensor. As a consequence, it can be deformed without affecting
the value of a correlator, as long as it is not taken across any field insertion
or through another defect line. All defect lines considered in the sequel will
be topological; hence we will usually refer to them just as defect lines, and
refer to (segments of) defect lines that are located next to each other as
running `parallel'.

For our purposes it is sufficient to regard world sheets as topological
manifolds. In the absence of defect lines, the structure of such
a world sheet is specified in detail in Definition B.2 of \cite{fjfrs}.
We will present a complete description including defect lines. Since, as it
turns out, factorization can be analyzed without using all details, this
description is postponed to Section \ref{sec:defws}.

\emph{Factorization} associates to a given world sheet $\ws$ a new world
sheet $\ws'$. In the case of \emph{bulk} factorization,
$\ws' \eq \wsf pq\alpha\beta$ is obtained as follows. The factorization is
performed along an embedded circle $\Scut$ that is contained in a cylindrical
region in the interior of $\ws$. In a first step, the world sheet is cut
along the circle $\Scut$, which gives rise to two new circular components
of the boundary of the world sheet. In the second step, the holes created
this way are closed by gluing a suitable hemisphere to each of these circular
boundary components.

Here we are interested in the situation that the cutting circle $\Scut$ is
crossed by finitely many defects which are running parallel.
Since topological defect lines can be fused (see \cite{pezu5,ffrs5} and also Section \ref{sec:equivfus} below),
it is actually enough to consider just a single defect line $X$ crossing the circle $\Scut$.
The hemispheres that are to be glued in the second step of the factorization
procedure are then obtained by cutting along the equator a specific world sheet
$\wsSD Xpq\alpha\beta$. As a surface, $\wsSD Xpq\alpha\beta$ is the two-sphere,
and it comes with two marked points (say, the North and South pole), at which disorder
fields $\Thet Xpq\alpha$ and $\Thet X{\pb}{\qb}\beta$ reside; and these disorder fields
are connected by the defect line $X$.
The resulting factorization is schematically displayed in the following picture:
  \eqpic{dis_sphere}{370}{16}{
  \put(0,0)     {\Includepic{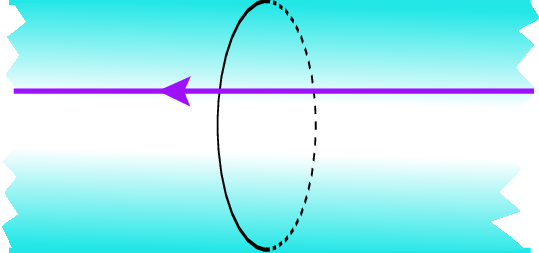}}
  \put(48,7)   {\begin{turn}{90}\scriptsize$\Scut$\end{turn}}
  \put(128,20)  {$\longmapsto$}
  \put(172,0)   {\Includepic{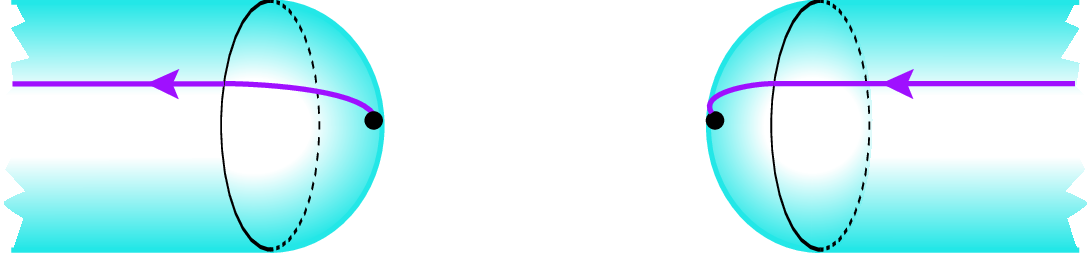}}
  \put(77,33)   {\scriptsize$X$}
  \put(187,33)  {\scriptsize$X$}
  \put(347,33)  {\scriptsize$X$}
  \put(245,22)  {\scriptsize$\Thet Xpq\alpha$}
  \put(278,22)  {\scriptsize$\Thet X\pb\qb\beta$}
  }
Two major impacts of the factorization are that the
new world sheet $\ws'$ has a different topology from that of $\ws$,
and that the set of field insertions has increased by two disorder fields.

Our task is now to relate the correlators on the world sheets $\ws$ and $\ws'$. The
correlator $C(\ws)$ of a world sheet $\ws$ is an element in the space $\cH(\widehat\ws)$
of conformal blocks on the complex double $\widehat\ws$ of $\ws$ \cite{ales,bcdcd,fuSc6}.
Since factorization changes the topology as well as the number of marked points on
the double, the spaces $\cH(\widehat\ws)$ and $\cH(\widehat{\ws'})$ are not
isomorphic. Still the correlators on $\ws$ and $\ws' \eq \wsf pq\alpha\beta$ can be
compared, with the help of a so-called \emph{gluing homomorphism}, which is a linear map
  \be\label{defGLL}
  \GLL pq:\quad \cH(\Hatwsf pq) \to \cH(\widehat\ws) \,.
  \ee
Here we write $\Hatwsf pq$, rather than $\hatwsf pq\alpha\beta$, for the double of the
world sheet $\wsf pq\alpha\beta$, in order to indicate that this two-manifold (and, as a
consequence, the associated space of conformal blocks) does not depend on the
labels $\alpha,\,\beta$ of the multiplicity space of disorder fields with chiral
labels $p$ and $q$.
The map $\GLL pq$ in \eqref{defGLL} is in fact precisely the same as the one defined in
formula (2.49) of \cite{fjfrs} and already used there in the proof of bulk factorization.
That this is still the correct gluing homomorphism in the more general situation
considered here is due to the fact that according to the TFT construction of
correlators, the relevant space of conformal blocks does not depend at all on
whether we deal with a trivial defect or with a non-trivial one.

The vector $\GLL pq\big( C(\wsf pq\alpha\beta )\big)$ lies in the same space
$\cH(\widehat\ws)$ as the original correlator $C(\ws)$. Indeed there is
\cite[Def.\,5.1.13(iv)]{BAki} an isomorphism
  \be
  \bigoplus_{p,q} \GLL pq:\quad
  \bigoplus_{p,q} \cH(\Hatwsf pq)\stackrel{\cong}{\longrightarrow} \cH(\widehat\ws)
  \ee
of vector spaces. In addition, further analysis with the help of the TFT
construction shows that the correlator $C(\ws)$ can be expressed as a
linear combination of the images $\GLL pq \big( C(\wsf pq\alpha\beta)\big)$
of the correlators for the factorized world sheets, i.e.\ one has
  \be\label{fact_rel_ansatz}
  C(\ws) = \sum_{p,q,\alpha,\beta} \xi_{pq,\alpha\beta}\; \GLL pq(C(\wsf pq\alpha\beta))
  \ee
with $\xi_{pq,\alpha\beta}\iN\mathbb C$. At the same time the TFT construction allows
one to express the coefficients $\xi_{pq,\alpha\beta}$ in terms of basic data of the CFT.

For explaining how the particular linear combination in question is found, we need to
recall the following information about the TFT construction. The category of
representations of the chiral symmetry algebra is a modular tensor category, to which
there is associated a three-dimensional topological field theory (TFT). In the TFT
construction the correlator $C(\ws)$ is interpreted as the invariant assigned by
that TFT to the \emph{connecting manifold} $\M_\ws$, a certain cobordism
$\M_\ws\colon \emptyset\To \partial\M_\ws \eq \widehat\ws$ with embedded ribbon graph
that is constructed from the data of the world sheet.
(The construction of $\M_\ws$ is detailed in e.g. \cite[App.\,B]{fjfrs}.)
Likewise, the gluing homomorphism is obtained as the invariant assigned by the TFT
to a cobordism
  \be
  \MGL \equiv \MGL{}_{,pq}\colon\quad \Hatwsf pq \to \widehat\ws \,,
  \ee
to which we will refer as the \emph{gluing cobordism},

One might be tempted to suspect that the two cobordisms $\M_\ws$ and
$\MGL \cir \M_{\ws'}$ from $\emptyset$ to $\widehat\ws$, while clearly containing
different ribbon graphs, at least coincide as topological manifolds. This is not
the case, though. However \cite[Figs.\,(5.3)\,\&\,(5.9)]{fjfrs}, the discrepancy
between the two manifolds, including their embedded ribbon graphs, is entirely
confined in a suitable embedded solid torus. More specifically, one
can realize $\M_\ws$ and $\MGL \cir \M_{\ws'}$ as compositions
  \be
  \M_\ws = \MST \circ \TORS
  \label{MST-TORS}
  \ee
and
  \be
  \MGL \circ \M_{\ws'} = \MST \circ \TOR
  \label{MST-TOR}
  \ee
of cobordisms, respectively, where $\MST\colon \torus\To \widehat\ws$ is a
cobordism from the torus \torus\
with two marked points labeled by $X$ to the double of the world sheet $\ws$, while
$\TOR$ and $\TORS$ are two different solid tori, regarded as cobordisms from the
empty set to \torus\ with two marked points labeled by $X$.

The coefficients in the expansion \eqref{fact_rel_ansatz}, and thus the
factorization identity, can therefore be determined by obtaining the precise
relationship between the invariant of the cobordism $\TORS$ and those of the cobordisms
$\TOR \eq \TOT$ for all values of the labels $p$, $q$, $\alpha$ and $\beta$.
Disregarding ribbon graphs, the two manifolds $\TORS$ and $\TOR$ are, informally,
related by a modular S-transformation of their boundary \torus. More specifically,
denote by $\mathrm{M}_\mathrm{S}$ the mapping cylinder over \torus{} of a homeomorphism
in the class of the modular S-transformation; then $\TORS$ and
$\mathrm{M}_\mathrm{S}\cir\TOR$ are related by an orientation preserving homeomorphism
that restricts to the identity on their common boundary.
Accordingly, the desired relation between the invariants of the cobordisms
\eqref{MST-TORS} and \eqref{MST-TOR} is referred to as a \emph{surgery relation},
and the geometrical input also gives a hint on what this relation can look like.
In the absence of defect lines, the relevant surgery relation is given in formula
(5.12) of \cite{fjfrs}. In Proposition \ref{surg_prop}
we will establish the generalization of that relation to the situation of our interest.
The resulting factorization identity is stated in Theorem \ref{thm:bulkfac} below.

%%%%%%%%%%%%%%%%%%%%%%%%%%%%%%%%%%%%%%%%%%%%%%%%%%%%%%%%%%%%%%%%%%%%%%%%

\subsection{The gluing cobordism}\label{sec:gluingcob}

Since our proof of the factorization identity will largely follow the lines
of the proof in \cite{fjfrs}, we will not need to present all details about
the cobordisms appearing in \eqref{MST-TORS} and \eqref{MST-TOR},
but can concentrate on those parts in which they
differ from the situation without defect lines. However, in order to provide
some impression of the structure of the manifolds involved, including the relevant
modular S-transformation, we present here a schematic pictorial description.
This description differs somewhat from the one in \cite{fjfrs}; we hope that it
is more easily accessible.

As a main tool, we illustrate orientable surfaces and three-manifolds, or pieces
thereof, through specific projections to $\R$ and to $\R^2$, respectively. Let us
first describe these projections for the case of surfaces. We regard an orientable
surface as embedded in $\R^3$, parameterized by Cartesian coordinates $(x,y,z)$; the
surface is then projected to the real line by forgetting the $y$- and $z$-coordinates,
i.e.\ according to
  \be\label{def_proj}
  \pS:\quad\R^3\rightarrow \R\,, \qquad (x,y,z) \,\mapsto\, x\,.
  \ee
For instance, for the unit sphere $S^2 \eq \{(x,y,z)\,|\,x^2+y^2+z^2\eq1\}
\,{\subset}\,\R^3$ this gives $\pS(S^2) \eq [-1,1]$, or in pictures:
  \eqpic{proj_S2}{280}{33}{
  \put(0,0)    {\Includepic{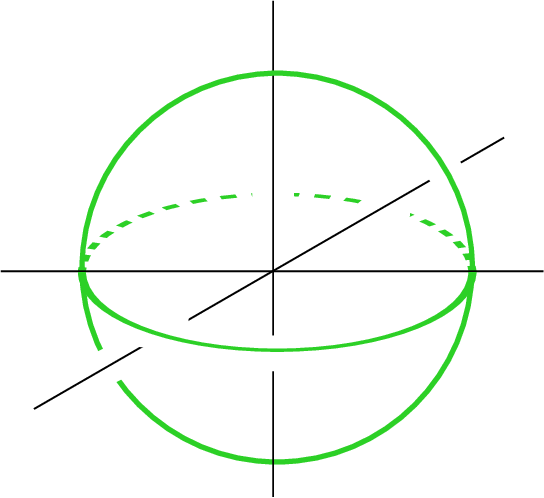}
  \put(97,36)    {\scriptsize$x$}
  \put(91,61)    {\scriptsize$y$}
  \put(52,89)    {\scriptsize$z$}
  }
  \put(184,39) {\Includepic{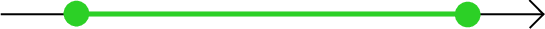}
  \put(7,-9)     {\scriptsize$-1$}
  \put(83.3,-9)  {\scriptsize$1$}
  \put(98,-5)    {\scriptsize$x$}
  }
  \put(140,39)    {$\stackrel\pi\longmapsto$}
  }
The preimage of a point $p\iN (-1,1)$ is a circle,
$\pS^{-1}(p) \eq \Sp(p)$ for $p\iN(-1,1)$, where we set
  \be
  \Sp(p) := \{(p,y,z) \,|\, y,z\iN\R\,,~ y^2+z^2 \eq 1-p^2\} \,\subset\R^3 \,.
  \ee
For $p\,{\not\in}\,[-1,1]$ we have $\pS^{-1}(p) \eq \emptyset$, while
over the points $\pm1$ the circle $\Sp(p)$ degenerates to radius zero, i.e.\
the fibers over $\pm1$ are just points, $\pS^{-1}(\pm1) \eq (\pm1,0,0)$. To
emphasize the special nature of these fibers, the points $\pm1$ have been
marked by blobs in the picture \eqref{proj_S2}.

As another illustration, the action of $\pS$ on a disk $D^2$, viewed as a half-sphere,
looks as follows:
  \eqpic{proj_D2}{230}{36}{
  \put(0,0)    {\Includepic{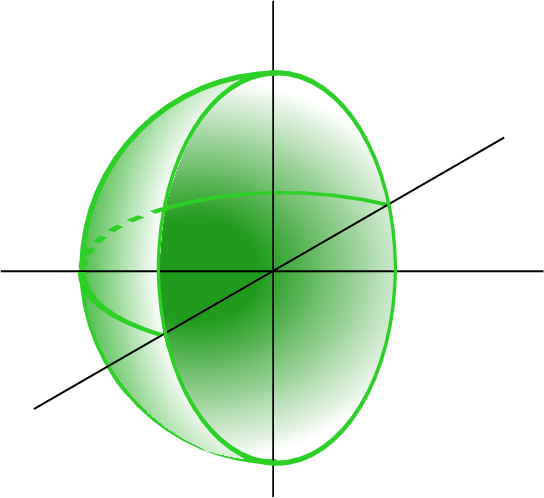}
  \put(97,36)    {\scriptsize$x$}
  \put(91,61)    {\scriptsize$y$}
  \put(52,89)    {\scriptsize$z$}
  }
  \put(194,39) {\Includepic{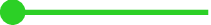}
  }
  \put(140,39)    {$\stackrel\pi\longmapsto$}
  }
Here the left end-point of $\pS(D^2)$ is of the same type as the end-points in
\eqref{proj_S2}, while the right end-point is the image of the boundary circle
$\partial D^2$.

Similarly, the three-manifolds of our interest are regarded as embedded in
$\R^4$, parameterized by Cartesian coordinates $(x,y,z,t)$, and as
projected to the paper plane $\R^2$ by again forgetting the $y$- and $z$-coordinates,
i.e.\ according to
  \be\label{pSvol}
  \pS:\quad\R^4\rightarrow \R^2\,, \qquad (x,y,z,t) \,\mapsto\, (x,t)\,.
  \ee
Using the same symbol $\pi$ as before is justified because
the projection \eqref{def_proj} is merely a special case of \eqref{pSvol},
obtained by restricting to $t \eq 0$.
As an example, for a cylinder over the unit sphere $S^2\,{\subset}\,\R^3$ we get
$\pS([-1,1]\Times S^2) \eq [-1,1]\times[-1,1]$, which we draw as
  \eqpic{proj_S2cyl}{213}{36}{
  \put(0,36)  {$ \pS([-1,1]\Times S^2) ~= $}
  \put(110,-3){
  \put(0,0)     {\Includepic{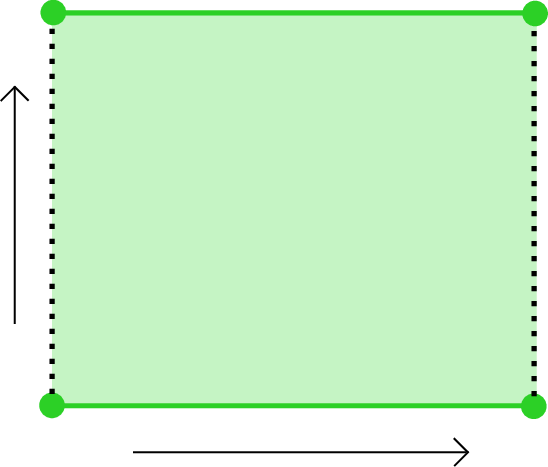}}
  \put(-4,68)   {\scriptsize$t$}
  \put(84,-4)   {\scriptsize$x$}
  } }
Note that $\pS^{-1}(t,x) \eq \{t\}\Times \Sp(x)$ for $x\iN (-1,1)$ and $t\iN\R$,
while $\pS^{-1}(t,\pm1)\eq (t,\pm1,0,0)$.

The boundary of the three-manifold represented by \eqref{proj_S2cyl} is projected to
the two solid lines. The dotted lines, on the other hand, consist of those points
$p\iN \R^2$ for which $\pS^{-1}(p)$ is a single point. The fiber over any other
point in $(-1,1)\Times[-1,1]$ (the shaded region) is a circle $S^1$. In particular,
the preimage of a point on the dotted lines does not belong to the boundary of
the three-manifold.
To illustrate this issue with another example, we also display the image under
$\pS$ of a solid cylinder $[-1,1]\Times D^2$ (for the disk $D^2$ being described
as in \eqref{proj_D2}):
  \eqpic{proj_D2cyl}{173}{33}{
  \put(0,37)  {$ \pS([-1,1]\Times D^2) ~= $}
  \put(110,-3){
  \put(0,0)     {\Includepic{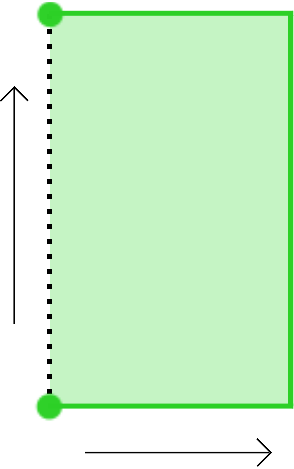}}
  \put(-4,68)   {\scriptsize$t$}
  \put(49,-4)   {\scriptsize$x$}
  } }

For studying the gluing cobordism the following three surfaces are of particular
interest:
\\[-2.66em]

\def\leftmargini{1.57em}~\\[-1.45em]\begin{itemize}\addtolength{\itemsep}{-7pt}
  \item[\nxt]
  $\wst \,{\equiv}\, \wst(\ws)$:\,
  \\ The cylindrical region of the world sheet $\ws$\,
  on which the factorization is performed;
  \item[\nxt]
  $\wstc$:\, the corresponding region after factorization;
  \item[\nxt]
  $\wstcc$:\, a corresponding region after a \emph{double cut} procedure.
\end{itemize}

\noindent
The projection \eqref{proj_S2cyl} acts on such surfaces as follows:
  \eqpic{tube_ws}{305}{16}{
  \put(40,-3)     {\Includepic{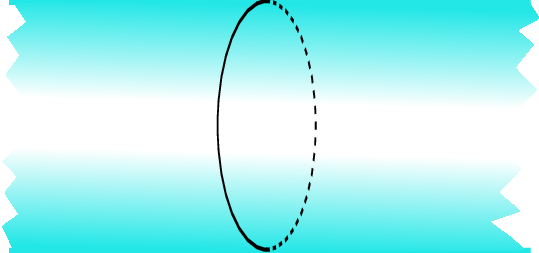}}
  \put(0,16)   {$ \wst ~= $}
  \put(220,20)   {\Includepic{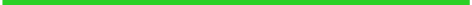}}
  \put(174,16)   {$\stackrel{\pS}{\longmapsto}$}
  }
  \eqpic{tube_cut_ws}{360}{14}{
  \put(44,-4)     {\Includepic{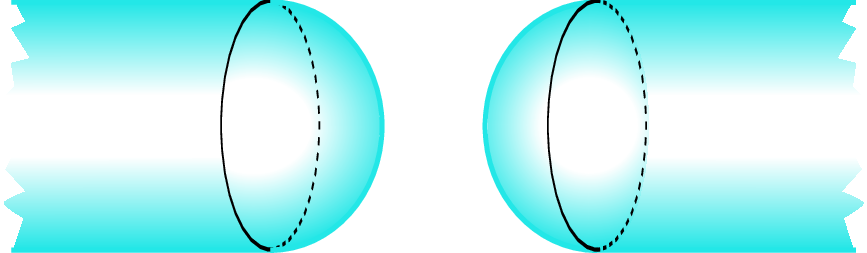}}
  \put(0,14)   {$ \wstc ~= $}
  \put(274,14)   {\Includepic{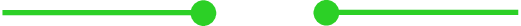}}
  \put(234,14)   {$\stackrel{\pS}{\longmapsto}$}
  }
  \Eqpic{tube_ccut_ws}{372}{14}{
  \put(0,0)     {\Includepic{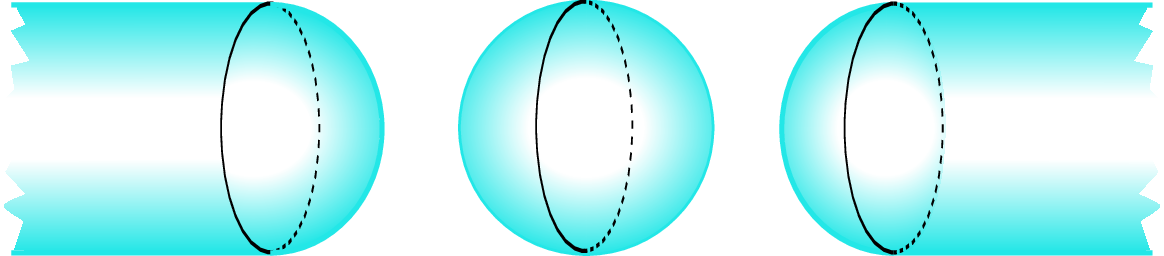}}
  \put(-40,21)   {$ \wstcc ~= $}
  \put(280,20)   {\Includepic{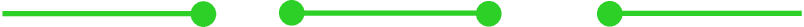}}
  \put(240,20)   {$\stackrel{\pS}{\longmapsto}$}
  }
For the cylinders over these regions the projection \eqref{pSvol} gives
  \eqpic{tube_ws}{320}{16}{
  \put(120,0)   {\Includepic{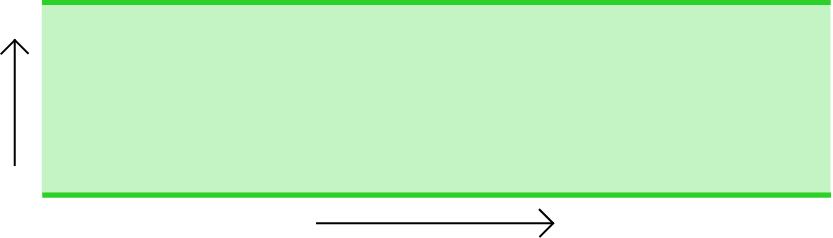}
  \put(-5,34.5) {\scriptsize$t$}
  \put(101,-4)  {\scriptsize$x$}
  \put(154,6)   {\scriptsize$t=-1$}
  \put(155,41)  {\scriptsize$t=+1$}
  }
  \put(6,19)   {$ \pS([-1,1]\Times\wst) ~= $}
  }
  \eqpic{tube_cut_ws}{320}{16}{
  \put(0,-2){
  \put(120,0)   {\Includepic{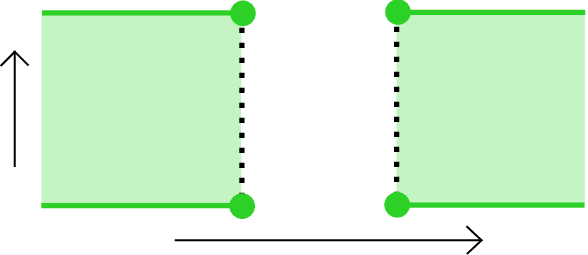}
  \put(-5,35.5) {\scriptsize$t$}
  \put(86,-4)   {\scriptsize$x$}
  \put(111,7.5) {\scriptsize$t=-1$}
  \put(111,42.5){\scriptsize$t=+1$}
  } }
  \put(6,20)   {$ \pS([-1,1]\Times \wstc) ~= $}
  }
  \eqpic{tube_ccut_ws}{320}{18}{
  \put(0,-2){
  \put(120,0)   {\Includepic{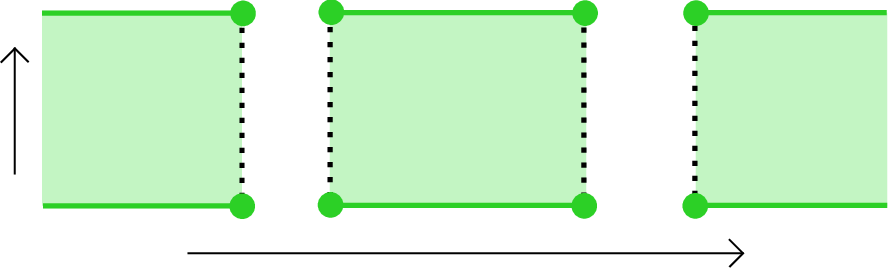}
  \put(-4,39)   {\scriptsize$t$}
  \put(134.5,-5){\scriptsize$x$}
  \put(166,9)   {\scriptsize$t=-1$}
  \put(166,44)  {\scriptsize$t=+1$}
  } }
  \put(6,17)   {$ \pS([-1,1]\Times\wstcc) ~= $}
  }

We are now in a position to address the construction of the gluing cobordism
in terms of the projection $\pS$. When doing so, we must in addition account
for the following two aspects.
First, the two-manifolds to be considered are \emph{extended surfaces},
meaning in particular \cite[Def.\,5.1.6]{BAki} that they come with a set of
marked points. Such marked points are labeled as $(U_k,\pm)$, with $U_k$ a
simple object of the representation category \cC\ of the chiral symmetry
algebra \cite{fuRs4}. In order not to overburden the pictures below, we
abbreviate such a label $(U_k,\pm)$ by the symbol $k_\pm$.

Second, the plane in which the picture is drawn can no longer be taken to be
just the $t$-$x$-plane; the projection $\pS$ must therefore be suitably
redefined. As the cylinder over $\hatt\wstcc$ is the disjoint union of two
three-manifolds, we can embed this region in $\R^4 \eq \{(t',x,y,z)\}$ by
making the replacement $t\,{\mapsto}\, t'{+}2$ on one component and
$t\,{\mapsto}\, {-}t'{-}2$ on the
other. The vertical axis in the picture then coincides with the
$t'$-axis. This will be implicitly understood in the sequel, and we will
refrain from indicating the ambient parameter $t'$ in our pictures.

We start from a cylinder over the double $\hatt\wstcc$ of $\wstcc$:
  \eqpic{tube_tube_ws}{350}{48}{
  \put(120,0)     {\Includepic{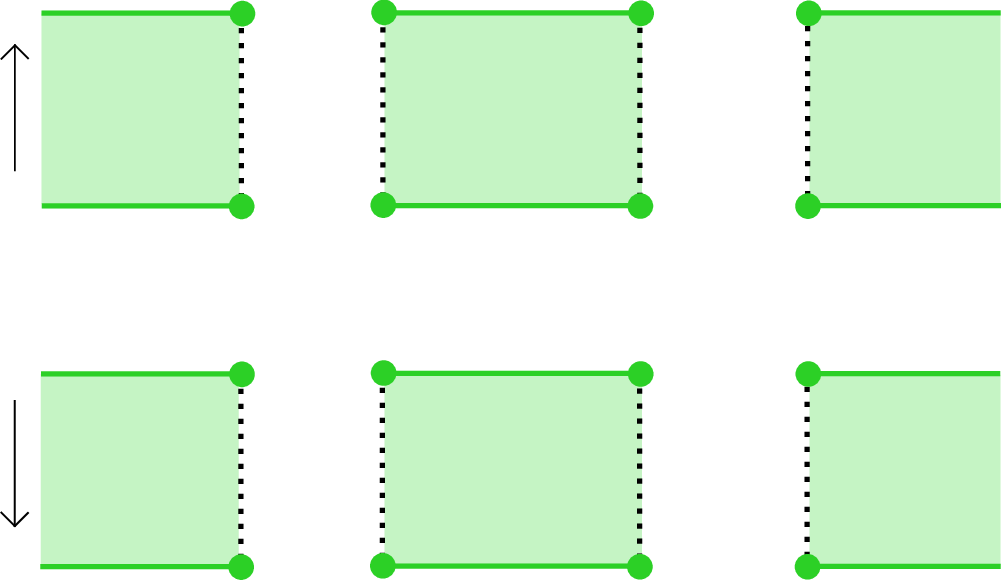}
  \put(-4,6)     {\scriptsize$t$}
  \put(-4,97)     {\scriptsize$t$}
  \put(186,1)     {\scriptsize$t=+1$}
  \put(186,36)    {\scriptsize$t=-1$}
  \put(186,67)    {\scriptsize$t=-1$}
  \put(186,102)   {\scriptsize$t=+1$}
  \put(39,59)     {\scriptsize$ \ib_- $}
  \put(66,59)     {\scriptsize$ \ib_+ $}
  \put(111,59)    {\scriptsize$ i_+ $}
  \put(144,59)    {\scriptsize$ i_- $}
  \put(39,43)     {\scriptsize$ \jb_- $}
  \put(66,43)     {\scriptsize$ \jb_+ $}
  \put(111,43)    {\scriptsize$ j_+ $}
  \put(144,43)    {\scriptsize$ j_- $}
  \put(39,110)    {\scriptsize$ \ib_- $}
  \put(66,110)    {\scriptsize$ \ib_+ $}
  \put(111,110)   {\scriptsize$ i_+ $}
  \put(144,110)   {\scriptsize$ i_- $}
  \put(38,-5)     {\scriptsize$ \jb_- $}
  \put(65,-5)     {\scriptsize$ \jb_+ $}
  \put(110,-5)    {\scriptsize$ j_+ $}
  \put(143,-5)    {\scriptsize$ j_- $}
  }
  \put(7,48)   {$ \pS\big([-1,1]\Times\hatt\wstcc\big) ~= $}
  }
The gluing cobordism is obtained from the cobordism $[-1,1]\Times\hatt\wstcc$
as follows. First we quotient out a relation, to be denoted by the symbol
``$\sim$'', by which segments on the $t\eq{+}1$-components of the boundary are
pairwise identified, in the way indicated in the following picture by the
arrows that are attached to the eight segments in question:
  \eqpic{tube_tube_ws_id1}{400}{63}{ \put(110,20) { \setlength\unitlength{1.2pt}
  \put(0,-9)  {
  \put(0,-2)  {\INcludepic{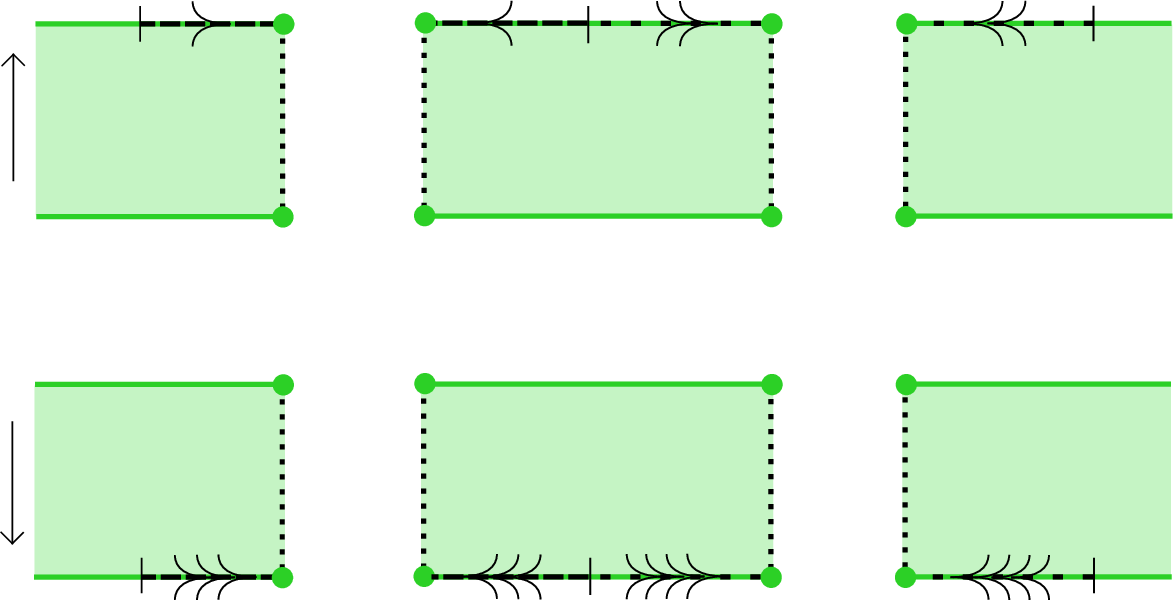}}
  \put(-3,6)      {\scriptsize$t$}
  \put(-3,98)     {\scriptsize$t$}
  \put(219,1)     {\scriptsize$t=+1$}
  \put(219,36)    {\scriptsize$t=-1$}
  \put(219,67)    {\scriptsize$t=-1$}
  \put(219,102)   {\scriptsize$t=+1$}
  \put(48,60)     {\scriptsize$ \ib_- $}
  \put(74,60)     {\scriptsize$ \ib_+ $}
  \put(137,60)    {\scriptsize$ i_+ $}
  \put(162,60)    {\scriptsize$ i_- $}
  \put(48,43)     {\scriptsize$ \jb_- $}
  \put(74,43)     {\scriptsize$ \jb_+ $}
  \put(137,43)    {\scriptsize$ j_+ $}
  \put(162,43)    {\scriptsize$ j_- $}
  \put(48,109)    {\scriptsize$ \ib_- $}
  \put(74,109)    {\scriptsize$ \ib_+ $}
  \put(137,109)   {\scriptsize$ i_+ $}
  \put(162,109)   {\scriptsize$ i_- $}
  \put(48,-5)     {\scriptsize$ \jb_- $}
  \put(74,-5)     {\scriptsize$ \jb_+ $}
  \put(137,-5)    {\scriptsize$ j_+ $}
  \put(162,-5)    {\scriptsize$ j_- $}
  } }
  \put(-25,71)   {$ \pS\big([-1,1]\Times\hatt\wstcc{\big/}{\sim}\,\big) ~= $}
  }
For the manifold described by the picture this means concretely that four pairs
of half-spheres are pairwise identified (in \cite{fjfrs} this is described by
the formula (2.15)). Carrying out this identification leads to
  \eqpic{tube_tube_ws_id2}{300}{45}{
  \put(120,0)     {\Includepic{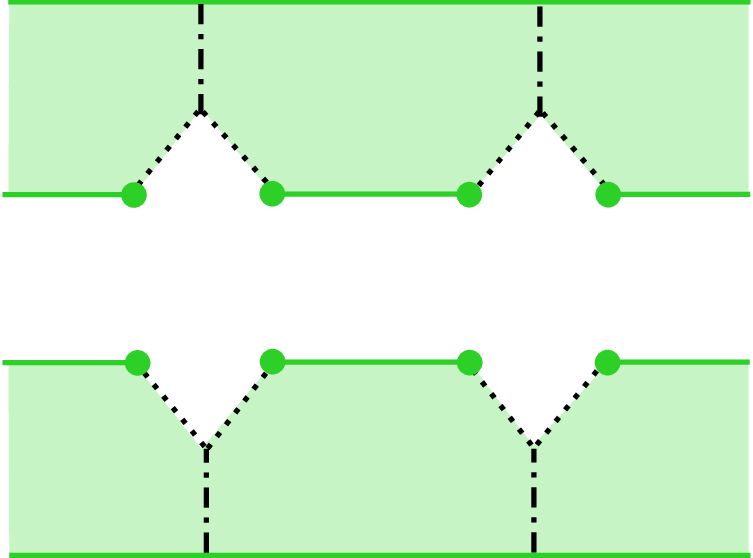}
  \put(142,-2)    {\scriptsize$t=+1$}
  \put(142,33)    {\scriptsize$t=-1$}
  \put(142,64)    {\scriptsize$t=-1$}
  \put(142,98)    {\scriptsize$t=+1$}
  \put(20,57)     {\scriptsize$ \ib_- $}
  \put(47,57)     {\scriptsize$ \ib_+ $}
  \put(82,57)     {\scriptsize$ i_+ $}
  \put(108,57)    {\scriptsize$ i_- $}
  \put(20,41)     {\scriptsize$ \jb_- $}
  \put(47,41)     {\scriptsize$ \jb_+ $}
  \put(82,41)     {\scriptsize$ j_+ $}
  \put(108,41)    {\scriptsize$ j_- $}
  }
  \put(-8,49)   {$ \pS\big([-1,1]\Times\hatt\wstcc{\big/}{\sim}\,\big) ~= $}
  }
Here the points on the vertical dashed-dotted lines come from the identification
of points on the boundary, but are now part of the interior of the manifold.

Note that in \eqref{tube_tube_ws_id2} the $(t,x)$-coordinates no longer correspond
to those of the original cylinder over $\hatt\wstcc$. Instead the picture shows
the result of first applying an obvious isotopy to the embedded three-manifold
in $\R^4$, and then projecting to $\R^2$. (Owing to this deformation, from now on
the $t$-direction is no longer everywhere vertical.) In the following we
allow for such isotopies of three-manifolds embedded in $\R^4$ before applying $\pS$,
which should cause no confusion. Thus \eqref{tube_tube_ws_id2} can be redrawn as
  \eqpic{tube_tube_ws_id3}{300}{46}{
  \put(120,0)  {\Includepic{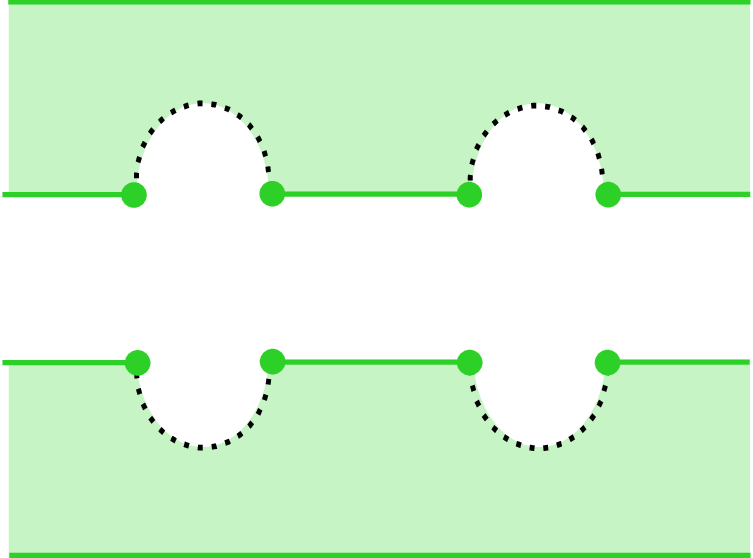}
  \put(142,-2)    {\scriptsize$t=+1$}
  \put(142,33)    {\scriptsize$t=-1$}
  \put(142,63)    {\scriptsize$t=-1$}
  \put(142,98)    {\scriptsize$t=+1$}
  \put(20,57)     {\scriptsize$ \ib_- $}
  \put(47,57)     {\scriptsize$ \ib_+ $}
  \put(82,57)     {\scriptsize$ i_+ $}
  \put(108,57)    {\scriptsize$ i_- $}
  \put(20,41)     {\scriptsize$ \jb_- $}
  \put(47,41)     {\scriptsize$ \jb_+ $}
  \put(82,41)     {\scriptsize$ j_+ $}
  \put(108,41)    {\scriptsize$ j_- $}
  }
  \put(-6,49)   {$ \pS\big([-1,1]\Times\hatt\wstcc{\big/}{\sim}\,\big) ~= $}
  }
The manifold above is actually
a cobordism with an inscribed ribbon graph, consisting of
rectangular ribbons running along the dotted lines; in the picture we have
suppressed these ribbons altogether. (In a three-dimensional description, the
cobordism \eqref{tube_tube_ws_id3} is shown, including the rectangular ribbons,
in picture (5.5) of \cite{fjfrs}.)

The underlying three-manifold of the gluing cobordism is now obtained by gluing
two three-balls to the manifold shown in \eqref{tube_tube_ws_id3}. The relevant
three-balls are cobordisms $\tpbm{\kb}k$ with two marked points on the boundary
and inscribed ribbon graph. A three-dimensional view of $\tpbm{\kb}k$ is
given in picture (2.34) of \cite{fjfrs}; in the present description we have
  \eqpic{B3proj}{140}{18}{
  \put(70,0)  {\Includepic{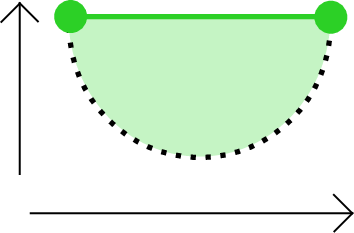}
  \put(-4,40)     {\scriptsize$t$}
  \put(63,-4)     {\scriptsize$x$}
  \put(8,44.6)    {\scriptsize$ \kb_- $}
  \put(56,44.6)   {\scriptsize$ k_- $}
  }
  \put(0,23)   {$\pS(\tpbm{\kb}k) ~= $}
  }
Gluing $\tpbm{\ib}i\,{\sqcup}\,\tpbm{\jb}j$ to the $t\,{=}\,{-}1$-components
of the boundary of \eqref{tube_tube_ws_id3} we obtain the underlying manifold $\GH$
of the gluing cobordism
as
  \Eqpic{gluinghomo_MF}{430}{49}{ \put(0,7){
  \put(65,0)   {\Includepic{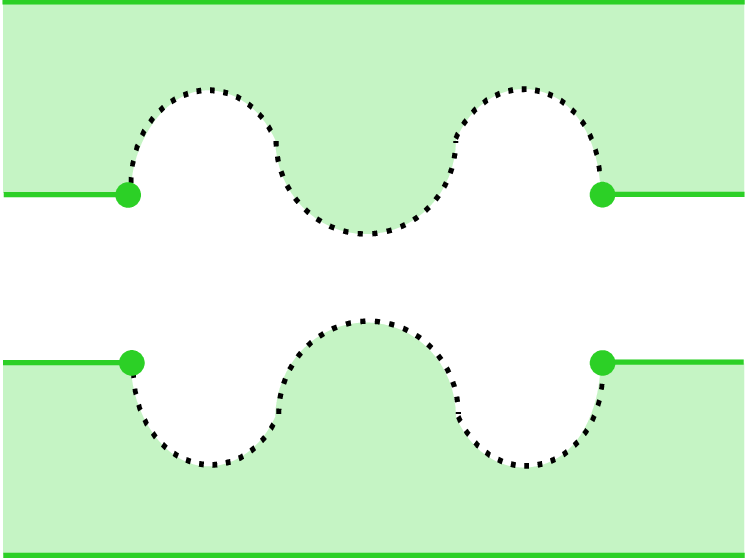}
  \put(141,-2)     {\scriptsize$t=+1$}
  \put(141,33)    {\scriptsize$t=-1$}
  \put(141,65)    {\scriptsize$t=-1$}
  \put(141,100)   {\scriptsize$t=+1$}
  \put(20,57)     {\scriptsize$ \ib_- $}
  \put(106,57)    {\scriptsize$ i_- $}
  \put(20,41)     {\scriptsize$ \jb_- $}
  \put(106,41)    {\scriptsize$ j_- $}
  }
  \put(252,48)    {$=$}
  \put(284,0)     {\Includepic{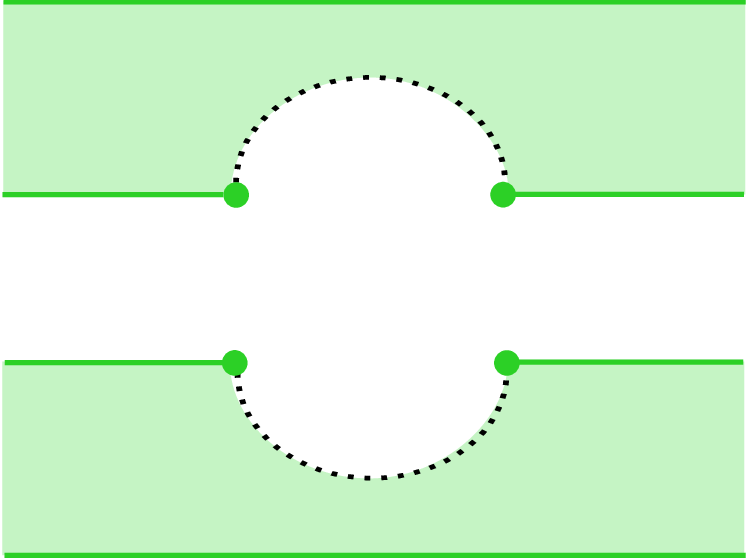}
  \put(141,-2)     {\scriptsize$t=+1$}
  \put(141,33)    {\scriptsize$t=-1$}
  \put(141,65)    {\scriptsize$t=-1$}
  \put(141,100)   {\scriptsize$t=+1$}
  \put(38,57)     {\scriptsize$ \ib_- $}
  \put(87,57)     {\scriptsize$ i_- $}
  \put(38,41)     {\scriptsize$ \jb_- $}
  \put(87,41)     {\scriptsize$ j_- $}
  }
  \put(-11,45)   {$\pS(\GH) ~= $}
  } }

Having obtained the gluing cobordism,
we can apply it to the connecting manifold of the surface $\wstc$.
The connecting manifold of $\wstc$ is $\M_{\wstc} \eq [-1,1 ]\Times \wstc$, so that
  \eqpic{con_mf_cut}{270}{16}{
  \put(173,0)  {\Includepic{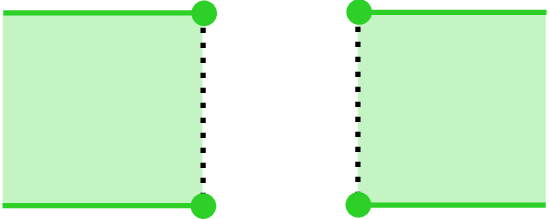}
  \put(33,44)     {\scriptsize$ \ib_+ $}
  \put(60,44)     {\scriptsize$ i_+ $}
  \put(32,-6)     {\scriptsize$ \jb_+ $}
  \put(59,-6)     {\scriptsize$ j_+ $}
  }
  \put(0,18)   {$ \pS(\M_{\wstc}) ~=~ \pS([-1,1]\Times\wstc) ~= $}
  }
Here again we have suppressed the pieces of ribbon graph that are running along
the dotted lines (in a three-dimensional view they are shown in picture
(5.8) of \cite{fjfrs}).

Applying the gluing cobordism amounts to identifying
the $t\eq{-}1$\,-component of the boundary of \eqref{gluinghomo_MF} with the
boundary of the three-manifold \eqref{con_mf_cut}. The result is
  \eqpic{gluingtocut}{255}{44}{
  \put(113,0)  {\Includepic{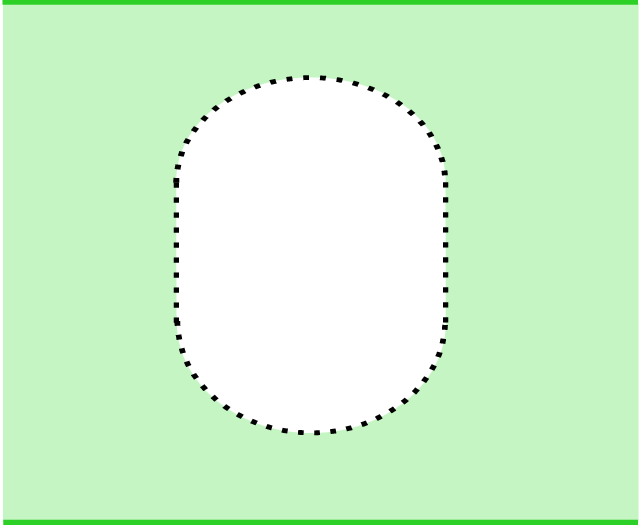}
  \put(72.5,31)   {\small$\mathcal S$}
  \put(121,-2)    {\scriptsize$t=-1$}
  \put(121,33)    {\scriptsize$t=+1$}
  \put(121,58)    {\scriptsize$t=-1$}
  \put(121,93)    {\scriptsize$t=+1$}
  }
  \put(1,46)   {$ \pS(\GH\cir\M_{\wstc}) ~= $}
  }
Again there is a piece of ribbon graph running along the dotted circle
$\mathcal S$ (as displayed in picture (5.9) of \cite{fjfrs}), and again we
refrain from drawing it explicitly. Note that the fiber over any point of
the circle $\mathcal S$ is just a point.

The three-manifold $\GH\cir\M_{\wstc}$ coincides up to a suitable surgery (and
up to an embedded ribbon graph) with $\M_{\wst}$ as displayed in \eqref{tube_ws}.
To see this we perform the following surgery on $\GH\cir\M_{\wstc}$.
We first cut out a tubular neighborhood of the circle $\pS^{-1}(\mathcal S)$.
(Recalling the presentation \eqref{proj_D2cyl} of solid cylinders,
it is easy to describe the resulting solid torus $\mathcal T$ explicitly in terms
of the coordinates of the ambient $\R^4$, as the union of solid cylinders in
$[-1,1]\Times \wstc$ and in $\GH$; we refrain from giving any details.)
Furthermore we can apply a suitable isotopy to $\GH\cir\M_{\wstc}$ in such a way
that the boundary of the solid torus $\mathcal T$ projects under $\pS$ to a
circle in the $x$-$t$-plane. (Here we use that we have identified the coordinate
$t$ with the ambient parameter, denoted $t'$ above.)
The result of cutting out $\mathcal T$ from $\GH\cir\M_{\wstc}$ is then
  \Eqpic{gluingtocut_surg}{420}{40}{ \put(53,-5){
  \put(160,0)  {\Includepic{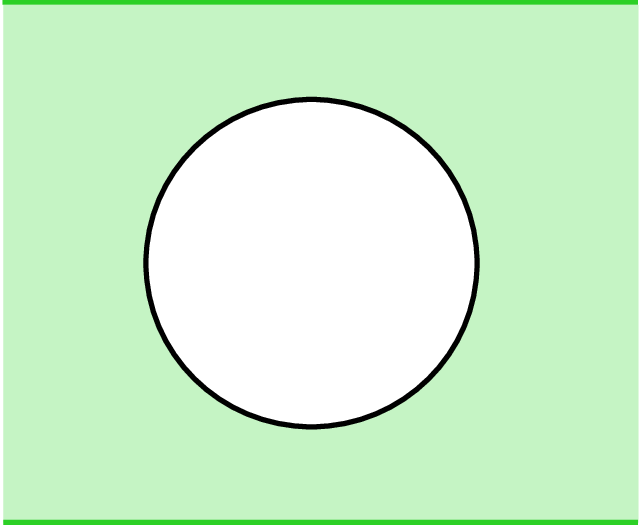}
  }
  \put(301,42)  {$\sqcup$}
  \put(331,16)  {\Includepic{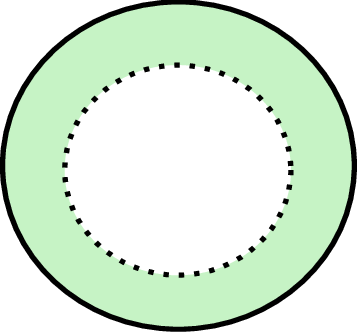}
   \put(42.8,21)   {\small$\mathcal S$}
  } }
  \put(-19,70)   {$ \pS\big((\GH\cir\M_{\wstc} {\setminus}\, \mathcal T ) \sqcup
                   \mathcal T \big) $}
  \put(30,42)    {$ =~ \pS(\GH\cir\M_{\wstc} {\setminus} \mathcal T)
                   \,\sqcup\, \pS(\mathcal T ) ~= $}
  }
Here the dotted circle is running along the non-contractible cycle of the solid torus.
Next we apply a modular $S$-transformation to the boundary $\partial \mathcal T$
of the cut-out solid torus (in terms of the ambient $\R^4$ this is afforded by
a homeomorphism homotopic to $(x,y,z,t) \,{\mapsto}\, (z,x,t,y)$).
After applying the projection $\pS$ this amounts to
  \eqpic{S_trans_proj}{340}{24}{
  \put(0,0)     {\Includepic{22ds}
  }
  \put(100,30)    {$\stackrel{\mathrm S}\longmapsto$}
  \put(153,0)     {\Includepic{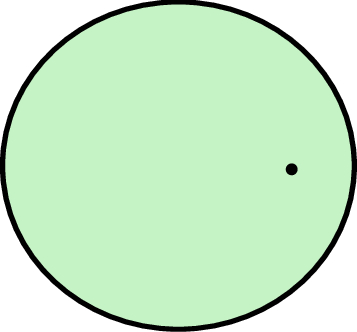}}
  \put(199,30)    {\footnotesize$P$}
  \put(245,30)    {$ \equiv $}
  \put(280,0)     {\Includepic{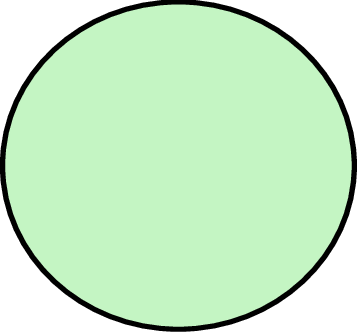}}
  }
Here the marked point $P$ in the middle picture is the image of the dotted
circle $\mathcal S$ under the S-transformation. As indicated by the
redrawing in the right-most picture, this is by no means a
distinguished point of the three-manifold. But when including the ribbon graph,
$P$ is still distinguished by the fact that a crucial piece of the
ribbon graph (an annular $A$-ribbon) is running along the circle $\pS^{-1}(P)$.
To complete the surgery, we glue back the S-transformed solid torus to
$\GH\cir\M_{\wstc} {\setminus}\, \mathcal T$. The result is
  \eqpic{gluingtocut_gluedback}{90}{39}{
  \put(0,0)     {\Includepic{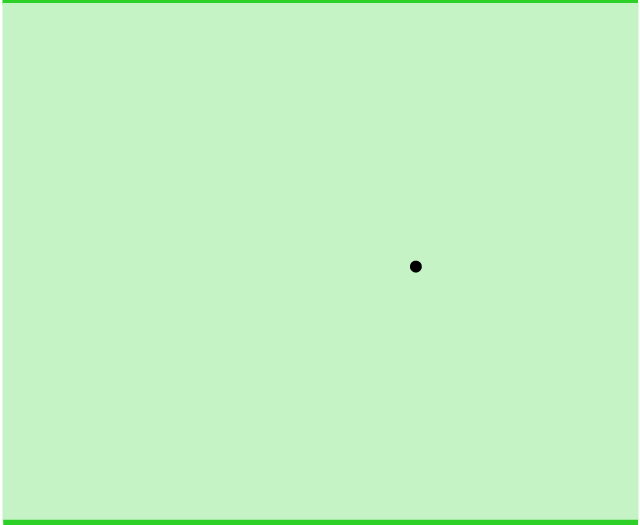}
  \put(67,46)    {\footnotesize$P$}
  } }
This manifold indeed coincides, up to a ribbon graph, with $\M_{\wst}$.

Details of the surgery relation, including ribbon graphs, are given in
proposition \ref{surg_prop}.
The annular $A$-ribbon along $\pS^{-1}(P)$ in $M_{\wst}$ is shown in
the picture \eqref{deftor} below (compare also picture (5.3) in \cite{fjfrs}).

%%%%%%%%%%%%%%%%%%%%%%%%%%%%%%%%%%%%%%%%%%%%%%%%%%%%%%%%%%%%%%%%%%%%%%%%

\subsection{The factorization identity}\label{FactorizationIdentity}

Performing the steps described in the previous section, we arrive at the following result.
We consider the situation that an oriented world sheet $\ws$ is factorized
as in \eqref{dis_sphere} into a world sheet $\wsf pq\alpha\beta$, by cutting
along a circle that is crossed by the topological defect line $X$.

\begin{theorem}\label{thm:bulkfac}
The correlator $C(\ws)$ for an oriented world sheet $\ws$ can be expressed in
terms of the correlators $C(\wsf pq\alpha\beta)$ of the factorized world sheets as
  \be\label{fact_rel}
  C(\ws) = \sum_{p,q\in\I}\;\sum_{\alpha,\beta}\, \dim(U_p)\,\dim(U_q)\;
  \cdefinv {X}\pb\qb\alpha\beta\; \GLL pq(C(\wsf pq\alpha\beta)) \,.
  \ee
\end{theorem}

Establishing the formula \eqref{fact_rel} will occupy the next subsection. Before
entering the proof, let us explain those parts of the notation appearing in this statement
that we have not yet used, as well as give some further background information (for more
details see e.g.\ Section 3 of \cite{fuRs10} and Appendices A, B and C of \cite{fjfrs}):
The (strictified) representation category \cC\ of the chiral symmetry algebra
is a modular tensor category. As such it has, up to isomorphism, a finite number
of simple objects, which we denote by $U_p$ with $p$ taking values in a finite
index set $\I$; $U_0\eq\tunit$ is the tensor unit of \cC. On the index set $\I$ there
is an involution $p\,{\mapsto}\,\bar p$ such that the simple object $U_{\bar p}$
is isomorphic to the dual $U_p^{\,\vee}$ of $U_p$. The number $\dim(U_p)$ is the
(quantum) dimension of the object $U_p$.

The labels $\alpha$ and $\beta$ in \eqref{fact_rel} are elements of a basis of
the space of disorder fields with chiral labels $p,q$ and $\bar p,\bar q$,
respectively, or in more mathematical terms, of the morphism space
$\Homaa {U_p\,\oti^+A\,\oti^-U_q}{X}$ and $\Homaa {U_\pb\,\oti^+X\oti^-U_\qb}A$,
respectively. Here the object $A$ of \cC\ is the simple symmetric special
\emph{Frobenius algebra} that together with the chiral data characterizes
\cite{fuRs4} the full CFT in the regions of \ws\ adjacent to $X$. (A Frobenius
algebra $A$ in \cC\ is an object $A$
together with a multiplication morphism $m$, unit $\eta$, comultiplication
$\Delta$ and counit $\eps$ such that $(A,m,\eta)$ is a unital associative algebra,
$(A,\Delta,\eps)$ is a counital coassociative coalgebra, and $\Delta$ is a
morphism of $A$-bimodules. For more details, and for the meaning of the
qualifications symmetric and special, see e.g.\ \cite[Sect.\,3.3]{fuRs4}.)

The label $X$ of the defect line is an object of \cC\ that carries the structure
of an $A$-bimodule. By $U_i\,\oti^+X\oti^-U_j$ we refer to an
$A$-bimodule whose underlying object in \cC\ is $U_i\,\oti\, X\oti\, U_j$ and
whose bimodule structure is defined by combining the bimodule structure of $X$
with the braiding of \cC\ (via braided induction, as summarized e.g.\ in
\cite[Sect.\,2.2]{fuRs12}). When the two regions adjacent to a defect line
are labeled by two different Frobenius algebras $A$ and $B$, one deals with
$A$-$B$-bimodules. For any two $A$-$B$-bimodules $X$ and $Y$, $\Homab XY$
is the space of $A$-$B$-bimodule morphisms from $X$ to $Y$, i.e.\ the subspace
of the morphism space $\Hom(X,Y)$ of \cC\ consisting of those morphisms that
commute with the left (right) action of the algebra $A$ ($B$).

Finally, $c^\text{def}_{X,pq}$ is the matrix whose entries are the structure
constants -- that is, the coefficients in a standard basis of conformal blocks --
of the \emph{defect two-point function}. By the latter we mean the correlator
$C(\wsSD Xpq\alpha\beta)$ for the two-sphere with two insertions of disorder
fields $\Theta$, labeled by $\alpha$ and $\beta$, and with a defect line $X$
connecting them. The world sheet for this correlator looks as follows:
  \eqpic{WS_twofieldsdefect}{80}{43}{ \setlength\unitlength{1.2pt}
  \put(0,4)     {\INcludepic{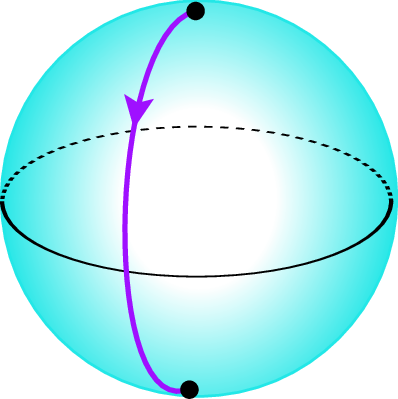}
  \put(30,-9)    {\footnotesize$\Thet X\pb\qb\beta$}
  \put(30,78)    {\footnotesize$\Thet Xpq\alpha$}
  \put(30,56)    {\footnotesize$X$}
  } }
For details about the matrix entries $\cdef {X}pq\alpha\beta$ we refer to Appendix
\ref{app:cdef}.

\begin{remark}
In \cite{fjfrs} a different convention for the gluing cobordism was used,
based on a different choice of basis morphisms in the spaces
$\Hom(\one,U_\ia \oti U_\ib)$. The relation between the relevant basis
morphisms will be explained in \eqref{lambdachoice} below.
With the choice made in \cite{fjfrs}, \eqref{fact_rel} gets replaced by
  \be\label{fact_FFRS}
  C(\ws)^{\scriptscriptstyle\rm FFRS}
  = \sum_{p,q\in\I}\sum_{\alpha,\beta} \theta_p^{}\,\theta_q^{-1}\,
  \RR \pb p 0\nl\nl\, \Rm \qb q 0 \nl\nl\,\dim(U_p)\,\dim(U_q)\;
  \cdefinv {X}\pb\qb\alpha\beta\; \GLL pq(C(\wsf pq\alpha\beta)) \,.
  \ee
Here $\theta_p \eq \exp(-2\pi \mathrm i\Delta_p)$, with
$\Delta_p$ the conformal weight of the primary field associated to $U_p$,
is the eigenvalue of the twist automorphism of $U_p$, while  $\RR \pb p 0 \nl\nl$
and $\Rm \pb p 0 \nl\nl$ are the braiding matrices in the chosen bases of
$\Hom(U_\pb\,\oti\, U_p,\one)$ and $\Hom(U_p\,\oti\, U_\pb,\one)$ (see
\cite[Sect.\,2.1\,\&\,2.2]{fjfrs} for further details).
\end{remark}

%%%%%%%%%%%%%%%%%%%%%%%%%%%%%%%%%%%%%%%%%%%%%%%%%%%%%%%%%%%%%%%%%%%%%%%%

\subsection{Proof of bulk factorization}\label{sec:factorizationproof}

In this subsection we prove the bulk factorization identity \eqref{fact_rel}.
In the proof we will freely use the
graphical calculus for ribbon categories, analogously as has been done in \cite{fjfrs}.
We will also need the following constructions with bimodules. First,
to any $A$-$B$-bimodule $X$, with $A$ and $B$ algebras in \cC, there is associated
a \emph{dual} or conjugate bimodule $X\ve$, which
allows us to describe the orientation reversal of defect lines. $X\ve$
is a $B$-$A$-bimodule, such that
${(X\ve)}\ve \,{\cong}\, X$; the actions of $A$ and $B$ on $X\ve$ are obtained from
those on $X$ with the help of the duality morphisms of \cC\ (for details see
formula (2.37) of \cite{fuRs8}).
Second, to any pair consisting of an $A$-$B$-bimodule $X$ and a $B$-$C$-bimodule $Y$,
there is associated their \emph{tensor product} $X\,{\otimes_B}\,Y$ over $B$,
which is an $A$-$C$-bimodule.
In terms of defect lines, this means that when two defect lines labeled by $X$ and
$Y$ are fused, their fusion product is labeled by $X\,{\otimes_B}\,Y$.
Analogously there are tensor products with any number of factors,
$X_1\,\oti_{\!A_1}\cdots\,\oti_{\!A_{m-1}}\, X_m$,
where for any $i\iN\{1,2,...\,,m\}$, $X_i$ is an $A_{i-1}$-$A_i$-bimodule.
Some details about this notion of tensor product are collected in Appendix
\ref{app:fusion}.
Via the multiplication morphism, any algebra is naturally a bimodule
over itself. We are only interested in algebras that are simple, i.e.\ simple
as bimodules. For any $A$-$B$-bimodule $X$ with simple algebras $A$ and $B$
there are natural bimodule isomorphisms $A\,{\otimes_A}\,X \,{\cong}\, X {\cong}\,
X\,{\otimes_B}\,B$; for our purposes we can take these isomorphisms to be equalities.

We will also need to express morphism spaces involving tensor products over
algebras as subspaces of morphisms that involve ordinary tensor products in $\cC$.
Let $A_i$ and $B_j$ be symmetric
special Frobenius algebras, with $i\eq 0,1,2,...\,,m$ and $j\eq 0,1,2,...\,,n$,
respectively. Consider any collection of $A_{i-1}$-$A_i$-bimodules
$X_i$ ($i\eq 1,2,...\,,m$) and $B_{j-1}$-$B_j$-bimodules $Y_j$ ($j\eq 1,2,...\,,n$)
such that $B_0 \eq A_0$ and $B_n \eq A_m$. We define the subspace
  \be
  \HomP_{A_0|A_m}(X_1\,\oti\cdots\oti\, X_m,Y_1\,\oti\cdots\oti\, Y_n) \subseteq
  \Hom_{A_0|A_m}(X_1\,\oti\cdots\,\oti X_m,Y_1\,\oti\cdots\oti\, Y_n)
  \ee
to be the space of morphisms $f\iN\Hom_{A_0|A_m}(X_1\,\oti\cdots
          $\linebreak[0]$
\oti\, X_m,Y_1\,\oti\cdots\oti\, Y_n)$ that satisfy
  \be
  f\circ\PA{X_1}{X_m} = f = \PA{Y_1}{Y_n}\circ f\,,
  \ee
where for $m\,{\ge}\,2$ the morphism $\PA{X_1}{X_m}$ is the idempotent
\eqref{bnd_idem} whose image equals (as explained Appendix \ref{app:fusion})
the bimodule tensor product, while $P_{X_1} \,{=}\, \id_{X_1}$. Associated
with the idempotent $\PA{X_1}{X_m}$ there are embedding and restriction
morphisms $e$ and $r$ which satisfy $e \cir r \eq \PA{X_1}{X_m}$ and
$r \cir e \eq \id_{X_1\otimes_{A_1}\cdots\otimes_{A_{m-1}}X_m}$. These
provide an isomorphism
  \be
  \bearl
  \HomP_{A_0|A_m}
  (X_1\,\oti\,X_2\,\oti\cdots\oti\, X_m , Y_1 \,\oti\, Y_2 \,\oti \cdots \oti\, Y_n)
  \\[-.6em]\\~ \hspace*{3.5em}
  \,\cong\,
  \Hom_{A_0|A_m}
  (X_1\,{\otimes_{A_1}}\,X_2\,{\otimes_{A_2}}\, \cdots \,{\otimes_{A_{m-1}}}\,X_n ,
  Y_1\,{\otimes_{B_1}}Y_2\,{\otimes_{B_2}}\, \cdots \,{\otimes_{B_{n-1}}}Y_n)
  \,.  \eear
  \ee
for any $m,n \iN \mathbb N$.

To enter the proof of Theorem \ref{thm:bulkfac}, we introduce
bases of the relevant morphism spaces:
For any $p,q\iN \I$ and any two $A$-bimodules $X$ and $Y$, choose a basis
  \be\label{basis_DF}
  \{\, \phi^\alpha_{pq} \,|\, p,q\iN\I,\,\alpha\iN\Homaa{U_p\,\oti^+\!X\oti^-U_q}Y \,\}
  \ee
of bimodule morphisms. There then exists a basis
  \be\label{basis_DFd}
  \{\, \bar\phi^\alpha_{pq} \,|\, p,q\iN\I,\,\alpha\iN\Homaa Y{U_p\,\oti^+\!X\oti^-U_q} \,\}
  \ee
that is dual to the basis \eqref{basis_DF} in the sense that
  \be\label{DF_dual}
  \mathrm{Tr}\big(\phi^\alpha_{pq}\cir\bar\phi^\beta_{pq}\big)
  = \delta_{\alpha,\beta}\;\dim(Y)\,.
  \ee
Such dual bases exist due to the presence of a non-degenerate pairing of the spaces
$\Homaa{U_p
   $\linebreak[0]$
\,\oti^+X\oti^-U_q}Y$ and $\Homaa Y{U_p\,\oti^+X\oti^-U_q}$. The existence of such
a pairing, in turn, follows by arguments analogous to those in the proof of
Lemma C.3 of \cite{fjfrs}. We will use the abbreviations $\alpha \,{\equiv}\,
\phi^\alpha_{pq}$ and $\bar\alpha \,{\equiv}\, \bar\phi^\alpha_{pq}$ whenever
the suppressed labels $p,q$ can be directly inferred from the context.

As we pointed out at the end of section \ref{BF_outline}, the connecting manifold
$\M_\ws$ of the original world sheet $\ws$ and the underlying cobordism of
$\GLL pq( C(\wsf pq\alpha\beta))$, the gluing homomorphism applied to the
correlator of the factorized world sheet, differ only in the cobordisms $\TORS$ and
$\TOR \eq \TOT$ (see their descriptions in \eqref{MST-TORS} and \eqref{MST-TOR}).
For the correlators this means that
  \be\label{C-tor}
  C(\ws) = Z(\MST) \circ Z(\TORS)
  \ee
and
  \be\label{Cfact-tor}
  \GLL pq(C(\wsf pq\alpha\beta)) = Z(\MGL) \circ Z(\M_{\ws'})
  = Z(\MST) \circ Z(\TOT)\,.
  \ee
Comparison with \eqref{fact_rel} thus shows that the bulk factorization
identity amounts to a relation between the invariants
$Z(\TORS)$ and $Z(\TOT)$. The cobordism $\TORS$ is given by
  \eqpic{deftor}{240}{103}{
  \put(17,109) {$\TORS ~=$}
  \put(75,0){ \Includepic{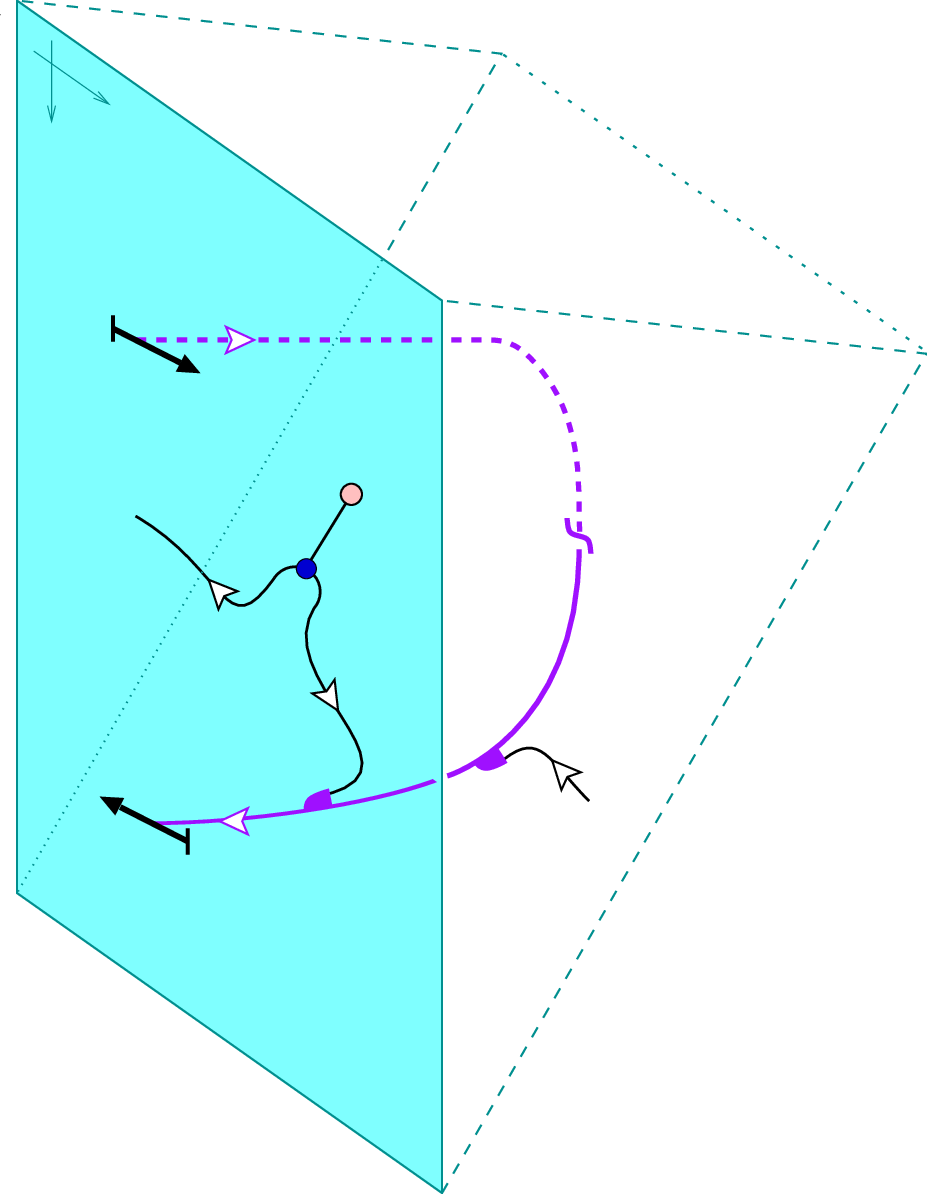}
  \put(18,197)    {\tiny $1$}
  \put(4,192)     {\tiny $2$}
  \put(11,150)    {\begin{turn}{-27}\scriptsize $(X,+)$\end{turn}}
  \put(10,64)     {\begin{turn}{-27}\scriptsize $(X,-)$\end{turn}}
  \put(58,98)     {\scriptsize $A$}
  \put(101,90)    {\scriptsize $X$}
  } }
while $\TOT$ looks as
  \eqpic{instor}{250}{102}{
  \put(0,110) {$\TOT ~=$}
  \put(80,2){ \Includepic{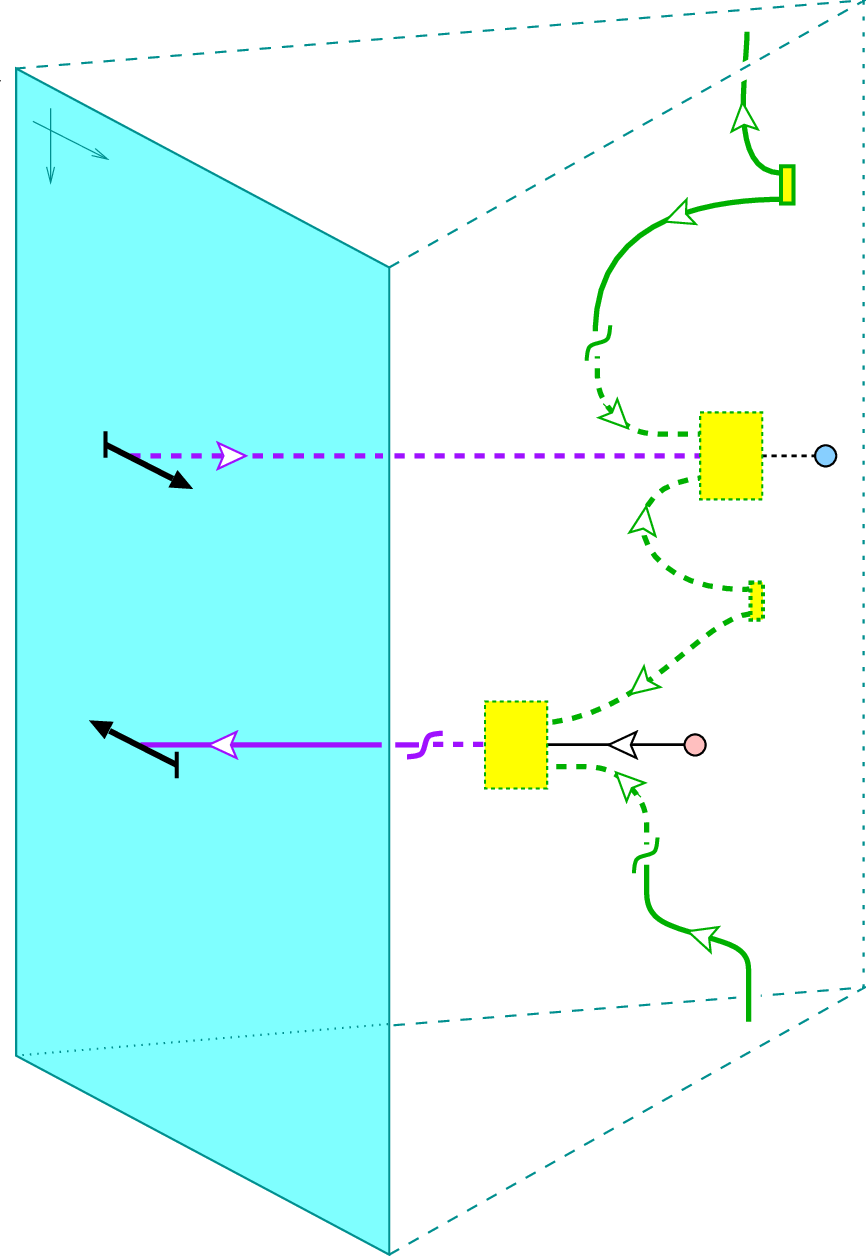}
  \put(19,199)    {\tiny $1$}
  \put(5,194)     {\tiny $2$}
  \put(140,210)   {\scriptsize $q$}
  \put(131,186)   {\scriptsize $\qb$}
  \put(120,80)    {\scriptsize $q$}
  \put(109.7,128) {\scriptsize $\pb$}
  \put(119,111)   {\scriptsize $p$}
  \put(91.2,91.6) {\small $\alpha$}
  \put(130.6,144) {\small $\beta$}
  \put(11,141)    {\begin{turn}{-27}\scriptsize $(X,+)$\end{turn}}
  \put(10,90)     {\begin{turn}{-27}\scriptsize $(X,-)$\end{turn}}
  \put(147,190)   {\tiny \begin{turn}{90}$\bar{\lambda}^{\qb q}$\end{turn}}
  \put(143,115)   {\tiny \begin{turn}{90}$\bar{\lambda}^{\pb p}$\end{turn}}
  } }
Here $\alpha$ and $\beta$ label a basis of $\Homaa{U_p\,\oti^+A\,\oti^-\,U_q}X$
and of $\Homaa{U_\pb\,\oti^+X\oti^-\,U_\qb}A$, respectively. The pictures
\eqref{deftor} and \eqref{instor} are drawn in the \emph{wedge presentation},
which means that the ``white'' faces of a wedge are to be identified in such
a manner that the shaded rectangle becomes a torus, with the tip of the
wedge -- that is, the horizontal dotted line in \eqref{deftor}, and the vertical
dotted line in \eqref{instor}, respectively -- describing a non-contractible
circle (for more details see section 5.1 of \cite{fjfrs}).

The definition of the gluing homomorphism in \cite{fjfrs} rests on the choice of
a basis $\bar\lambda^{k\kb}$ of the one-dimensional space $\Hom(\one,U_k\oti U_\kb)$
for each $k\iN\I$. Here we work with a slightly different choice of basis
than in \cite{fjfrs}; instead of $\bar\lambda^{k\kb}$ we use the morphism
  \eqpic{lambdachoice}{240}{19}{ \put(0,-9){
  \put(15,0){ \Includepic{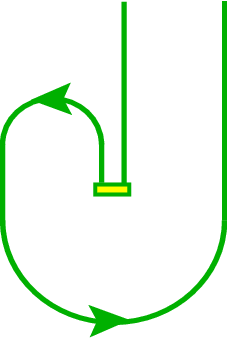}
  \put(21,66)        {\scriptsize $k $}
  \put(39,66)        {\scriptsize $\kb $}
  \put(14,17)        {\scriptsize $\bar\lambda^{\kb k}$}
  }
  \put(80,34) {$=~ \theta_k\;\RR {\kb} k0\nl\nl \,\bar\lambda^{k\kb}
               ~=~ \theta_k^{-1}\,\Rm {\kb} k0\nl\nl \,\bar\lambda^{k\kb} $}
  } }
in the space $\Hom(\one,U_k\oti U_\kb)$. As a consequence the gluing homomorphism
used by us differs from the one in \cite{fjfrs} by a factor
$\theta_p\,\RR {\pb} p0\nl\nl \,\theta_q^{-1}\,\Rm {\qb} q0\nl\nl$. This is
the origin of the additional factors in the expression \eqref{fact_rel} for
$C(\ws)^{\scriptscriptstyle\rm FFRS}$ as compared to $C(\ws)$ \eqref{fact_FFRS}.

The cobordisms \eqref{deftor} and \eqref{instor} are obtained in a manner
completely analogous to the derivation of the corresponding results (5.3) and (5.9)
of \cite{fjfrs}. Thus we refrain from giving any details of this derivation. The
precise form of the relation between the invariants of $\TORS$ and $\TOT$ is given
in the following statement:

\begin{proposition}\label{surg_prop}
The invariant $Z(\TORS )$ can be expanded as
  \be\label{surg_DF}
  Z(\TORS)=\sum_{p,q\in\I}\,\sum_{\alpha\in H^{A,X}_{p,q}}
  \sum_{\beta\in H^{X,A}_{\pb,\qb}}\, C_{pq,\alpha\beta}\; Z(\TOT) \,,
  \ee
where $H^{Y,Y'}_{r,s} \,{:=}\, \Homaa {U_r\oti^+Y\oti^-U_s}{Y'}$
and the coefficients $C_{pq,\alpha\beta}$ are given by
  \be\label{coeff_value}
  C_{pq,\alpha\beta} =
  \dim(U_p)\;\dim(U_q)\;\cdefinv {X}\pb\qb\alpha\beta \,.
  \ee
\end{proposition}

\proof
Consider the space $\im(\proj)$, where $\proj \eq Z(\projMF)$ is the projector obtained
as the invariant of the cobordism
  \eqpic{CCinv}{240}{99}{
  \put(15,110)  {$\projMF~:=$}
  \put(30,0){ \Includepic{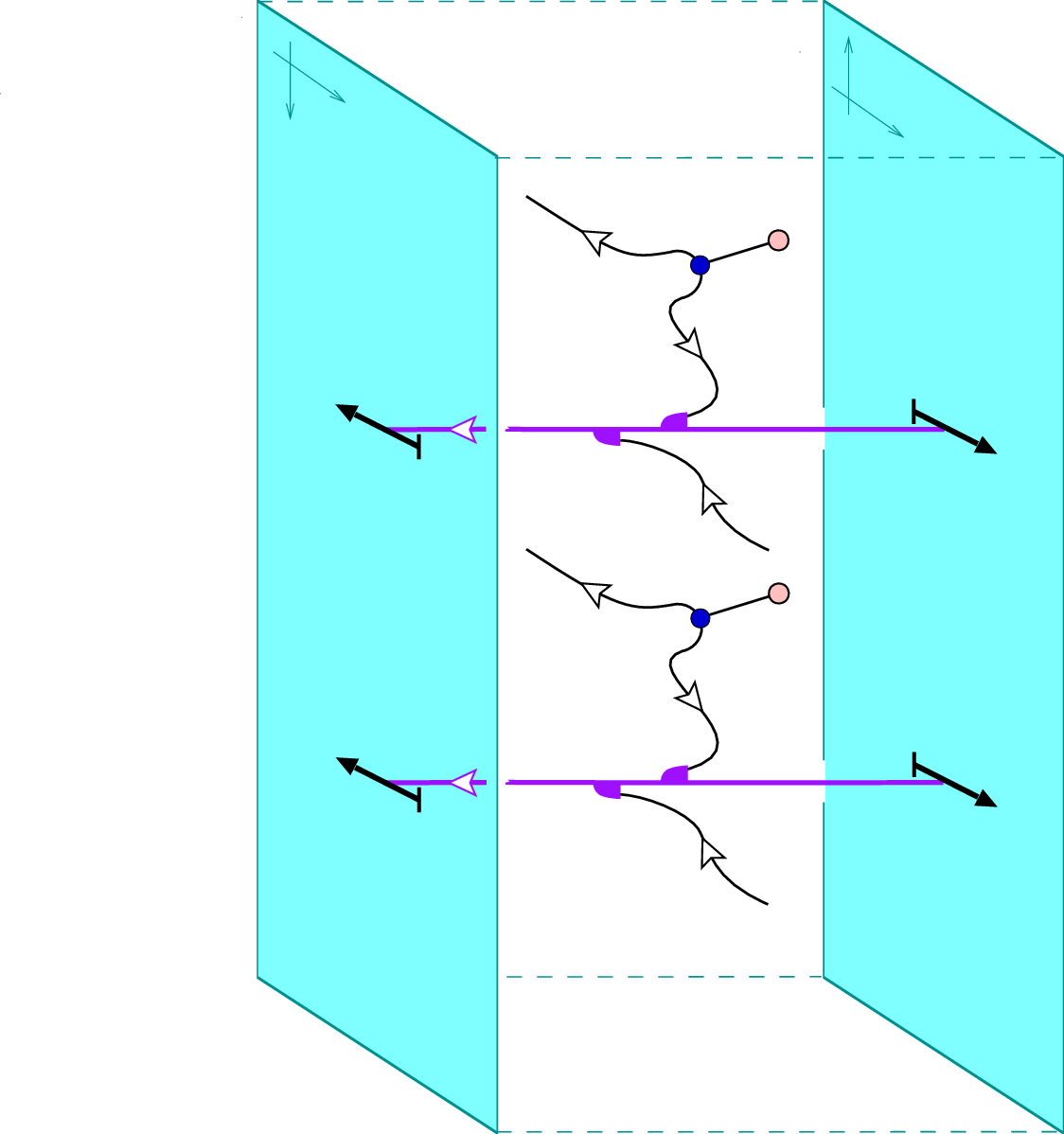}
  \put(65,195)    {\tiny $1$}
  \put(51,192)    {\tiny $2$}
  \put(59,66)     {\begin{turn}{-27}\scriptsize $(X,-)$\end{turn}}
  \put(139,74)    {\scriptsize $A$}
  \put(140,50)    {\scriptsize $A$}
  \put(175,77)    {\begin{turn}{-27}\scriptsize $(X,+)$\end{turn}}
  \put(59,134)    {\begin{turn}{-27}\scriptsize $(X,+)$\end{turn}}
  \put(139,145)   {\scriptsize $A$}
  \put(140,119)   {\scriptsize $A$}
  \put(175,145)   {\begin{turn}{-27}\scriptsize $(X,-)$\end{turn}}
  \put(174,195)   {\tiny $1$}
  \put(166,208)   {\tiny $2$}
  }}
Here top and bottom as well as front and back are to be identified, i.e.\
one deals with a ribbon graph in the three-manifold $\ExtTdef X\Times[0,1]$,
where $\ExtTdef X \eq \partial\TORS \eq \partial\TOT$ is a torus with two marked
points labeled by $(X,+)$ and $(X,-)$. This projector is an obvious
generalization of the projector used in \cite[Eq.\,(5.15)]{fjfrs}. That $\proj$
is a projector is seen in the same way as in Lemma 5.4(i) of \cite{fjfrs}, and
in complete analogy with the proof of Lemma 4.5(i) of \cite{ffrs5} one proves
that $Z(\TORS)\iN\im(\proj)$. Further, the set
  \be\label{basis_Im(P)}
  \mathcal B := \{ Z(\TOT) \}_{p,q,\alpha,\beta}^{}
  \ee
of vectors constitutes a basis of $\im(\proj)$, as can be seen in the same way as
in the proof of Lemma 5.4(ii) of \cite{fjfrs}.
This already shows that an expansion of the form \eqref{surg_DF} indeed exists.
\\
The values of the coefficients in that expansion are obtained by composing
\eqref{surg_DF} with an element of the basis $\mathcal B^*$ dual to $\mathcal B$.
Using the explicit form of the basis $\mathcal B^*$, which is obtained in
Appendix \ref{app:B*}, we find
  \eqpic{coeff-1}{370}{79}{
  \put(-4,90) {$S_{0,0}^{-1}\;\mathcal N^{-1}\, C_{pq,\alpha\beta}~= $}
  \put(106,2){ \Includepic{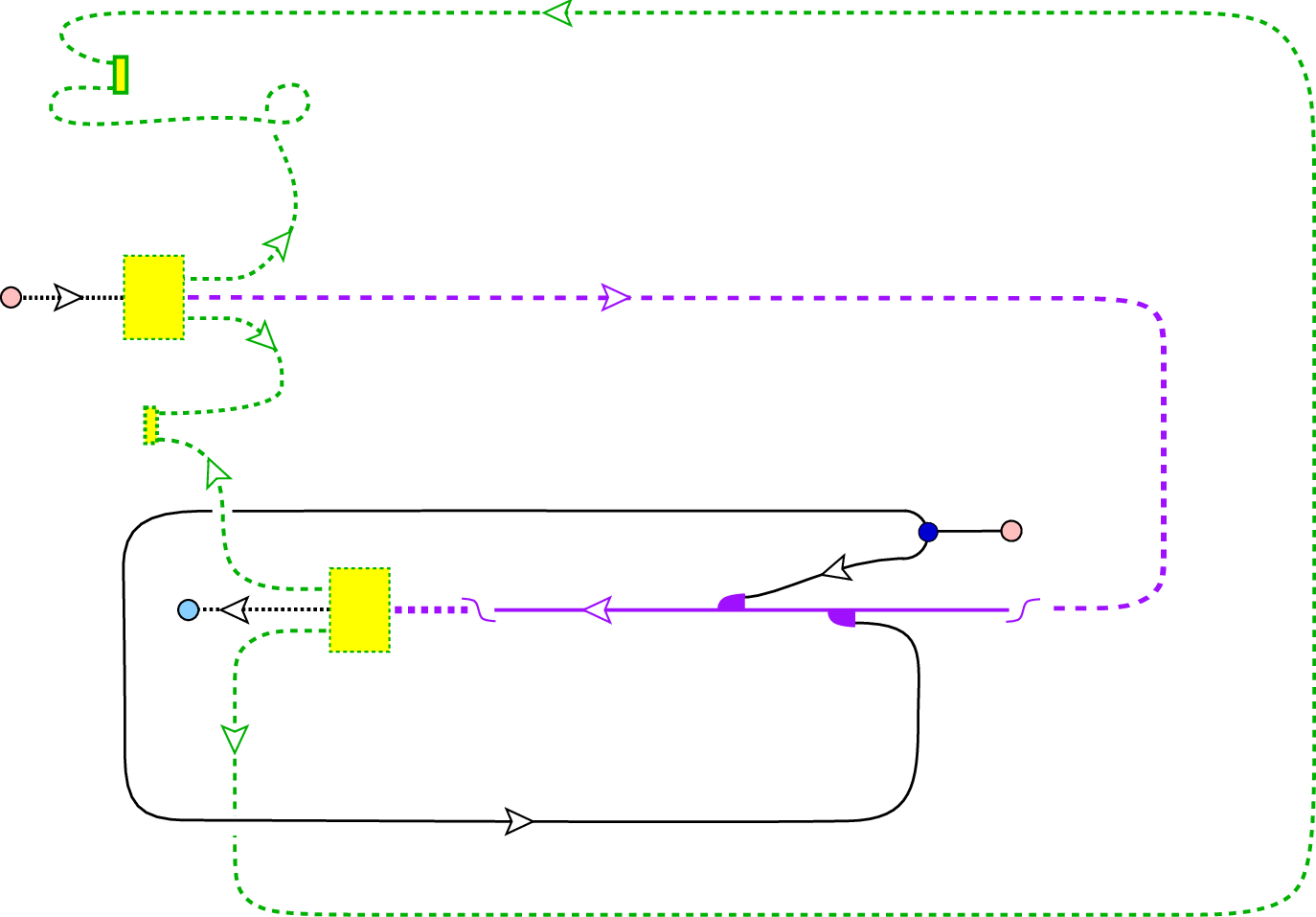}
  \put(56.5,102)   {\scriptsize $\pb $}
  \put(45.2,82)    {\scriptsize $p$}
  \put(59,134)     {\scriptsize $\qb $}
  \put(125,165)    {\scriptsize $q$}
  \put(48.2,28)    {\scriptsize $q$}
  \put(65.3,57)    {\small $\bar{\alpha}$}
  \put(26.1,115)   {\small $\bar{\beta}$}
  \put(115,80)     {\scriptsize $A$}
  \put(155,122)    {\scriptsize $X$}
  } }
Here the normalization factor $\mathcal N$ is given by
  \be
  \mathcal N = \frac{(\dim(U_p)\, \dim(U_q))^2_{}} {\dim(A)\, \dim(X)} \,,
  \ee
while $\bar\alpha$ and $\bar\beta$ label elements of the basis of
$\Homaa X{U_p\,\oti^+\!A\,\oti^-U_q}\}$ and of $\Homaa A{U_\pb\,\oti^+
  $\linebreak[0]$
X\oti^-U_\qb}\}$,
respectively. These bases are dual, in the sense of \eqref{DF_dual}, to the ones
appearing in \eqref{instor}. Using that $\bar\alpha$ is a morphism of bimodules and
that $A$ is symmetric special Frobenius, we can rewrite \eqref{coeff-1} as
  \eqpic{coeff-2}{170}{41}{ \put(0,-2){
  \put(0,55) {$C_{pq,\alpha\beta}~=~ S_{0,0}\;\mathcal N$}
  \put(98,0){ \Includepic{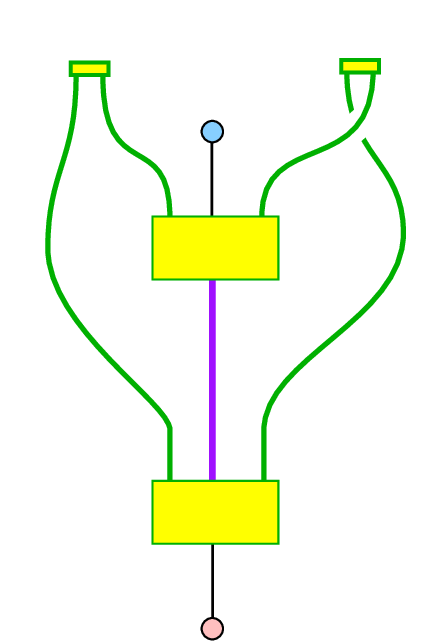}
  \put(21,82)        {\scriptsize $p $}
  \put(23,34)        {\scriptsize $\pb$}
  \put(52,82)        {\scriptsize $q $}
  \put(50.5,34)      {\scriptsize $\qb$}
  \put(35.5,19.8)    {\small $\bar{\beta}$}
  \put(36,69.3)      {\small $\bar{\alpha}$}
  \put(40.8,53)      {\scriptsize $X$}
  } } }
As shown in Lemma \ref{coeff_calc}, the number on the right hand side equals
$\dim(U_p)\,\dim(U_q)\, \cdefinv {X}pq\alpha\beta$. This finally establishes the 
expression \eqref{coeff_value} for the coefficients in the expansion \eqref{surg_DF}.
\endofproof

\noindent{{\it Proof of Theorem} \ref{thm:bulkfac}.}\hspace{7pt}
~\\
Combining the expressions \eqref{C-tor} and \eqref{Cfact-tor} for the correlators
on the original and cut world sheets with the expansion \eqref{surg_DF} results in
the factorization formula \eqref{fact_rel}.
\endofproof

\begin{remark}\label{mistake_lemma_coeff_calc}
In the case of the trivial, or \emph{transparent},
defect $X \eq A$ we have
  \be
  \bearll
  \theta_p^{}\,\theta_q^{-1}\, \RR \pb p 0 \nl\nl\, \Rm \qb q 0 \nl\nl \,
  \cdefinv {A}\pb\qb\alpha\beta \!\!\! &
  \equiv\theta_p\,\theta_q^{-1}\,\RR \pb p 0 \nl\nl\, \Rm \qb q 0 \nl\nl \,
  (c^{\text{bulk}-1}_{A,\pb,\qb})_{\alpha\beta}^{}
  \\{}\\[-.7em]&
  = \RR\pb p0 \nl\nl\, \Rm\qb q0 \nl\nl\, (c^{\text{bulk}-1}_{A,\pb,\qb})_{\alpha\beta}^{}
  = (c^{\text{bulk}-1}_{A,p,q})_{\beta\alpha}^{} \,.
  \eear
  \ee
Thus this special case of \eqref{fact_FFRS} reproduces
the result given in Theorem 2.13 of \cite{fjfrs}.
\end{remark}

%%%%%%%%%%%%%%%%%%%%%%%%%%%%%%%%%%%%%%%%%%%%%%%%%%%%%%%%%%%%%%%%%%%%%%%%

\section{Characterizing the world sheet}\label{Worldsheets}

It is worth noting that up to this point there was no need of a
detailed specification of all aspects of a world sheet. In contrast,
to proceed to the discussion of fundamental world sheets
we need a precise definition of what is to be meant by a world sheet.
Actually, several ingredients will be used only implicitly below, for instance
for comparison with the literature, but for the sake of completeness we present
them nevertheless.

Briefly, a world sheet involves two types of data: First, its description as
a two-dimensional \emph{manifold}. What additional structure this manifold should
carry depends on the application one has in mind \cite{scfu4}.
For our purposes, it is sufficient to regard it just as an oriented
topological manifold, in particular there is no need to
think of it as being endowed with a conformal structure.
The second piece of data encodes information about \emph{field insertions,
boundary conditions} and \emph{defect lines}. There is a lot of freedom in presenting
the latter information. In the description below we closely follow some of the
conventions in \cite{fuRs10,ffrs5}. For instance, we reserve the term
\emph{defect field} to insertions with precisely one incoming and one outgoing
defect line. Furthermore, a \emph{disorder field} is a defect field for which
one of the incident defect lines is trivial.

We denote by \cC\ the modular tensor category characterizing the chiral CFT, with
conventions as listed in Section \ref{FactorizationIdentity}.
Besides aspects of \cC\ and of symmetric special Frobenius algebras
$A \eq (A,m,\eta,\Delta,\eps)$ and their (bi)modules in \cC\ that were already used
above, we will also make use of the isomorphism $\Fiso\colon A\To A^\vee$ given by
  \be
  \Fiso = ((\eps\circ m)\,\oti\,\id_{A^\vee}) \cir (\id_A \,\oti\, b_A)
  = (\id_{A^\vee}\,\oti\,(\eps\circ m)) \cir (\tilde b_A \,\oti\, \id_A)
  \label{Fiso} \ee
that canonically comes with the symmetric Frobenius algebra structure.

%%%%%%%%%%%%%%%%%%%%%%%%%%%%%%%%%%%%%%%%%%%%%%%%%%%%%%%%%%%%%%%%%%%%%%%%

\subsection{Definition}\label{sec:defws}

We are now ready to specify what we mean by a world sheet.

\begin{definition}\label{def:ws}~\\[2pt]
        (i)
A \emph{world sheet}
   $\ws \eq (\tws,\D,\assd,\P)$ consists of:
\def\leftmargini{1.57em}~\\[-1.63em]\begin{itemize}\addtolength{\itemsep}{-6pt}%
  \item[\nxT] an oriented compact two-dimensional topological manifold $\tws$,
with orientation of the boundary $\partial\tws$ induced by the orientation of $\tws$;
  \item[\nxT] the \emph{defect graph} $\D$,
a finite oriented topological graph without univalent vertices, together with
a selection of a distinguished edge $e_v$ for each vertex $v$ of $\D$;
  \item[\nxT]
an injective continuous map $\,\assd\colon \D \To \tws$
such that $\dD$ defines a cell decomposition of $\tws$;
  \item[\nxT]
a finite unordered subset $\P$ of the set of two-valent vertices $v$ of $\D$;
elements of $\assd(\P)$ are called \emph{insertion points}.
  \end{itemize}
(ii)
The individual pieces of these data are decorated
by objects and morphisms of \cC. These decoration data are given by:
\def\leftmargini{1.57em}~\\[-1.63em]\begin{itemize}\addtolength{\itemsep}{-6pt}%
  \item[\nxT]
an assignment $\assa\colon f\Mapsto A_f$ of a simple symmetric special Frobenius algebra
$A_f$ in \cC\ to each \emph{face} $f$, i.e.\ to each connected component
of $\tws\,{\setminus}\,\dD$;
  \item[\nxT]
an assignment $\assb\colon e\Mapsto X_e$ of an $A_f$-$A_{f'}$-bimodule $X_e$ in \cC\ to each
edge $e$ of $\D$ for which $\assd(e) \,{\cap}\, \partial\tws \eq \emptyset$,
where $f$ and $f'$ are the faces to the right and to the left of the edge $e$, respectively;
  \item[\nxT]
an assignment $\assc\colon e\Mapsto M_e$ of an $A_f$-module $M_e$ in \cC\
to each edge $e$ of $\D$ for which $\assd(e) \,{\subset}\, \partial\tws$,
where $f$ is the face adjacent to $e$;
  \item[\nxT]
an assignment $\assp$ of either 0, 1, or 2 objects of \cC\ and of a corresponding
module or bimodule morphism to each vertex $v$ of $\D$, according to
  \be
  \assp:\quad v \,\longmapsto\, \left\{ \bearll
  \varPhi = (U_v^{},U'_v,\phi_v^{}) ~
  &{\rm for}~ v\iN \P\,\text{ and }\,\assd(v)\,{\notin}\,\partial\tws \,,\\~\\[-9pt]
  \varPsi = (U_v^{},\psi_v^{}) ~
  &{\rm for}~ v\iN \P\,\text{ and }\,\assd(v)\iN\partial\tws \,, \\~\\[-9pt]
  \varkappa_v ~
  &{\rm for}~ v \,{\notin}\,\P\, \text{ and }\,\assd(v)\,{\notin}\,\partial\tws\,,  \\~\\[-9pt]
  \chi_v    ~
  &{\rm for}~ v \,{\notin}\,\P\, \text{ and }\,\assd(v)\iN\partial\tws\,,\eear\right.
  \ee
with objects $U_v^{},U'_v \iN \cC$, and (bi)module morphisms
  \be
  \bearl
  \varPhi_v
  \in \Hom_{A_f|A_{f'}}(U_v \,\oti^+ X_{e_v}^\pm \,\oti^-\, U'_v,X_{\overline{e_v}}^\mp) \,,
  \\~ \\[-8pt]
  \varPsi_v \in \Hom_{A_f}(M_{e_v} \,\oti\, U_v,M_{\overline{e_v}}) \,,
  \\~\\[-7pt]
  \varkappa_v \in \HomP_{A_{f_1}|A_{f_0}}(X_{e_v}\vv,X_{e_1}\vvi \,\oti\, X_{e_2}\vvi\,
  \oti \,\cdots\, \oti\, X_{e_{N-1}}\vvi) \,,
  \\~\\[-7pt]
  \chi_v \in \HomP_{A_{f_1}}(M_{e_v},X_{e_1}\vvi \,\oti\, X_{e_2}\vvi\, \oti \,\cdots\,
  \oti\, X_{e_{N-2}}\vvi \,\oti\, M_{\overline{e_v}}) \,.
  \eear
  \ee
Here we have introduced the following notation.
  \end{itemize}
  \def\leftmargini{2.85em}~\\[-3.66em]\begin{itemize}\addtolength{\itemsep}{-6pt}%
  \item[\nxx]
For a bimodule $X^\pm \eq X$ we write $X^\mp \,{:=}\, X\ve$.
  \item[\nxx]
For any edge $e$ incident to a vertex $v$, the bimodule $X_e\vv$ is $X_e^{}$
if $e$ is an incoming edge, and $X_e^{\,\mathrm v}$ if $e$ is outgoing.
  \item[\nxx]
For $v\,{\not\in}\,\P$
a vertex with $N$ incident edges, those edges are totally ordered by obtaining a
cyclic ordering using $\assd$ and $\mathrm{or}(\tws)$ and declaring $e_v$ to be
the first edge. We label the edges as $e_0\,{\equiv}\, e_v,\, e_1,\, ...\, ,e_{N-1}$,
while the adjacent faces are denoted by $f_0,\, f_1,\,...\,,f_{N-1}$ in such a way
that $X_{e_i}\vvi$ is an $A_{f_{i}}$-$A_{f_{i+1}}$-bimodule.
  \item[\nxx]
For a vertex $v\iN\P$, $\overline{e_v}$ is the non-distinguished edge incident to $v$.
  \item[\nxx]
For a vertex with $\assd(v)\iN\partial\tws$ we choose the incoming boundary
edge to be the distinguished edge $e_v$ and denote the outgoing boundary
edge by $\overline{e_v}$.
 \end{itemize}
(iii) In addition the following restrictions are imposed:
 \def\leftmargini{1.63em}~\\[-1.45em]
 \begin{itemize}\addtolength{\itemsep}{-6pt}%
  \item[\nxT]
The orientation of any edge $e\iN\dD$ is opposite to that of the corresponding
defect line or boundary component.
  \item[\nxT]
For a vertex $v\,{\notin}\,\P$ for which at least one incident edge is decorated
by a Frobenius algebra $A_i$, the morphism $\varkappa_v$, respectively $\chi_v$, is
entirely composed of structure morphisms of the algebras $A_i$ and their (bi)modules
and of morphisms in spaces of the type $\HomP_{A_0|A_1}(X\vv_1,X\vvi_{2}\oti %\,
    $\linebreak[0]$
\cdots\,\oti X\vvi_{p})$
respectively $\HomP_{A}(M_{e_v}, X_{1}\vvi \,\oti\, X_{2}\vvi\, \oti \,\cdots\,
\oti X_{p}\vvi \oti M_{\overline{e_v}})$, with $X_i$, $i\eq 1,2,...\,,p$, the
non-trivial bimodules labeling edges incident to $v$, in a manner compatible
with the orientations of defect lines. Explicitly:
  \end{itemize}
  \def\leftmargini{2.85em}~\\[-3.68em]\begin{itemize}\addtolength{\itemsep}{-6pt}%
  \item[\nxx]
For a two-valent vertex with each edge decorated by $A$ or $A^\vee$, $\varkappa_v$
is given by $\id_A$, $\id_{A^\vee}$, $\Fiso$ or $\Fisoi$.
  \item[\nxx]
For an $n$-valent vertex with $n \,{\ge}\, 2$, $\assd(v)\,{\notin}\,\partial\tws$
and precisely one incoming edge decorated
by $A$ (or equivalently, precisely one outgoing edge decorated
by $A^\vee$),
and assuming, without loss of generality, the ordering of incident edges to be
such that $\assb(e_0) \eq A$, $\varkappa_v$ is required to be of the form
  \eqpic{Avertex}{150}{54}{
  \put(52,0)	{\Includepic{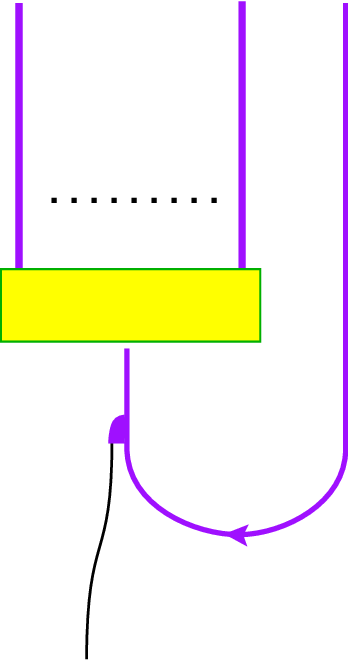}}
  \put(0,53)	{$\varkappa_v\ =$}
  \put(70,62)	{$\varkappa{}'$}
  \put(49,126)  {\scriptsize$X\vvi_{1}$}
  \put(82,126)  {\scriptsize$X\vvi_{n-2}$}
  \put(107,126) {\scriptsize$X\vvi_{n-1}$}
  \put(62.9,-9) {\scriptsize$A$}
  \put(77.7,44) {\scriptsize$\rho_{\!X_{n-1}\vv}^{}$}
   }
for some
$\varkappa'\iN\HomP_{A_|A_{n-1}}(X\vv_{n-1},X\vvi_{1}\oti\,\cdots\,\oti X\vvi_{n-2})$.
  \item[\nxx]
Similarly, for $n \,{\ge}\, 2$ and precisely one outgoing edge decorated
by $A$ (or equivalently, precisely one incoming edge decorated
by $A^\vee$),
$\varkappa_v$ is required to be like in \eqref{Avertex}, but with
the representation morphism $\rho$ replaced by $\rho\cir (\Fiso\,\oti\,\id)$.
The case $\assd(v)\iN\partial\tws$ works analogously.
  \item[\nxx]
The case that $n \,{\ge}\, 2$ and more than one edge incident to $v$ is decorated
by $A$
or $A^\vee$ is treated recursively, by first applying the previous prescription to
one choice of $A$- or $A^\vee$-labeled edge and then treating the auxiliary morphism
$\varkappa'$ in the same way as $\varkappa_v$.
 \end{itemize}
\end{definition}

\medskip

\begin{remark}
Our definition of world sheet differs somewhat from the one given in Definition B.2
of \cite{fjfrs}. The two descriptions are related as follows.
    \begin{enumerate}
    \item[(i)]
In \cite{fjfrs} no network of transparent defects is chosen. Instead, the definition
in \cite{fjfrs} involves the choice of a dual triangulation of the surface.
This dual triangulation is covered by defect lines that are labeled by Frobenius
algebras and by the structure morphisms of those algebras.
\\
The equivalences of world sheets to be discussed below can be used
to show that for any world sheet in the sense of our definition there is
an equivalent one for which $\dD$ contains such a dual triangulation.
    \item[(ii)]
In \cite{fjfrs}, one datum of a field insertion is the choice of a germ of
arcs containing the insertion point. In the present description, a possible
choice of germ is induced by the two edges
incident to the insertion point. We tacitly make this natural choice.
    \end{enumerate}
\end{remark}

\begin{remark} ~\\[2pt]
(i)
While the subset $\dD$ of the world sheet $\ws$ contains essential physical
information, the defect graph $\D$ itself is only an auxiliary datum. $\D$
has been included in the data for $\ws$ because various aspects of world
sheets can be formulated more conveniently by making reference to $\D$
rather than to $\dD$.
When doing so, for ease of notation we often refer to an edge $\assd(e)\iN\dD$ just as $e$.
\\[3pt]
(ii)
That $\dD$ defines a cell decomposition of $\tws$ implies in particular
that $\dD \,{\supset}\, \partial\tws$, that the interior of any edge of $\D$ is
either entirely mapped to the interior $\tws\,{\setminus}\,
\partial\tws$ or entirely to the boundary $\partial\tws$ of $\tws$, as well as
the following connectivity property: the pre-image of $\dD$ restricted to a
connected component of $\tws$ is a connected subgraph of $\D$.
\\[3pt]
(iii)
By considering world sheets for which every edge of the defect graph is labeled
by one and the same Frobenius algebra $A$ we obtain all orientable world
sheets considered in \cite{fjfrs}.
\\[3pt]
(iv)
At vertices with an algebra line attached we could in principle allow for more
general morphisms $\varphi$ than the appropriate representation morphism $\rho$
(see \eqref{Avertex}). However, by decomposing the relevant bimodule into its
simple summands $X_\mu$ and then choosing as bases for the one-dimensional
spaces $\HomP_{A|B}(A\,\oti\, X_\mu,X_\mu)$
and $\HomP_{A|B}(X_\mu,X_\mu\,\oti\, B)$ the left and right representation morphisms
of the $A$-$B$-bimodule $X_\mu$, any such morphism $\varphi$ is expressed in terms
of representation morphisms. By the properties
of the TFT, the correlator of a world sheet involving defect junctions
labeled by morphisms of the more general type is obtained as a sum over correlators
of world sheets containing only the types of morphisms allowed by our definition.
\\
For the same reason, in the case $n\eq2$ of \eqref{Avertex} one can without loss of
generality assume that $X^\mp_1 \,{\cong}\,A^\vee$ or $A$, implying that
$\kappa' \eq \eps_A$.
\\[3pt]
(v)
Elsewhere (e.g.\ in \cite[(3.31)]{fuRs10} or \cite{ruSu}) also field insertions
joining more than two defect lines are admitted.
In our description such a generalized field insertion is taken care of by the
concatenation of a vertex $u\iN\P$ and a multi-valent vertex $v\,{\notin}\,\P$.
Thus our definition does not impose any restriction in this respect.
\\[3pt]
(vi)
When constructing the connecting manifold $\M_\ws$, \emph{ribbons} and \emph{coupons}
are placed on edges and vertices, respectively, of $\dD$ with the same decoration
as $\D$. In order to match the conventions of \cite{fuRs10}, on each edge $\assd(e)
\iN\dD$ we place a ribbon whose core orientation coincides with the orientation
of $\assd(e)$, but whose 2-orientation is \emph{opposite} to the one of $\tws$.
\end{remark}

%%%%%%%%%%%%%%%%%%%%%%%%%%%%%%%%%%%%%%%%%%%%%%%%%%%%%%%%%%%%%%%%%%%%%%%%

\subsection{Equivalences of world sheets}\label{sec:eqWS}

World sheets that according Definition \ref{def:ws} are different can nevertheless
give the same correlator. We call two world sheets $\ws$ and $\ws'$ with the same
underlying surface $\tws$ \emph{equivalent} as world sheets iff
  \be
  C(\ws) = C(\ws') \,.
  \ee
In most cases that we will discuss, the equality of correlators becomes
tautological in the TFT construction and accordingly we refrain from spelling out
the corresponding proofs. Establishing the remaining equivalences
is not difficult either. As an illustration, the proof of ``independence
of transparent subgraph'' will be presented in Appendix \ref{app:iots}.

We first observe that the TFT construction immediately implies that $\ws$ and $\ws'$
give the same correlator if they differ only in such a way that the assignments
$\assd$ and $\assd'$
are related by an isotopy that fixes $\assd\big(\P\big) \eq \assd'\big(\P'\big)$
pointwise. We will freely employ such equivalences, and from now on do not distinguish
between world sheets whose embedded graphs are related by such an isotopy.

In the sequel we refer to elements of $\P$ as \emph{insertion vertices}, and to
vertices not contained in $\P$ as \emph{\fusion\ vertices}. Furthermore, we refer
to vertices labeled by structure morphisms of algebras or (bi)modules as
\emph{structure vertices}.
The orientation of the graph $\D$ supplies two functions $s$ and $t$ from the
set of edges of $\D$ to the set of vertices. We call $s(e)$ the \emph{source}
vertex and $t(e)$ the \emph{target} vertex of $e$.

There are also equivalences of world sheets that have different defect graphs $\D$
and $\D'$. Most of these need to be accompanied by a corresponding modification of
the decoration data. In the sequel we list fundamental equivalences of various types;
these will be used in section \ref{sec:fundws} to obtain a set of fundamental world sheets.

\medskip

In addition to these equivalences we may also consider world sheets related
in such a way that their correlators differ only by a non-zero
\emph{multiplicative factor}. Consider a face $f$ of $\ws$ labeled by $A$.
Let the world sheet $\ws'$ be obtained from $\ws$ by first inserting a circular
defect line, labeled by an $A$-$B$-bimodule $X$, in the face $f$ and then
fusing $X$ to the edges that bound the face $f$. By using equivalences from
the list below to express the correlator for $\ws'$ in terms of the original
one, it follows that
  \be
  C(\ws') = \frac{\dim(X)} {\dim(B)}\; C(\ws) \,.
  \ee
An interesting special case is when the bimodule $X$ is invertible, so that $X$ and
$X\ve$ give rise to a Morita context. The relation then also includes isomorphisms
of conformal field theories, as discussed in \cite{dakr}, compare also
\cite[Sect.\,3.3]{ffrs5}.

\medskip

In the remaining parts of this section we will consider pairs of world sheets
$\ws$ and $\ws'$ that have the same underlying surface $\tws$. We denote the
additional structure on $\tws$ by $\D$, $\assa$, $\assb$ etc.\ for $\ws$ and
by $\D'$, $\assa'$, $\assb'$ etc.\ for $\ws'$.  We provide a list of possible
ways in which this structure can differ for equivalent world sheets. In each
item below only those pieces of data are indicated that differ between $\ws$
and $\ws'$, while all remaining data coincide. When the differences in data are
confined to some specific region of $\tws$, we denote that region by $D$ and $D'$
for the world sheets $\ws$ and $\ws'$, respectively.

%%%%%%%%%%%%%%%%%%%%%%%%%%%%%%%%%%%%%%%%%%%%%%%%%%%%%%%%%%%%%%%%%%%%%%%%

\subsubsection{Equivalences involving only the assignments
               $\assa$, $\assb$, $\assc$ and $\assp$}

World sheets are equivalent if their assignments
$\assa$, $\assb$, $\assc$, and $\assp$ differ in one of the following ways:
   \def\leftmargini{1.57em}~\\[-1.45em]
\begin{itemize}\addtolength{\itemsep}{-6pt}%
  \item[\nxt]
\emph{Isomorphisms of defect lines and boundary conditions:}
For $e$ an edge of $\D$ and $\varphi$ an isomorphism of bimodules from $\assb'(e)$
to $\assb(e)$, or an isomorphism of modules from $\assc'(e)$ to $\assc(e)$,
$\assp'(s(e))$ and $\assp'(t(e))$ differ from $\assp(s(e))$ and $\assp(t(e))$
by composition with $\varphi$ and $\varphi^{-1}$, respectively (tensored with
identity morphisms for all other
edges incident to $s(e)$ and $t(e)$).
  \item[\nxt]
\emph{Orientation reversal of edges:}
Source and target of an edge $e$ of $\D$
with $\assd(e)\,{\not\in}\,\partial\tws$ differ according to
$s'(e) \eq t(e)$ and $t'(e) \eq s(e)$, while $\assb'(e) \eq \assb(e)\ve$.
  \item[\nxt]
\emph{Choice of distinguished edge:}
For $v\,{\not\in}\,\P$ such that $\assd(v) \,{\not\in}\, \partial\tws$
a vertex with $N$ incident edges, the distinguished edge
in $\ws$ is $e_v$, implying a total ordering $e_v,e_1,\ldots,e_{N-1}$ of edges
incident to $v$, while in $\ws'$ the distinguished edge is $e'_v=e_1$. This is
shown in the following figure.\,%
  \eqpic{rotvertex}{300}{40}{  \setlength\unitlength{1.2pt}
  \put(0,0)	{\INcludepic{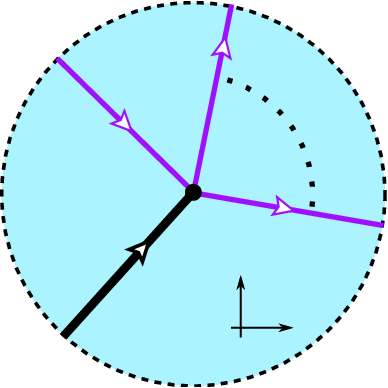}}
  \put(52,7)	{\tiny $1$}
  \put(46,20.5)	{\tiny $2$}
  \put(73,27)	{\footnotesize $e_1$}
  \put(42,72.8) {\footnotesize $e_{N-2}$}
  \put(-8,63)	{\footnotesize $e_{N-1}$}
  \put(3,5)	{\footnotesize $e_v$}
  \put(170,0)	{\INcludepic{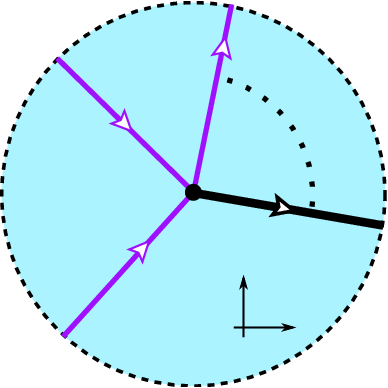}}
  \put(222,7)	{\tiny $1$}
  \put(216,20.5){\tiny $2$}
  \put(242,27)	{\footnotesize $e'_v$}
  \put(210,74)	{\footnotesize $e'_{N-3}$}
  \put(161,64)	{\footnotesize $e'_{N-2}$}
  \put(169,3)	{\footnotesize $e'_{N-1}$}
  \put(-28,70)	{\fbox{$D$}}
  \put(140,70)	{\fbox{$D'$}}
  }
The morphisms $\varkappa_v \eq \assp(v)$ and $\varkappa_v' \eq \assp'(v)$ are related by
  \eqpic{edgemove}{230}{63}{
  \put(0,10)	{\Includepic{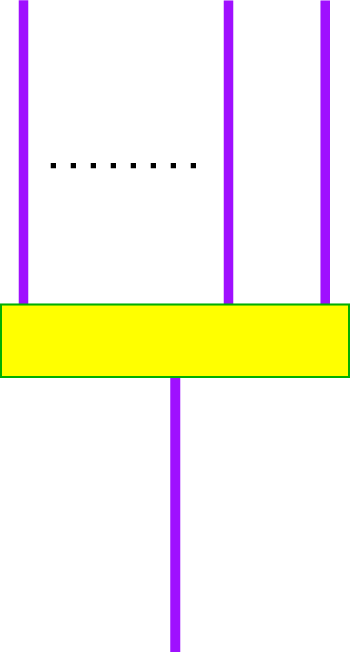}}
  \put(27,65)	{\footnotesize $\varkappa_v'$}
  \put(25,0)	{\footnotesize $X_{e'_v}\vv$}
  \put(-5,138)	{\footnotesize $X_{e'_1}\vvi$}
  \put(13,138)	{\footnotesize $\cdots$}
  \put(26,138)	{\footnotesize $X_{e'_{N-2}}\vvi$}
  \put(53,138)	{\footnotesize $X_{e'_{N-1}}\vvi$}
  \put(94,65)	{$=$}
  \put(135,10)	{\Includepic{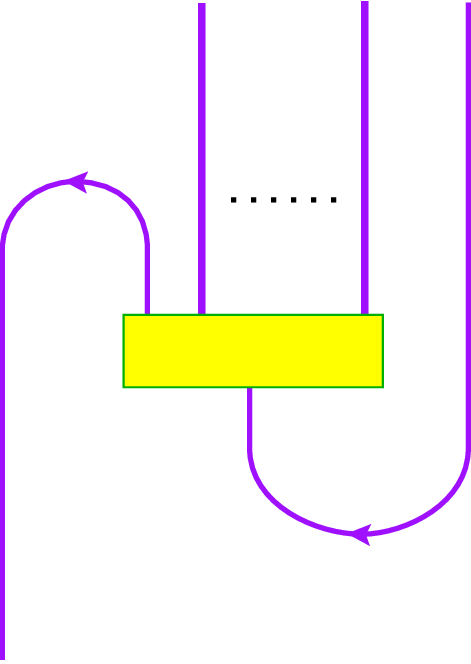}}
  \put(133,0)	{\footnotesize $X_{e'_v}\vv$}
  \put(177.5,65.5){\footnotesize $\varkappa_v$}
  \put(160,139)	{\footnotesize $X_{e'_1}\vvi$}
  \put(178,139)	{\footnotesize $\cdots$}
  \put(191,139)	{\footnotesize $X_{e'_{N-2}}\vvi$}
  \put(218,139)	{\footnotesize $X_{e'_{N-1}}\vvi$}
  }
Any other choice of distinguished edge is obtained from the original one by
iterating this equivalence an appropriate number of times.
It is straightforward to check that an $N$-fold iteration brings one back
to the original morphism, as needed for consistency.

Since the incoming boundary edge defines a natural total ordering of edges
incident to a boundary vertex, there is no need for analogous
considerations for vertices $v$ with $\assd(v)\iN\partial\tws$.
(Recall that in our conventions the distinguished edge of such vertex
is indeed required to be the incoming boundary edge.)
\end{itemize}

%%%%%%%%%%%%%%%%%%%%%%%%%%%%%%%%%%%%%%%%%%%%%%%%%%%%%%%%%%%%%%%%%%%%%%%%

\subsubsection{\Essential\ equivalence}\label{ssec:essequiv}

Other equivalences between world sheets involve manipulations of the defect
graph $\D$. To address these, some additional notation is helpful.
By a \emph{transparent edge} we mean an edge labeled by a Frobenius
algebra $A$ (viewed as a bimodule over itself). For $\D$ the defect graph
of a world sheet $\ws$, we introduce three substructures $\DX$, \TD\ and \DTD,
as illustrated in \eqref{subgraphs} below. For each of them
the number of vertices attached to an edge may be $0$, $1$, or $2$, and for
lack of a better term we still refer to them as sub\emph{graphs}. First,
denote by $\DX$ the subgraph obtained by removing from $\D$ all transparent
edges and all vertices incident to them. Similarly, by $\TD$ we denote the
subgraph that consists of all transparent edges together with all vertices in
$\D \,{\setminus}\, \P$ incident to transparent edges. Finally, \DTD\ is given
by $\DX\,{\cup}\,\P$ adjoined with all half edges incident to $\P$.
Note that $\DTD \,{\cup}\, \TD \eq \D$. We refer to $\TD$ as the
\emph{transparent subgraph}, to $\DX$ as the \emph{opaque subgraph}, and
to $\DTD$ as the \emph{opaque subgraph with insertions}.

The following picture illustrates these structures in an example. In the
picture, transparent edges are drawn as dotted lines, while all other edges
are drawn as solid lines, and vertices in $\P$ are marked with $\phi$ or
$\phi'$, while vertices in $\D\,{\setminus}\,\P$ are unmarked.
  \eqpic{subgraphs}{350}{80}{
  \put(0,100)  {\Includepicfj{3}{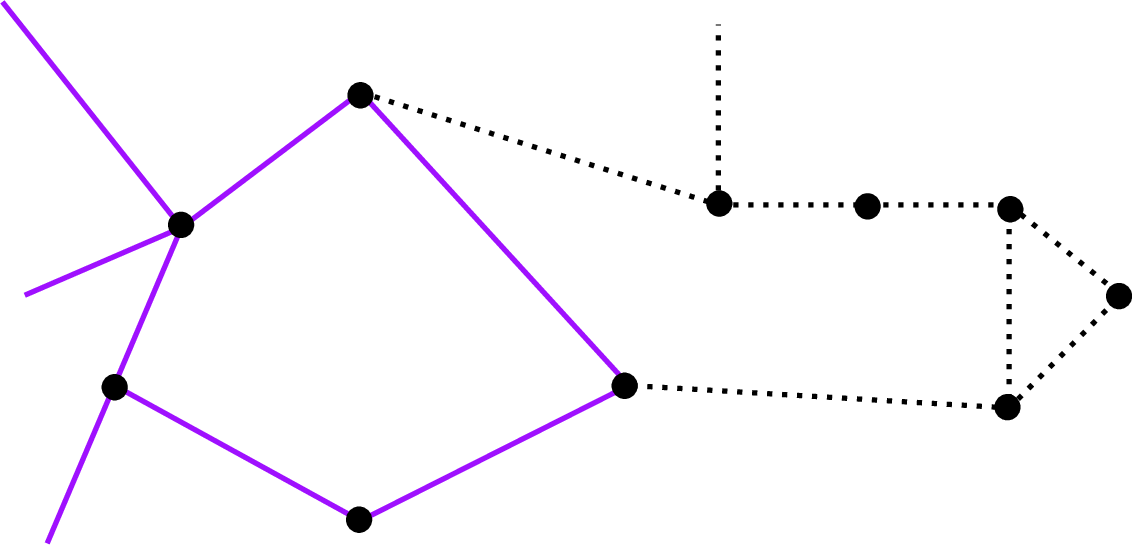}}
  \put(-42,133)	{$ \D ~= $}
  \put(123,155)	{\small $\phi'$}
  \put(49,110)	{\small $\phi$}
  \put(260,100){\Includepicfj{3}{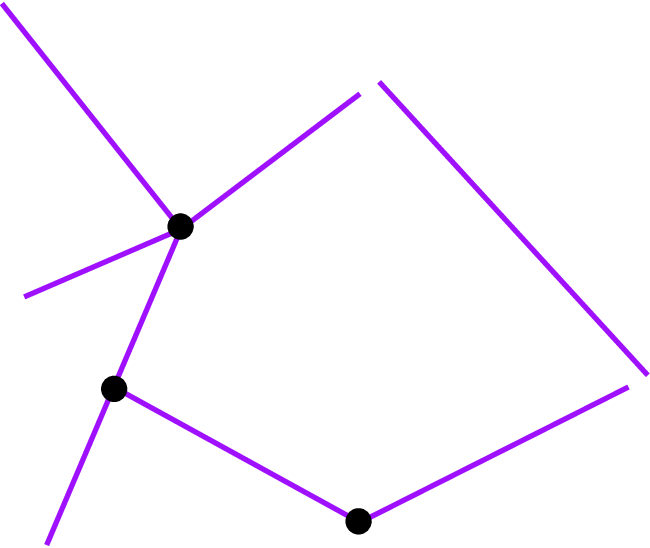}}
  \put(218,133)	{$ \DX ~= $}
  \put(309,110)	{\small $\phi$}
  \put(0,0)	  {\Includepicfj{3}{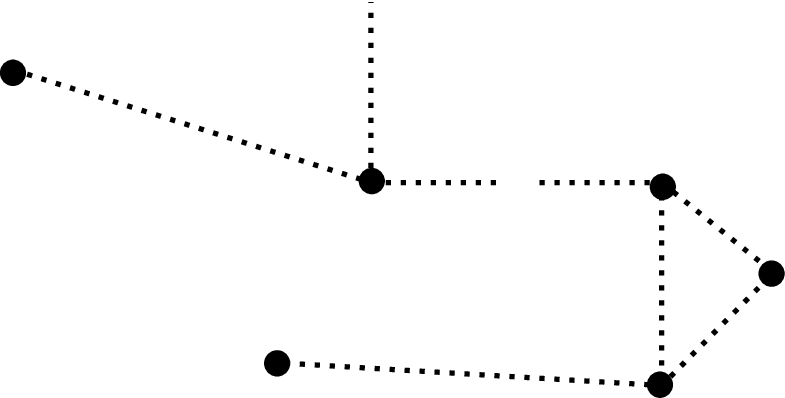}}
  \put(-42,18)	{$ \TD ~= $}
  \put(220,0)	{\Includepicfj{3}{45b}}
  \put(325,45)	{\Includepicfj{3}{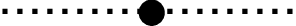}}
  \put(178,18) 	{$ \DTD ~= $}
  \put(269,10)	{\small $\phi$}
  \put(343,53)	{\small $\phi'$}
  }

Having introduced these sub`graphs', we can consider the following notion: We say that
the defect graphs $\D$ and $\D'$ together with their assignments of data, belonging to
world sheets $\ws$ and $\ws'$ respectively, are \emph{\essentially\ equivalent} iff:
   \begin{itemize}\addtolength{\itemsep}{-6pt}%
\item[(i)]
the sets of (labeled) vertices of the two opaque subgraphs with insertions,
i.e.\ of $\DTD$ and $\DTD'$, coincide, and
\item[(ii)]
the set of decorated edges in $\DTD$ and $\DTD'$, coincide,
if necessary after replacing a subset of edges in $\DX$ or $\DX'$ having one or zero
incident vertices with a smaller set having the same decorations (and also, if necessary,
after first changing orientations and the corresponding assignments).
\end{itemize}

The meaning of the requirement (ii)
is illustrated in the following picture, in which
$X$ and $X\ve$ denote the bimodules assigned by $\assb$ respectively $\assb'$
to the corresponding edges:
  \eqpic{essequiv}{280}{15}{
  \put(0,0)	{\Includepic{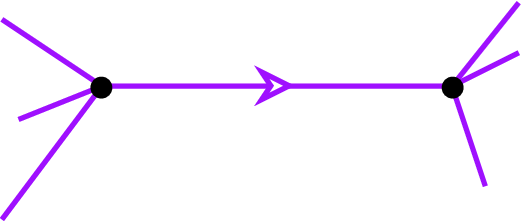}}
  \put(-46,22)	{$ \DX ~= $}
  \put(47,31)	{\scriptsize $X$}
     \put(50,0){
  \put(160,0)	{\Includepic{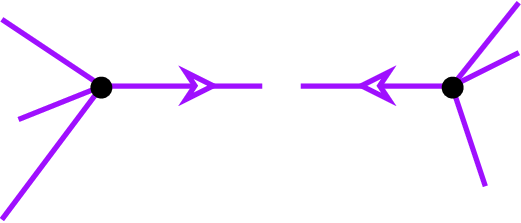}}
  \put(114,22)	{$ \DX' ~= $}
  \put(192,31)	{\scriptsize $X$}
  \put(225,31)	{\scriptsize $X\ve$}
  } }

\begin{lemma}\label{lem:intrinsiceq}
Two world sheets $\ws$ and $\ws'$ with equal underlying manifolds $\tws \eq \tws'$
and with \essentially\ equivalent defect graphs $\D$ and $\D'$
    for which the restrictions $\assd|_{\DTD}$ and $\assd'|_{\DTD'}$ coincide,
are equivalent.
\end{lemma}

\begin{proof}
(i) Let us first formulate three specific instances of \essential\ equivalence:
      \def\leftmargini{1.57em}~\\[-1.45em]
\begin{itemize}\addtolength{\itemsep}{-6pt}%
 \item[\nxt]
\emph{Sliding of structure vertices:}
We may `slide' any structure vertex, i.e.\ vertex labeled by a representation morphism
(of modules or bimodules, respectively) past any other vertex, in the way
indicated in the following picture.
  \eqpic{vertex_slide}{360}{44}{
  \put(0,-5)	{\Includepic{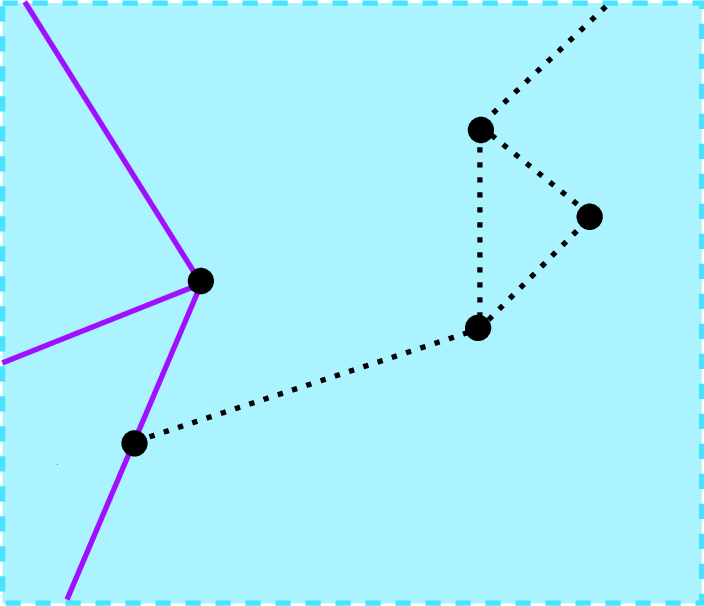}}
  \put(220,-5)	{\Includepic{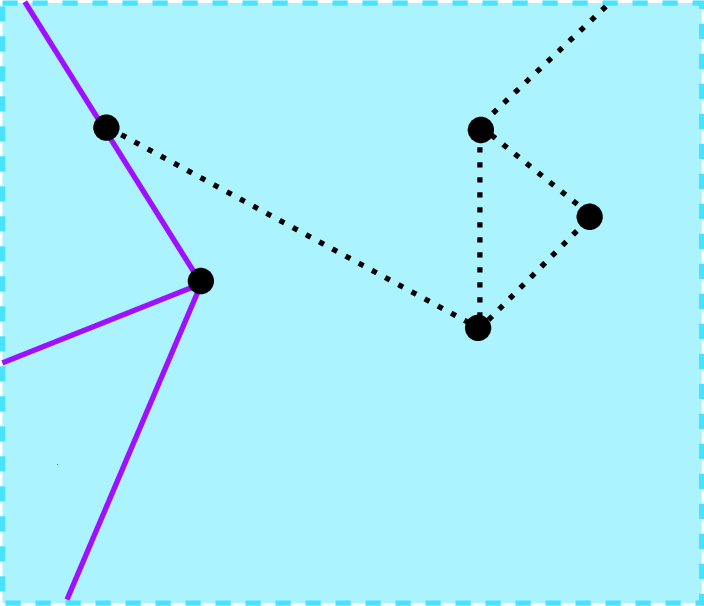}}
  \put(-23,91)	{\fbox{$D$}}
  \put(196,91)	{\fbox{$D'$}}
  \put(40,53)	{\scriptsize $u$}
  \put(17,24)	{\scriptsize $v$}
  \put(260,54)	{\scriptsize $u$}
  \put(242,85)	{\scriptsize $v'$}
  }
 \item[\nxt]
\emph{Removal of transparent tadpoles:}
Consider a subgraph consisting of two transparent edges
and two vertices that are labeled by the coproduct $\Delta$ and by a (bi)module
representation morphism $\rho$, respectively, as shown on the left hand side of
the picture below. Such a subgraph may be discarded
entirely, in the way shown on the right hand side:
  \eqpic{remstrvert}{240}{44}{
  \put(0,-5)	{\Includepic{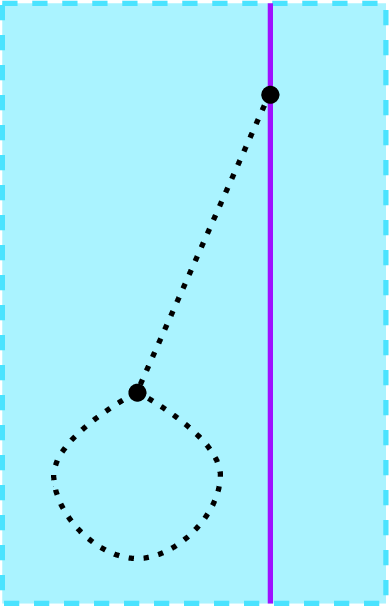}}
  \put(52.5,97){\footnotesize $X$}
  \put(52.5,45)	{\footnotesize $X$}
  \put(35,90)	{\footnotesize $\rho_X^{}$}
  \put(14.5,36)	{\footnotesize $\Delta$}
  \put(180,-5)	{\Includepic{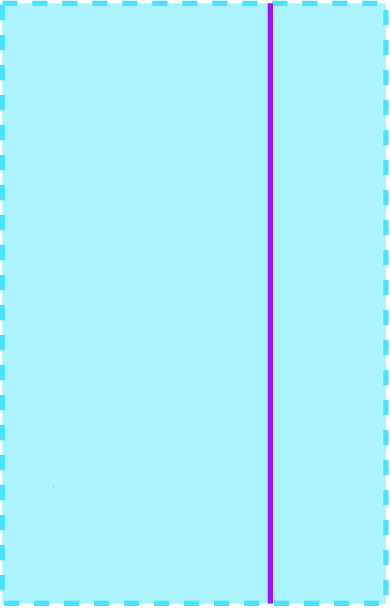}}
  \put(233,50)	{\footnotesize $X$}
  \put(-23,93)	{\fbox{$D$}}
  \put(155,93)	{\fbox{$D'$}}
  }
   \item[\nxt]
\emph{Independence of transparent subgraph:}
Two world sheets $\ws$ and $\ws'$ are equivalent if their opaque subgraphs with insertions $\DTD$ and $\DTD'$ as well as the decorations restricted to $\DTD \eq \DTD'$
agree. In other words, world sheets
have the same correlator if they differ only in the part of $\D$ that consists
of vertices labeled by structure morphisms of algebras and of the edges
$e$ with $\assd(e) \,{\cap}\, \partial\tws \eq \emptyset$ that are
incident to those vertices.
\end{itemize}

~\\[-17pt]
(ii) It is easily checked that any two intrinsically equivalent graphs are related
by a sequence of modifications each of which is of one of the three specific types
described in (i).
\\[2pt]
(iii)
Now independence under sliding of structure vertices is almost tautological since,
owing to the fact that the vertex past which the sliding occurs is labeled by a module
or bimodule morphism, both sides of such an equivalence define the same morphism.
Independence under removal of transparent tadpoles becomes tautological once one
employs the fact that, for a symmetric special Frobenius algebra $A$, the morphisms
$(\id_A \,\oti\, \tilde d_A) \cir (\Delta\,\oti\,\id_{A^\vee}) \cir b_A$ and
$(d_A \,\oti\, \id_A) \cir (\id_{A^\vee} \,\oti\, \Delta) \cir \tilde b_A$
both coincide with the unit morphism $\eta$. Finally, independence of transparent
subgraph is established in Appendix \ref{app:iots}.\endofproof
\end{proof}

%%%%%%%%%%%%%%%%%%%%%%%%%%%%%%%%%%%%%%%%%%%%%%%%%%%%%%%%%%%%%%%%%%%%%%%%

\subsubsection{Further equivalences involving the defect graph}\label{sec:equivfus}

Finally there are equivalences in which the defect graph
is changed more drastically than in \essential\ equivalence.

      \def\leftmargini{1.63em}~\\[-2.55em]
\begin{itemize}\addtolength{\itemsep}{-6pt}%
   \item[\nxt]
\emph{Local defect fusion}:
Defect lines can be fused; thus `parallel' pieces of two defect lines may be replaced
locally by a single defect line $e_{uv}$, at the cost of adding two extra vertices
$u$ and $v$. More explicitly, consider an open contractible region
$D \,{\subset}\, \tws\setminus\partial\tws$ that transversally intersects the
embedding of two parallel edges $e_1$ and $e_2$ and does not contain any vertex
of $\assd(\D)$.  Denoting the same region in $\ws'$ by $D'$, the regions in
which the two world sheets differ look like as indicated in the following figure.
  \eqpic{locfusion}{340}{55}{ \setlength\unitlength{.842pt}
  \put(0,0)	{\includepic{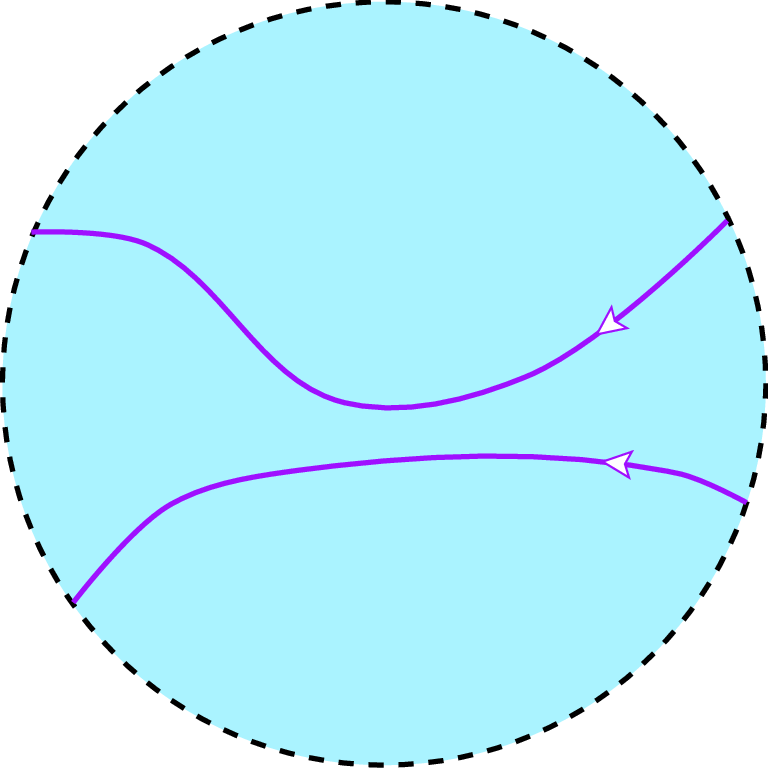}}
  \put(-11,130)	{\fbox{$D$}}
  \put(73,71)	{\footnotesize $e_1$}
  \put(93,61)	{\footnotesize $e_2$}
  \put(250,0)	{\includepic{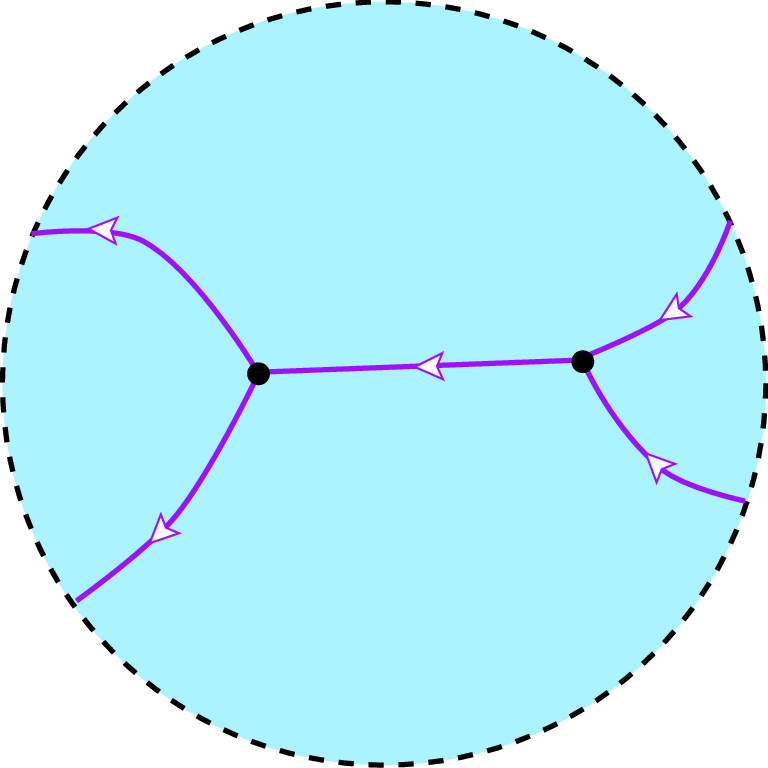}}
  \put(239,130)	{\fbox{$D'$}}
  \put(244,99)	{\footnotesize $\dot e_1$}
  \put(252,25)	{\footnotesize $\dot e_2$}
  \put(386,102)	{\footnotesize $\ddot e_1$}
  \put(390,45)	{\footnotesize $\ddot e_2$}
  \put(320,80)	{\footnotesize $e_{uv}$}
  \put(298,64)	{\footnotesize $u$}.
  \put(348,66)	{\footnotesize $v$}
  }
Further, let the faces of $\ws$ that intersect $D$ be labeled by Frobenius
algebras $A_0$, $A_1$ and $A_2$,
in such a way that $\assb(e_1) \eq X_1$ is an $A_0$-$A_1$-bimodule and
$\assb(e_2) \eq X_2$ is an $A_1$-$A_2$-bimodule. Also, take $e_{uv}$ as the
distinguished edge for both $u$ and $v$. Then $\ws$ and $\ws'$ are equivalent
if the edge $e_{uv}$ is labeled by $X_1\,\oti_{A_1}X_2$,
the other four edges in $D'$ are labeled as $\assb'(\dot e_i) \eq \assb'(\ddot e_i)
\eq X_i$ for $i \eq 1,2$, and the morphisms $\eTA{u}$ and $\rTA{v}\ve$ labeling
the vertices $u$ and $v$ are the embedding morphism $\eTA{u} \eq e \iN
\Hom_{A_0|A_2}(X_1\oti_{A_1}X_2,X_1\oti X_2)$ and the analogous restriction
morphism $\rTA v \eq r$ of the retract $(X_1\oti_{A_1}X_2,e,r)$.
\item[\nxt]
\emph{Local boundary fusion:}
A defect line may be locally fused with a boundary component in a manner analogous
to local defect fusion. This equivalence is in fact already accounted for by the
local defect fusion equivalence, namely as the particular case of $A$-$\one$-bimodules.
\item[\nxt]
\emph{Removal of internal \fusion\ vertices:}
Consider an edge $e_{vw}$ that connects two different vertices $v$ and $w$, with
either $v,w \,{\notin}\, \P$ or else $w\iN\P$ and $v\,{\notin}\,\P$ being two-valent.
Such an edge may be discarded, provided that the remaining edges incident to $v$
and $w$ are joined at a single vertex $v'$.
\\
More specifically, consider the case that $v,w \,{\notin}\, \P$. Let $\D'$ be
obtained from $\D$ by removing the edge $e_{vw}$ and identifying the two vertices
$v$ and $w$. Labeling the so obtained vertex by $v'$, the regions in which $\ws$
and $\ws'$ differ can be illustrated as follows:
  \eqpic{removeedgefig}{360}{53}{ \put(0,-3){
  \put(0,0)	{\includepic{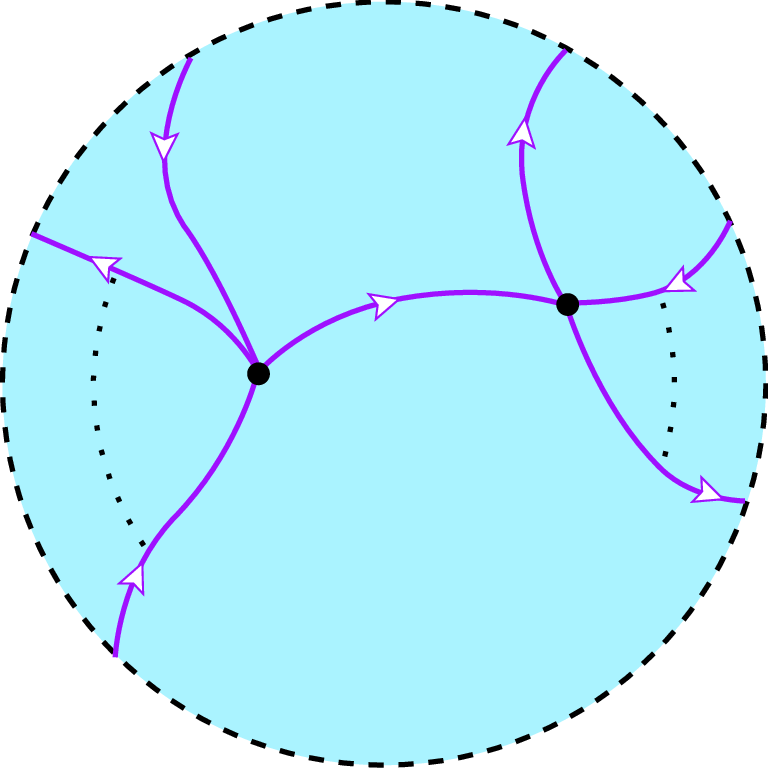}} \setlength\unitlength{.842pt}
  \put(-11,130)	{\fbox{$D$}}
  \put(30,138)	{\footnotesize $e^v_1$}
  \put(-7,100)	{\footnotesize $e^v_2$}
  \put(7,15)	{\footnotesize $e^v_k$}
  \put(47,80)	{\footnotesize $v$}
  \put(70,93)	{\footnotesize $e_{vw}$}
  \put(103,89)	{\footnotesize $w$}
  \put(139,45)  {\footnotesize $e_1^{w}$}
  \put(135,103)	{\footnotesize $e_{m-1}^{w}$}
  \put(105,135)	{\footnotesize $e_m^{w}$}
  \put(270,0)	{\includepic{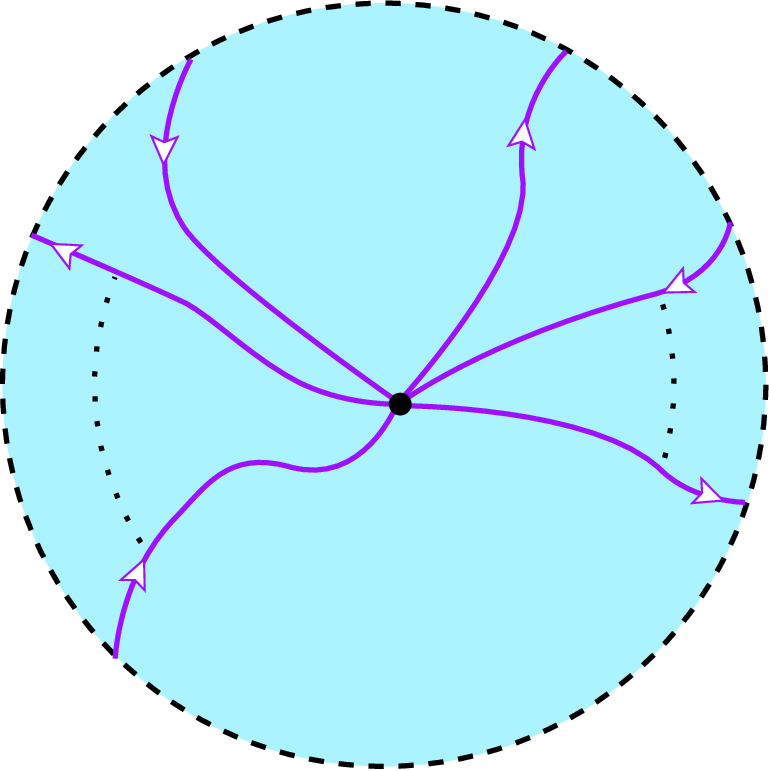}}
  \put(252,130)	{\fbox{$D'$}}
  \put(287,138)	{\footnotesize $e_{m+1}'$}
  \put(248,100)	{\footnotesize $e_{m+2}'$}
  \put(271,15)	{\footnotesize $e_{k+m}'$}
  \put(409,45)	{\footnotesize $e_1'$}
  \put(405,103)	{\footnotesize $e_{m-1}'$}
  \put(375,135)	{\footnotesize $e_{m}'$}
  \put(339,73)	{\footnotesize ${v'}$}
  } }
For definiteness, assume that $e_{vw}$ is the distinguished edge of $v$ and $w$,
and that $e'_{m+1}$ is the distinguished edge of $v'$. Equivalence of $\ws$ and
$\ws'$ holds if $\assb'(e_{i}') \eq \assb(e_i^{w})$ for $i\eq 1,2,...\,,m$ and
$\assb'(e_{m+j}') \eq \assb(e_j^{v})$ for $j\eq 1,2,...\,,k$, as well as
  \eqpic{removeedge}{240}{59}{
  \put(5,65)   {$\assp'(v')~=$}
  \put(75,10)	{\Includepic{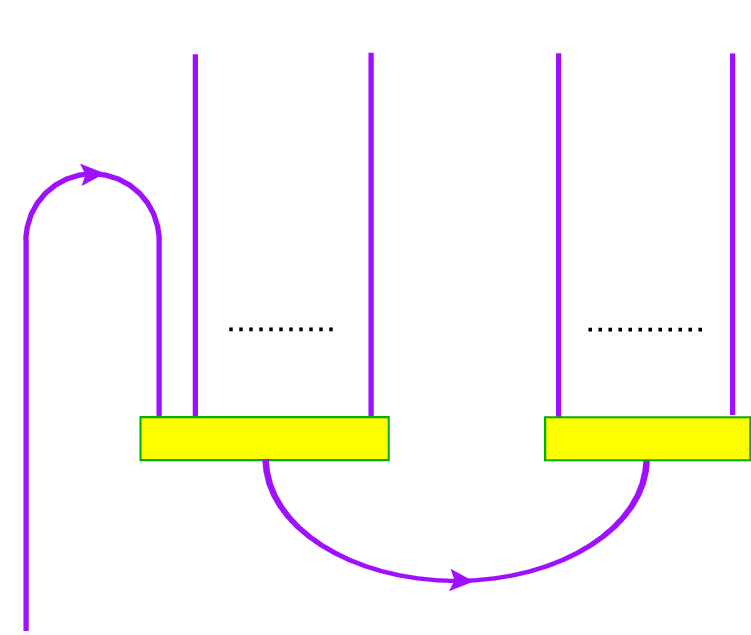}
  \put(-5,-10)	{\footnotesize $X_{e'_{m+1}}\vv$}
  \put(43,34.5)	{\footnotesize $\varkappa_v$}
  \put(111,34.5){\footnotesize $\varkappa_w$}
  \put(25,116)	{\footnotesize $X_{e'_{m+2}}\vvi$}
  \put(44,116)	{\footnotesize $\cdots$}
  \put(58,116)	{\footnotesize $X_{e'_{m+k}}\vvi$}
  \put(97,116)	{\footnotesize $X_{e'_{1}}\vvi$}
  \put(115,116)	{\footnotesize $\cdots$}
  \put(129,116)	{\footnotesize $X_{e'_{m}}\vvi$}
  } }
The second case, i.e.\ $v\,{\notin}\,\P$ being two-valent and $w\iN\P$, can be
treated in a similar way; we omit the details.
  \item[\nxt]
\emph{Removal of a boundary \fusion\ vertex:}
If a boundary component stretches from a vertex $v$ to a vertex $w$ in such a way
that either $v,w\,{\notin}\, \P$ or else $w\iN\P$ and $v\,{\notin}\,\P$ is two-valent,
it can be removed if combined with a relabeling completely analogous to the removal
of an internal defect edge. Again this case is covered by considering $A$-$\one$-bimodules
on edges lying on the boundary.
  \item[\nxt]
\emph{Collapse of a defect bubble:} A collection of defect lines forming a `bubble'
and involving at most one insertion point can be replaced by a single vertex.
Invoking equivalences already treated it is sufficient to consider the situation
that $\ws$ and $\ws'$ differ only in a region of the form indicated in the following
figure, in which $v_3\iN\P$.
  \eqpic{defbubblefig}{340}{54}{ \setlength\unitlength{.842pt}
  \put(0,0)	{\includepic{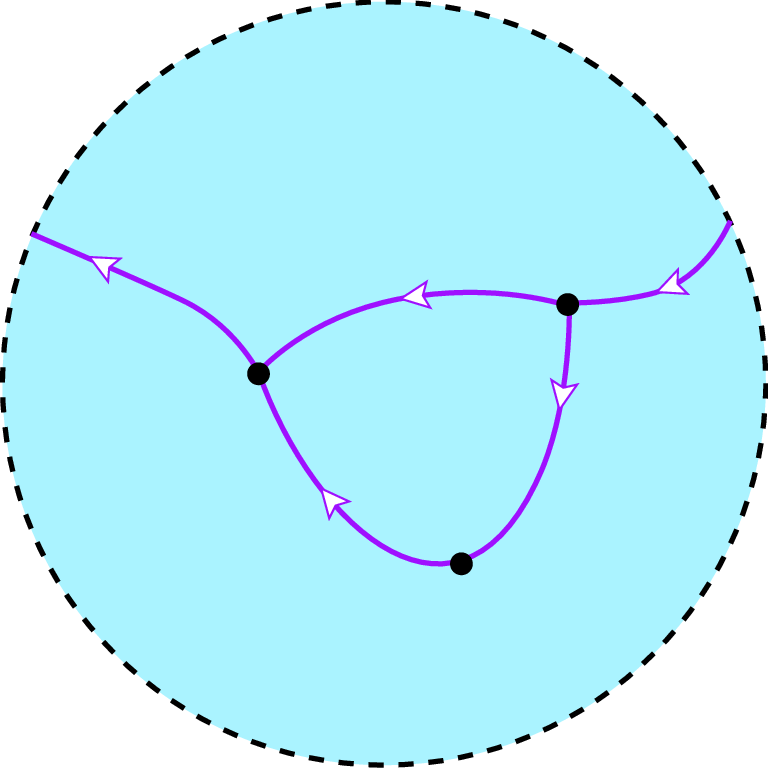}}
  \put(-11,130)	{\fbox{$D$}}
  \put(-6,100)	{\footnotesize $e_{1}$}
  \put(65,93)	{\footnotesize $e_{21}$}
  \put(103,56)	{\footnotesize $e_{23}$}
  \put(39,53)	{\footnotesize $e_{31}$}
  \put(135,103)	{\footnotesize $e_{2}$}
  \put(44,80)	{\footnotesize $v_1$}
  \put(100,90)	{\footnotesize $v_2$}
  \put(81,29)	{\footnotesize $v_3$}
  \put(260,0)	{\includepic{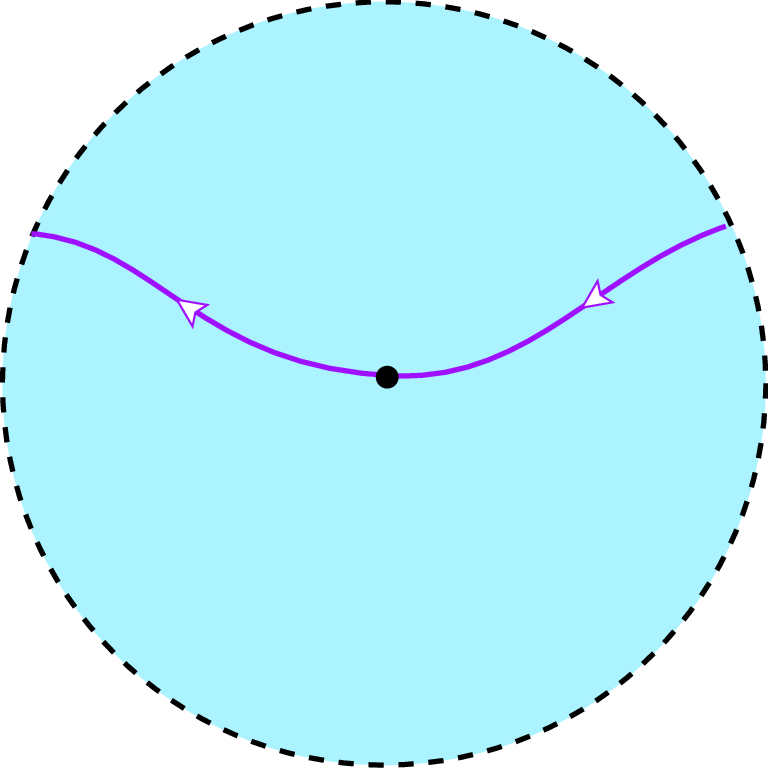}}
  \put(249,130)	{\fbox{$D'$}}
  \put(254,100)	{\footnotesize $e_{1}'$}
  \put(395,103)	{\footnotesize $e_{2}'$}
  \put(326.6,76){\footnotesize ${w}$}
  }
The world sheets $\ws$ and $\ws'$ are equivalent if $\assb'(e_1')\eq \assb(e_1)$,
$\assb'(e_2') \eq \assb(e_2)$ and
  \eqpic{collapsebubble}{190}{68}{
  \put(0,70)   {$\assp'(w)~=$}
  \put(75,10)	{\Includepic{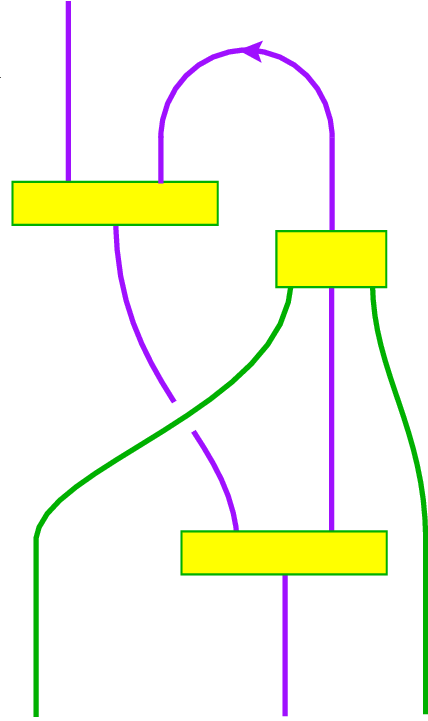}
  \put(2,-10)	{\footnotesize $U_w$}
  \put(72,-10)	{\footnotesize $U'_w$}
  \put(45,-10)	{\footnotesize $X_{e_{2}}$}
  \put(48.5,45)	{\begin{turn}{90}\footnotesize $X_{e_{23}}$\end{turn}}
  \put(44,29.5)	{\footnotesize $\varkappa_{v_2}$}
  \put(55,82)	{\footnotesize $\phi_{v_3}$}
  \put(5,68)	{\footnotesize $X_{e_{21}}$}
  \put(13,93)	{\footnotesize $\varkappa_{v_1}$}
  \put(6,137)	{\footnotesize $X_{e_{1}}$}
  \put(63,100)	{\footnotesize $X_{e_{31}}$}
  } }
The morphism $\assp'(w)$ is the same as the morphism obtained
from the domain $D$. Furthermore $\assp'(w)$ belongs to the space
$\Hom_{A_1|A_2}(U_w\oti^+X_{e_{2}}\oti^-U_w',X_{e_1})$ and is therefore indeed an
adequate morphism to be used in $D'$.
 \\
Note that in the particular situation that the edges $e_{31}$ and $e_{23}$ are
labeled by transparent defects, this equivalence relates
a bulk field to a defect field.
Also, we can regard the situation that $v_3$ is absent as a special case of
\eqref{collapsebubble}; thereby the equivalence also describes the removal of
a defect bubble without insertion vertex.
 \item[\nxt]
\emph{Collapse of a boundary defect bubble:}
Consider the situation that the two end-points of an edge $e$ are located on the same
boundary component in such a way that exactly one boundary field insertion is
present on the segment of the boundary between $u$ and $v$. Then there is an
equivalent world sheet for which the resulting bubble and the boundary field
are replaced by just a single boundary field. This is seen in complete analogy with
the defect bubble case just described.
  \end{itemize}

For notational convenience, definite orientations of the edges in the equivalences
listed above have been chosen. But making these choices does not constitute any
restriction. Indeed, in case an edge $e$ has orientation opposite to the one chosen
above, we can invoke the equivalence of world sheets that differ only in the
orientation of a single edge $e\iN\D$ so as to arrive at one of the situations
described above.

\begin{remark}\label{rem:equiv_moves}
All the relations listed above are indeed equivalence relations in the technical
sense, and are thus in particular symmetric with respect to the two world sheets
$\ws$ and $\ws'$ involved. Below we will, as already indicated by the chosen
designations of these equivalences, often think of an equivalence of the type
depicted in \eqref{locfusion}, \eqref{removeedgefig} and \eqref{defbubblefig}
as an active `move'. By this we mean that starting
from a world sheet of the type denoted by $\ws$, we replace it by an equivalent
one of the type denoted by $\ws'$. Even though we deal with equivalence relations,
in terms of these moves the description is asymmetric: while
in the pictures \eqref{locfusion}, \eqref{removeedgefig} and \eqref{defbubblefig}
we can always move from left to right, we can\emph{not} always move from right to
left. Concretely, we can `un-fuse' defects only if the decorations on
the right hand side are of the specific form stated in the respective equivalences.
Also, while removal of internal network vertices can for all values of $k,m \ge 1$
be performed in both directions, moving in picture \eqref{removeedgefig} from right
to left constitutes an independent `move' only if $k\eq m \eq 1$, since otherwise
one can apply fusion of defect lines.
\end{remark}

%%%%%%%%%%%%%%%%%%%%%%%%%%%%%%%%%%%%%%%%%%%%%%%%%%%%%%%%%%%%%%%%%%%%%%%%

\section{Boundary defect factorization}\label{sec:bdfact}

Unlike bulk factorization, boundary factorization in the presence of defects is
not an essentially new feature. Rather, it can entirely be reduced to
Theorem 2.9 of \cite{fjfrs}. To see this, consider a rectangular region of a world sheet
$\ws$ containing two intervals of the boundary, and with a finite number of defect lines
running parallel to the boundary and intersecting transversally a \emph{factorization line},
i.e.\ the image of an embedding $f\colon [0,1]\To \tws$ of an interval, stretching between
the two boundary intervals.
(In the figure below the factorization line is indicated by a dashed line.)
Using recursively local boundary-defect fusion, we may replace this world sheet
with an equivalent one for which the factorization line $f([0,1])$ does not intersect
any defect line. In pictures:
  \eqpic{WSbnd}{300}{76}{
  \put(0,0)	{\Includepic{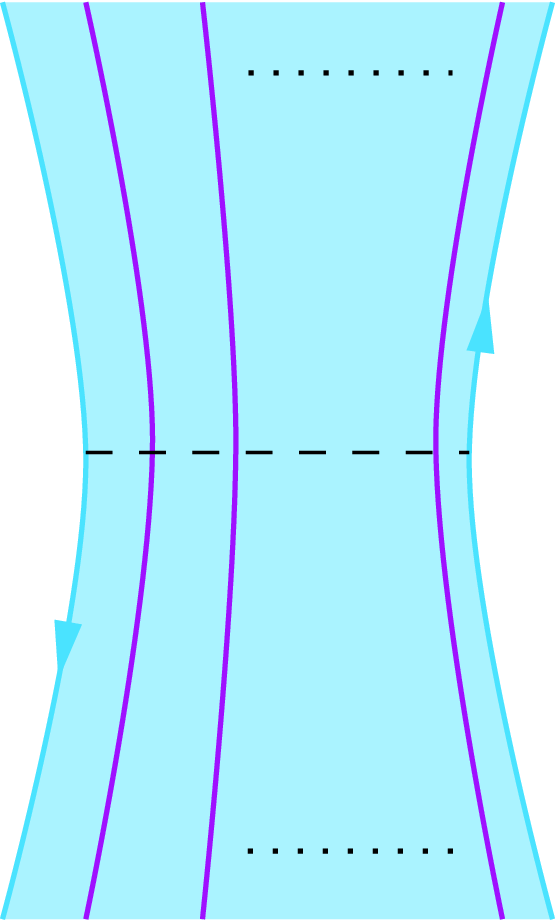}}
  \put(140,83)	{\large$\longmapsto$}
  \put(200,0)	{\Includepic{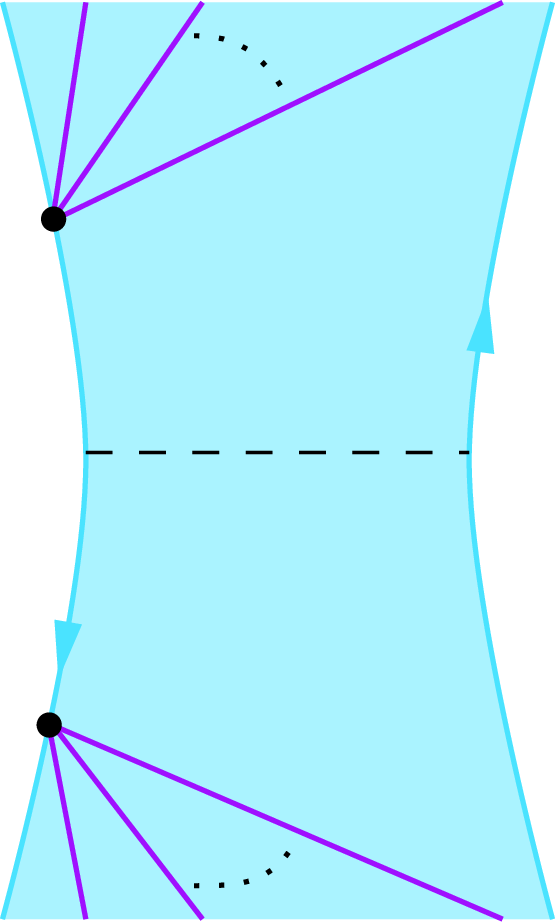}}
  }
This kind of equivalence  reduces factorization along $f$ to the conventional boundary
factorization as described by Theorem 2.9 of \cite{fjfrs}.

Albeit thus not an independent novel feature, it seems to us nevertheless to be of
interest to establish a genuine boundary defect factorization result. The proofs of
the pertinent assertions below are, however, fully parallel
to the corresponding proofs in \cite{fjfrs} and will therefore be omitted.

Let $A$ and $B$ be simple symmetric special Frobenius algebras in $\cC$, $M_l$
a left $B$-module, $M_r$ a left $A$-module, and $X$ an $A$-$B$-bimodule. Then by
$D \eq D(M_l,M_r,X,i,\psi_+,\psi_-,\chi_+,\chi_-)$ we denote the following
world sheet:
  \eqpic{DD2pt}{110}{59}{ \setlength\unitlength{.842pt}
  \put(10,10)	{\includepic{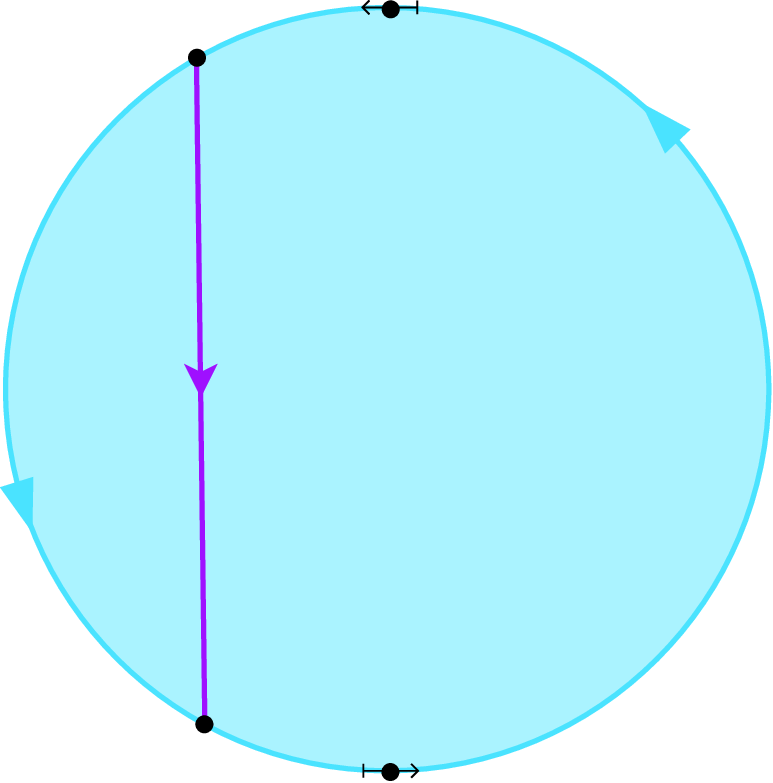}}
  \put(77,157)	{\footnotesize $\psi_+$}
  \put(35,148)	{\footnotesize $\chi_+$}
  \put(77,3)	{\footnotesize $\psi_-$}
  \put(35,15)	{\footnotesize $\chi_-$}
  \put(153,80)	{\footnotesize $M_r$}
  \put(-6,80)	{\footnotesize $M_l$}
  \put(50,75)	{\footnotesize $X$}
  }
Thus the defect graph $\D$ of $D$ consists of four vertices and five edges.
Four of the edges are boundary edges, labeled by the modules $M_l$, $X\oti_B M_l$
(twice) and $M_r$, respectively, while the remaining one is an internal edge,
labeled by the bimodule $X$. Two of the vertices are insertion vertices, labeled
by $\psi_+\iN \Hom_A((X\oti_B M_l)\,\oti U_i,M_r)$ and by
$\psi_-\iN \Hom_A(M_r\,\oti\, U_\ib,X\oti_B M_l)$, respectively,
and the other two are \fusion\ vertices, labeled by
$\chi_+\iN\HomP_B(M_l,X\ve\oti(X\oti_B M_l))$ and by
$\chi_-\iN\HomP_A(X\oti_B M_l,X\oti M_l))$, respectively, and in all cases
the distinguished edge is the incoming boundary edge. The internal edge
is taken to be embedded by $\assd$ as a straight line stretching between the
two \fusion\ vertices. Denoting by $\{\psi_\alpha\}$ and $\{\varphi_\beta\}$ choices
of bases of the morphism spaces for the two boundary fields, the structure constants
$c^{\mathrm{bdef}}$ of the correlator are defined through
  \be
  C(D) = {(c_{M_l,M_r,X,p}^{\mathrm{bdef}})}_{\alpha\beta}\, Z(B_{p,\bar p}) \,,
  \ee
where $\psi_+ \eq \psi_\alpha$, $\psi_- \eq \varphi_\beta$, $\tilde\chi \eq \eTA{}$ and
$\chi \eq (\id_{X\ve} \,\oti\,\rTA{}) \cir (\tilde b_X\oti \id_{M_l})$ (with $\eTA{}$
and $\rTA{}$ the appropriate embedding and restriction morphisms), and where the cobordism
$B^-_{p \bar p}$ is the three-ball shown in picture (2.34) of \cite{fjfrs}.

Let now $\ws$ be a world sheet containing a rectangular region
with two boundary intervals labeled by $M_r$ and $M_l$ and
(invoking, if necessary, local fusion) a single defect line labeled by $X$
intersecting the factorization line $f$ transversally. Factorization along $f$
involves cutting along $f([0,1])$ and gluing back the top and bottom halves of
$D(M_l,M_r,X,U_p,\psi_\alpha,\varphi_\beta,\chi,\tilde\chi)$,
as shown in the following figure:
  \eqpic{BdDfact}{400}{120}{ \put(-10,-4){
  \put(0,45)	{\Includepic{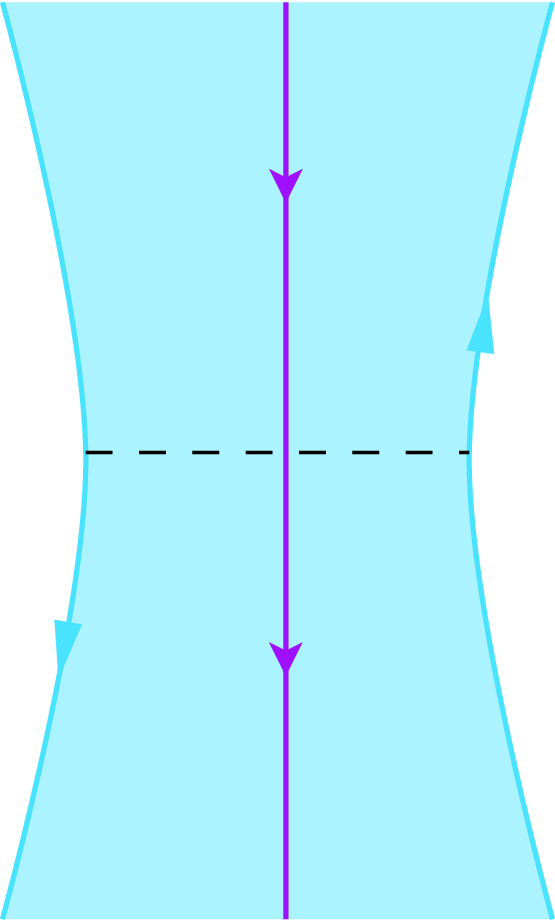}}
  \put(89,117)	{\footnotesize $M_r$}
  \put(0,117)	{\footnotesize $M_l$}
  \put(55,155)	{\footnotesize $X$}
  \put(118,127) {$ \longmapsto $}
  \put(150,30)	{\Includepic{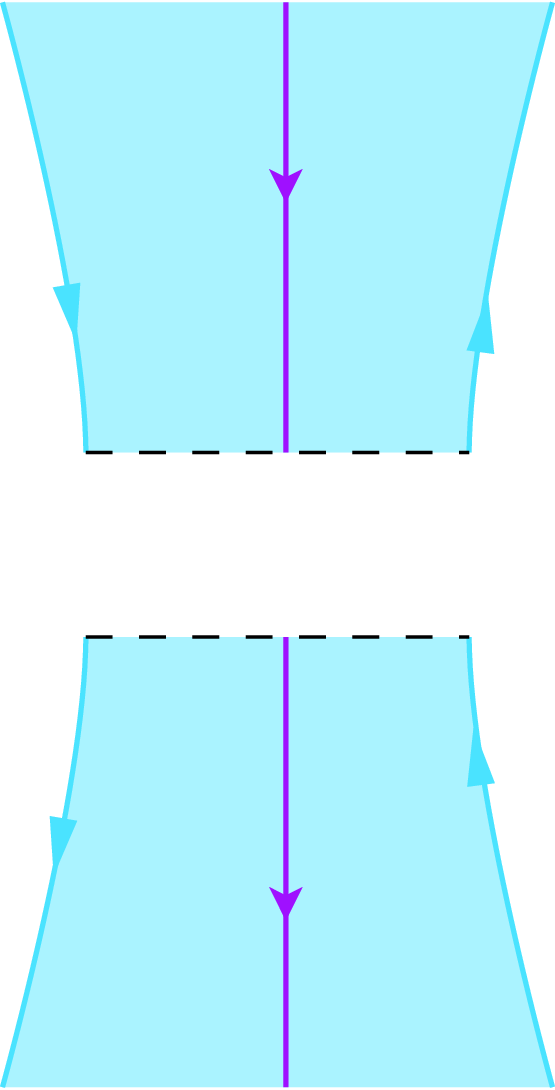}}
  \put(240,100)	{\footnotesize $M_r$}
  \put(150,100)	{\footnotesize $M_l$}
  \put(205,88)	{\footnotesize $X$}
  \put(240,150)	{\footnotesize $M_r$}
  \put(150,150)	{\footnotesize $M_l$}
  \put(205,166)	{\footnotesize $X$}
  \put(271,127) {$ \longmapsto $}
  \put(300,0)	{\Includepic{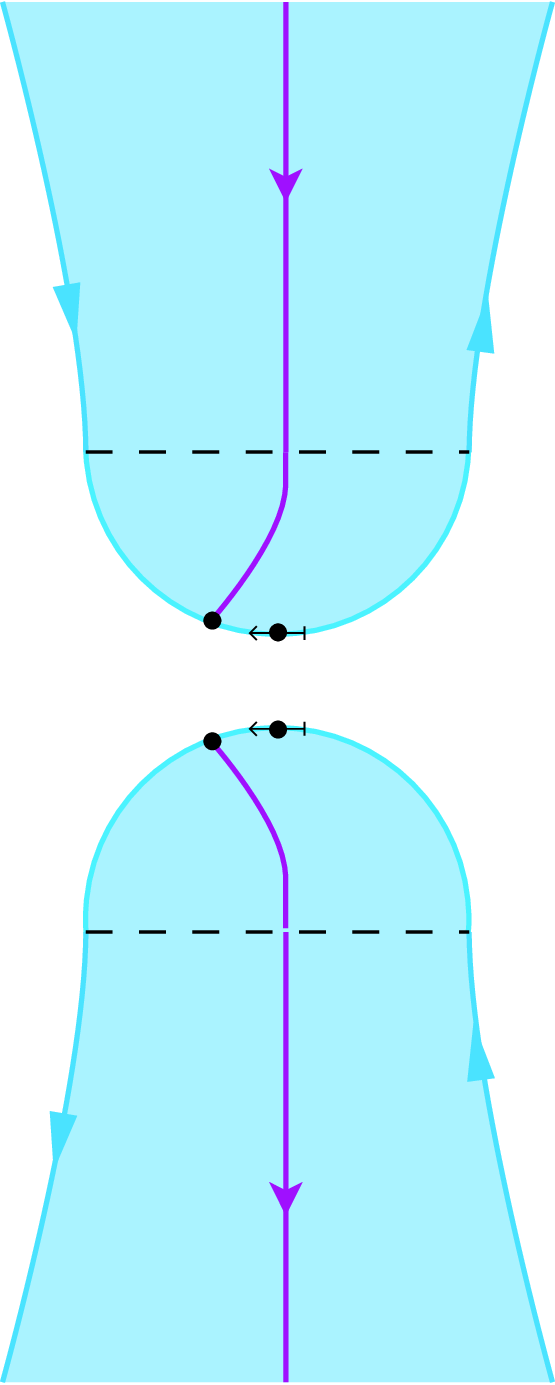}}
  \put(390,70)	{\footnotesize $M_r$}
  \put(300,70)	{\footnotesize $M_l$}
  \put(355,56)	{\footnotesize $X$}
  \put(390,185)	{\footnotesize $M_r$}
  \put(300,185)	{\footnotesize $M_l$}
  \put(355,189)	{\footnotesize $X$}
  \put(355,133)	{\footnotesize $\varphi_\beta$}
  \put(360,120)	{\footnotesize $\psi_\alpha$}
  \put(330,132)	{\footnotesize $\tilde\chi$}
  \put(330,120)	{\footnotesize $\chi$}
  } }
We denote the resulting world sheet by $\ws_{p,\alpha\beta}$.

A corresponding gluing homomorphism $\GLL p{\phantom{q}}^{\mathrm{bnd}}$ can be
defined analogously to the gluing homomorphism $\GLL pq$ that we used for bulk
factorization; it is given in formula (2.37) of \cite{fjfrs}. We can then state the

\begin{thm} Boundary defect factorization:\label{thm:bdfact}
The correlators of the world sheets $\ws$ and $\ws_{p,\alpha\beta}$ are related as
  \be\label{bdefthm}
  C(\ws) = \sum_{p\in\I} \sum_{\alpha,\beta} \dim(U_p)\, {(c^{\mathrm{bdef}\ -1}
  _{M_l,M_r,X,p})}_{\beta\alpha}^{}\, \GLL p{\phantom q}^{\mathrm{bnd}}
  \big( C(\ws_{p,\alpha\beta})\big) \,.
  \ee
\end{thm}

\begin{remark}
Fusing $X$ to the boundary in both $\ws_{p,\alpha\beta}$ and
$D(M_l,M_r,X,U_i,\psi_\alpha,\varphi_\beta,\chi,\tilde\chi)$, Theorem \ref{thm:bdfact}
reduces to Theorem 2.9 of \cite{fjfrs}, and thus these two statements are equivalent.
In particular we then have $\big(c^{\mathrm{bdef}}_{M_l,M_r,X,p}\big)_{\!\alpha\beta}
\eq \big(\mathrm{c}^{\mathrm{bnd}}_{X\oti_B M_l,M_r,p}\big)_{\!\alpha\beta}$, with
the two-point structure constants $\mathrm{c}^{\mathrm{bnd}}$ from
\cite[(2.26)]{fjfrs} defined using the basis morphisms $\psi_\alpha$ and $\varphi_\beta$.
\end{remark}

\begin{remark}
The explicit form of $\ws_{p,\alpha\beta}$ and
$D(M,N,X,U,\psi_+,\psi_-,\varkappa_+,\varkappa_-)$ is correlated with the choices
we make when describing world sheets. Had we allowed for more general boundary field
insertions corresponding to multivalent vertices with additional incident edges,
the boundary defect factorization would be conveniently described with the help of
corresponding
boundary-defect fields. However, the morphisms labeling such more general field
insertions can be expressed as the composition of morphisms labeling conventional
boundary fields and morphisms labeling boundary \fusion\ vertices. This means
that there does not arise any new structure related to factorization
when allowing for more general field insertions.
\end{remark}

%%%%%%%%%%%%%%%%%%%%%%%%%%%%%%%%%%%%%%%%%%%%%%%%%%%%%%%%%%%%%%%%%%%%%%%%

\section{Fundamental correlators}\label{sec:fundcorr}

\subsection{Modular covariance}\label{sec:modcov}

The behavior of correlators under homeomorphisms of world sheets with defects
is quite similar to the behavior in the absence of defects that has been
studied in \cite{fjfrs}. Let $\ws$ and $\ws'$ denote world sheets with
decoration data $\D$, $\assa$, ...  and $\D'$, $\assa'$, ...\,, respectively.
There is a natural notion of \emph{isomorphism} of world sheets, which
involves the \essential\ equivalence of world sheets (recall the latter, and in
particular the definition of the opaque subgraph
$\DX$ of $\D$, from Section \ref{ssec:essequiv}):

\begin{defi}\label{def:wsiso}
An \emph{isomorphism $h\colon \ws \,{\stackrel\simeq\to}\, \ws'$ of world sheets}
is an orientation preserving homeomorphism $h\colon \tws \,{\stackrel\simeq\to}\,
\tws'$ between world sheets with \essentially\ equivalent decorated defect graphs
that preserves the decoration data when restricted to the opaque subgraph with
insertions and for which $h^{-1}$ has these properties as well. More explicitly:
\def\leftmargini{1.45em}~\\[-1.75em]\begin{itemize}\addtolength{\itemsep}{-6pt}%
  \item[\nxx]
The two graphs satisfy $\,\assd'|_{\DX'}^{} \eq h\circ\assd|_{\DX}^{}$;
  \item[\nxx]
$\assd'$ and $h\cir \assd$ determine the same arc germs on insertion points
not belonging to $\DX'$ or $\DX$;
  \item[\nxx]
$\mathrm{y}'|_{\DTD}^{} \eq \mathrm{y}|_{\DTD}^{}$ for the decoration data
$\,\mathrm{y} \eq \assb, \assc, \assp$; and
  \item[\nxx]
$\assa(f) \eq \assa'(h(f))$ for any face $f$ of $\ws$.
\end{itemize}
\end{defi}

The natural notion of \emph{homotopy} of world sheets is homotopy relative to the
structure that is preserved under an isomorphism of world sheets.
In particular, every automorphism of a world sheet is required to preserve
$\assd(\DTD) \,{\subseteq}\, \ws$ as a \emph{decorated} set, and homotopy is defined
relative to this property. The set of
homotopy classes, with respect to this notion, of isomorphisms
$\ws \,{\stackrel\sim\to}\, \ws'$ of world sheets is denoted $\Mapw(\ws,\ws')$.
For $\ws' \eq \ws$ we obtain the \emph{mapping class group}, to be denoted
by $\Mapw(\ws)$.

We stress that an isomorphism is not required to preserve \emph{all} data of a
world sheet; this simplifies the discussion of covariance and mapping class group
invariance of correlators. Clearly, by forgetting the datum of defect lines,
every isomorphism of world sheets becomes an isomorphism of world sheets
without defects. As a consequence, various results established in a setting
not involving defects can be adopted to the present setting with only minor
modifications. Denoting by $\pi_\ws\colon \widehat\ws\To\ws$ the natural
projection, an isomorphism $f\colon \ws\To \ws'$ of world sheets has a unique
orientation preserving lift $\hat f\colon \widehat\ws\To\widehat\ws'$ to the
doubles, defined as the orientation preserving diffeomorphism that satisfies
  \be\label{lifttodouble}
  f\circ \pi_\ws = \pi_{\ws'}\circ \hat f \,.
  \ee
The map $\hat f$ preserves marked arcs and Lagrangian subspaces and is thus an
isomorphism of extended surfaces \cite{fjfrs}. Denote by $\hat f_\sharp\colon
\bl(\widehat\ws)\To \bl(\widehat\ws')$ the linear isomorphism that by the mapping
cylinder is associated to $\hat f$. If $f \,{\sim}\, g$ as isomorphisms of
world sheets, then $\hat f \,{\sim}\,\hat g$ as homeomorphisms of extended surfaces,
and hence $\hat f_\sharp \eq \hat g_\sharp$ by the axioms of topological field
theory.

It follows that there is an isomorphism
  \be\label{Mapw-iso-Maprel}
  \Mapw(\ws) \,\cong\, \mathrm{Map}(\tws;[\assd(\DTD)]) \;\subset\,
  \mathrm{Map}(\tws;\assd(\P)) \,,
  \ee
where $\mathrm{Map}(\tws;[\assd(\DTD)])$ and $\mathrm{Map}(\tws;\assd(\P))$ denote
the groups of homotopy classes of orientation preserving homeomorphisms of $\tws$
relative to the decorated isotopy class $[\assd(\DTD)]$ and relative to the set
$\assd(\P)$, respectively. In other words, we can identify the mapping class group
of a world sheet containing defects with the subgroup of the mapping class group of
the underlying surface with marked points and without defects that preserve the
isotopy class and decorations of the defects.

We are now in a position to make the following assertions:

\begin{thm}\label{thm:modcov}
Covariance of correlators with defects: We have
  \be\label{modcov}
  C(\ws') = \hat f_\sharp\big(C(\ws)\big)
  \ee
for any $[f]\iN \Mapw(\ws,\ws')$.
\end{thm}

\begin{cor}\label{cor:modinv}
Mapping class group invariance of correlators with defects: We have
  \be\label{modinv}
  C(\ws) = \hat f_\sharp\big( C(\ws)\big)
  \ee
for any mapping class $[f]\iN\Mapw(\ws)$.
\end{cor}

\begin{remark} ~\\[2pt]
(i)
Restricting to world sheets with empty opaque subgraph $\DX$ these assertions
reduce to the statements of Theorem 2.2 and Corollary 2.3 of \cite{fjfrs}. For
$\DX \,{\neq}\, \emptyset$, i.e.\ in the presence of non-trivial defects, the
correlators are, however, only invariant under a proper \emph{subgroup} of
the full mapping class group of the underlying surface with marked points.
\\[3pt]
(ii)
The correlator $C(\ws)$ of a given world sheet $\ws$ may be invariant under
a larger subgroup of $\mathrm{Map}(\tws;\assd(\P))$ than implied by Theorem
\ref{thm:modcov} and the isomorphism \eqref{Mapw-iso-Maprel}.
Consider, for instance, a world sheet $\ws$ with an opaque subgraph $\DX$
that has no vertices and every edge of which is decorated by one and the same
\emph{invertible} $A$-bimodule $X$. If $\dim(X) \eq \dim(A)$ (as is the case
for any unitary CFT), then the correlator $C(\ws)$ is invariant under the the
subgroup of $\mathrm{Map}(\tws;\assd(\P))$ that preserves the homology class
of $\assd(\DX)$.
The correlators of closed world sheets of this type provide endomorphisms of
the spaces $\bl(\widehat{\ws})$ of conformal blocks, intertwining the action of
proper subgroups of $\mathrm{Map}(\widehat{\ws})$ containing the Torelli group.
\\[3pt]
(iii)
We have $\hat f_\sharp\cir \hat g_\sharp \eq (\widehat{f{\circ} g})_\sharp$, since
\cite[footn.\,5]{fjfrs} the relevant gluing anomalies vanish. It follows that the
correlators are invariant under a \emph{genuine}, rather than only projective,
action of the relevant mapping class groups.
\end{remark}

The basic result needed for establishing Theorem \ref{thm:modcov} is independence
of transparent subgraph, which was introduced in the proof of Lemma
\ref{lem:intrinsiceq} and which is proven in Appendix \ref{app:iots}.
\\[4pt]
\noindent {\it Proof of Theorem \ref{thm:modcov}.}
\\[1pt]
Note that, since the defect graphs $\D$ and $\D'$ are intrinsically
equivalent (see Definition \ref{def:wsiso}), by invoking
Lemma \ref{lem:intrinsiceq} it suffices to consider the case $\DX \eq \DX'$.
As already stated, the gluing anomaly vanishes, so the right hand side
of \eqref{modcov} is obtained as the invariant of the three-manifold that
results from gluing the mapping cylinder of $\hat{f}$ to the connecting manifold
of $\ws$. Moreover, the so obtained extended three-manifold is homeomorphic,
via a homeomorphism that restricts to the identity on the boundary, to the one
obtained by pushing $\hat{f}$ through the connecting manifold. The latter is
the connecting manifold of a world sheet $\ws'''$ differing from $\ws'$ only
in the transparent subgraph $\TD'''$ and its embedding $\assd'''(\TD''')$,
compare formula (3.12) of \cite{fjfrs}. Independence of transparent subgraph now
immediately implies the equality \eqref{modcov} and thus completes the proof.
\endofproof

%%%%%%%%%%%%%%%%%%%%%%%%%%%%%%%%%%%%%%%%%%%%%%%%%%%%%%%%%%%%%%%%%%%%%%%%

\subsection{Fundamental world sheets}\label{sec:fundws}

Taking world sheets in the sense of Definition \ref{def:ws} as objects, and isomorphisms
in the sense of Definition \ref{def:wsiso} as morphisms, we obtain a \emph{category}
of world sheets (all morphisms are invertible, so this is in fact a groupoid).
This category is monoidal, with the tensor product of objects given by
disjoint union and the tensor unit being the empty set.
The processes of bulk and boundary factorization constitute
additional operations on world sheets. It is convenient to think of morphisms
and factorizations as \emph{two types of arrows} in a larger structure, say of
morphisms as \emph{vertical} arrows and of factorizations as \emph{horizontal} ones.

The factorization and modular covariance theorems assign to any connected sequence
of vertical and horizontal arrows, say from $\ws$ to $\ws'$, a way of expressing
the correlator $C(\ws)$ in terms of $C(\ws')$, generically involving a summation
over several target world sheets $\ws'$ that differ only in their decoration data. By
employing the equivalence relations from Section \ref{sec:eqWS} we may furthermore
replace $C(\ws')$ with $C(\ws'')$ for a suitable
equivalent world sheet $\ws''$.\,%
 \footnote{~The types of manipulations (vertical and horizontal arrows) that
 are relevant here should have analogues outside the realm of rational CFT. However,
 some of the specific equivalences that we use, and as a consequence the set
 $\cS$ of fundamental world sheets at which we arrive, depend strongly on the
 properties of the representation category of the chiral algebra. In
 particular, the notion of fundamental world sheets as discussed here is not a
 purely topological one, and thus cannot be immediately relevant for the study
 of generators and relations in open/closed topological field theory with defects.}

Putting this information together,
we can now finally introduce a notion of \emph{fundamental world sheets}, and
correspondingly of \emph{fundamental correlators}. The purpose is to identify a
collection $\cS$ of (fundamental) world sheets of low genus, such that any arbitrary
correlator can be obtained from a disjoint union of world sheets in $\cS$ using
gluing and mapping class group covariance according to Theorems \ref{thm:bulkfac},
\ref{thm:bdfact}, and \ref{thm:modcov}. Making heavy use of the equivalences
discussed in Section \ref{sec:eqWS}, we will be able to end up with a
\emph{finite} set $\cS$ of fundamental correlators.

\medskip

In the absence of defect lines there is a finite set $\wcS$ of fundamental world
sheets of genus $0$ such that any world sheet can be connected, via a sequence of
vertical and horizontal arrows, to the disjoint union of a finite collection of
elements of $\wcS$ \cite[Sect.\,10]{fuRs10}. For oriented CFT the set $\wcS$
can be taken to consist of the following world sheets:

~\\[-2.85em]\begin{itemize}\addtolength{\itemsep}{-6pt}%
\item[\nxt]
      three bulk fields on the sphere;
\item[\nxt]
      three boundary fields on the disk;
\item[\nxt]
      one bulk and one boundary field on the disk.
\end{itemize}

In the following discussion we will include also separately \emph{one}-point functions,
even though they can be obtained as special cases of the correlators in $\wcS$,
namely by taking some of the insertions to be identity fields.
One-point correlators on the disk or on the sphere take values in (at most)
one-dimensional spaces of conformal blocks, irrespective of whether the world
sheet contains defects or not. Thus in particular such a correlator on a world
sheet \emph{with} defects is proportional to the corresponding correlator
\emph{without} defects.
        We will find that the constant of proportionality
        can be determined algorithmically (see Remark \ref{rem_algo} below).

In the rest of this subsection we establish one possible choice for the set $\cS$.

\begin{thm}\label{thm:cS}
The set $\cS$ of fundamental world sheets with defects can be taken to consist of
      ~\\[-1.75em]\begin{itemize}\addtolength{\itemsep}{-6pt}%
 \item[\nxt]
      three boundary fields on the disk $($without non-trivial defect lines$)$;
 \item[\nxt]
      three defect fields on a circular configuration of defect lines on the sphere;
 \item[\nxt]
      one boundary field and one disorder field on the disk.
  \end{itemize}
\end{thm}

\noindent
In particular, there are no new one-point correlators and no new correlators with
three boundary fields on the disk.
The proof of Theorem \ref{thm:cS} will occupy the rest of this section.

\smallskip

\begin{proof}
(i)\,
We observe that by invoking the results about bulk and boundary factorization
we can take as candidates for the world sheets in $\cS$ those in $\wcS$, with
insertions corresponding to vertices in $\P$, but complemented in all possible
ways by (embeddings $\assd$ of) arbitrary defect graphs $\D$. Not surprisingly,
the so obtained collection
$\widehat{\cS}$ is hugely redundant. Let us thus see how we may reduce
$\widehat{\cS}$ by making use of equivalent decorations of world sheets.
 \\
First note that with the help of defect and boundary fusion, followed by
removal of internal and boundary \fusion\ vertices, it is enough to consider
defect graphs $\D$ for which any two vertices are connected by either
$0$, $1$ or $2$ edges. By removing boundary \fusion\ vertices it is furthermore
enough to consider graphs having at most one boundary \fusion\ vertex
between any two boundary insertions. In addition, by defect fusion and
removing internal \fusion\ vertices,
any boundary \fusion\ vertex can be assumed to be trivalent.
\\
Another fact of which we will make frequent use below is the following.
Consider an edge $e$ that starts and ends at the same vertex, forming a loop
that can be contracted to the vertex without crossing other parts of the
defect graph. By first using local fusion and then collapsing the
defect bubble, the loop can be replaced by a \fusion\ vertex; in pictures,
  \Eqpic{remloop}{440}{18}{
  \put(0,0)     {\Includepicfj{25}{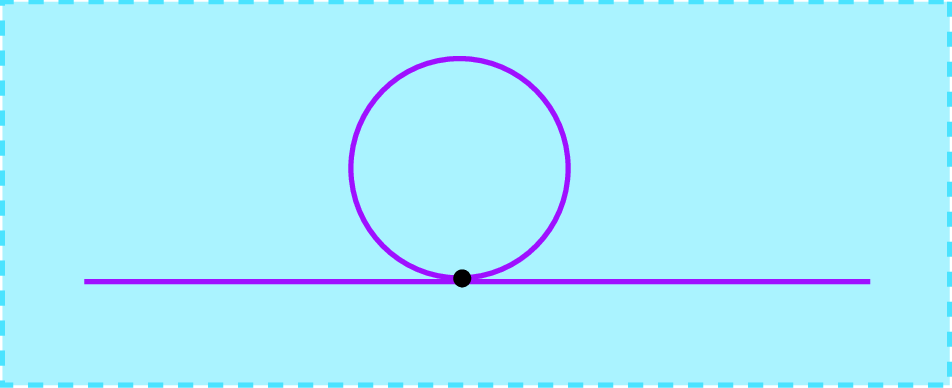}}
  \put(130,22)	{\footnotesize $\longmapsto$}
  \put(160,0)	{\Includepicfj{25}{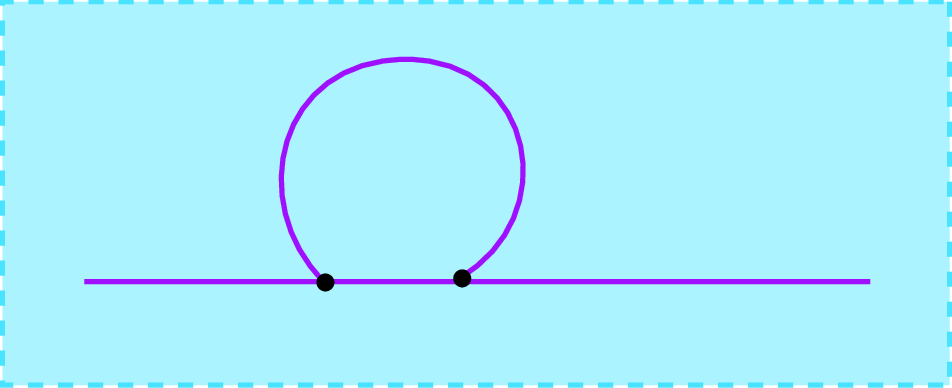}}
  \put(290,22)	{\footnotesize $\longmapsto$}
  \put(320,0)	{\Includepicfj{25}{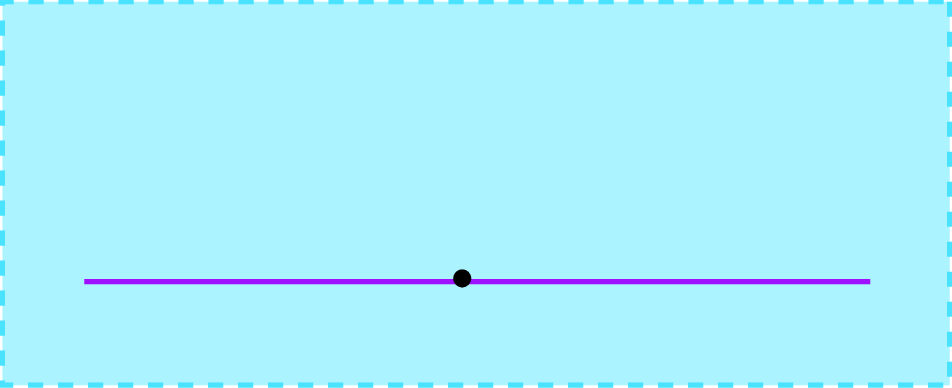}}
  }
(ii)\, We now treat, one at a time, the various world sheets in $\wcS$,
complemented with arbitrary defect graphs without additional insertion vertices.
In the pictures below a dotted line indicates a transparent edge, while a
solid line indicates an edge with arbitrary decoration. A \fusion\ vertex is
indicated by a fat dot, while an insertion point is drawn as a dot located on
a small arrow. A shaded region with
brighter hue than the background contains an arbitrary embedded graph without
insertion points; we refer to such parts of the world sheet as \emph{\nregion s}.

%%%%%%%%%%%%%%%%%%%%%%%%%%%%%%%%%%%%%%%%%%%%%

          \pagebreak

\iline{One boundary field on the disk}

\noindent
Consider a world sheet that is a disk with one boundary field insertion and
with a defect graph $\D$ having an arbitrary number of \fusion\ vertices.
Since there is only one insertion on the boundary it suffices, as explained
in part (i) of the proof, to consider graphs with a single boundary \fusion\
vertex. As $\D$ is connected, we can
then further restrict to a graph with a single trivalent \fusion\ vertex in the
interior. Afterwards, removing the interior \fusion\ vertex leaves a loop attached
to the boundary. Finally, discarding this loop leaves a boundary \fusion\ vertex, which in turn can be removed at the cost of modifying the boundary insertion.
In pictures, this sequence of manipulations looks as follows:
  \eqpic{Disk1pt}{420}{110}{ \put(0,2){
  \put(0,135)   {\Includepicfj{25}{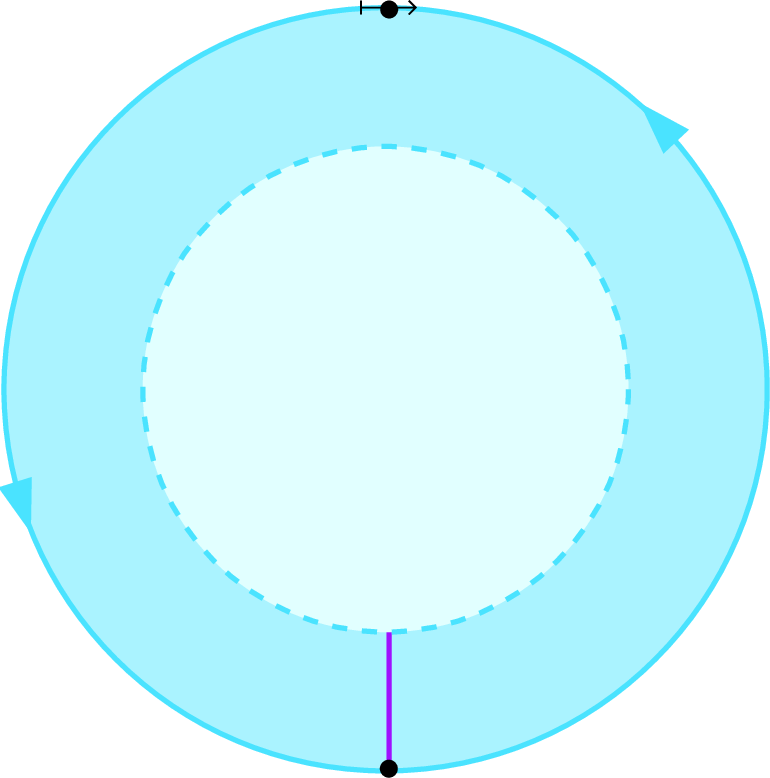}}
  \put(45,233)	{\footnotesize $\psi$}
  \put(45,130)	{\footnotesize $\chi$}
  \put(112,179)	{$\longmapsto$}
  \put(10,0){
  \put(140,135)	{\Includepicfj{25}{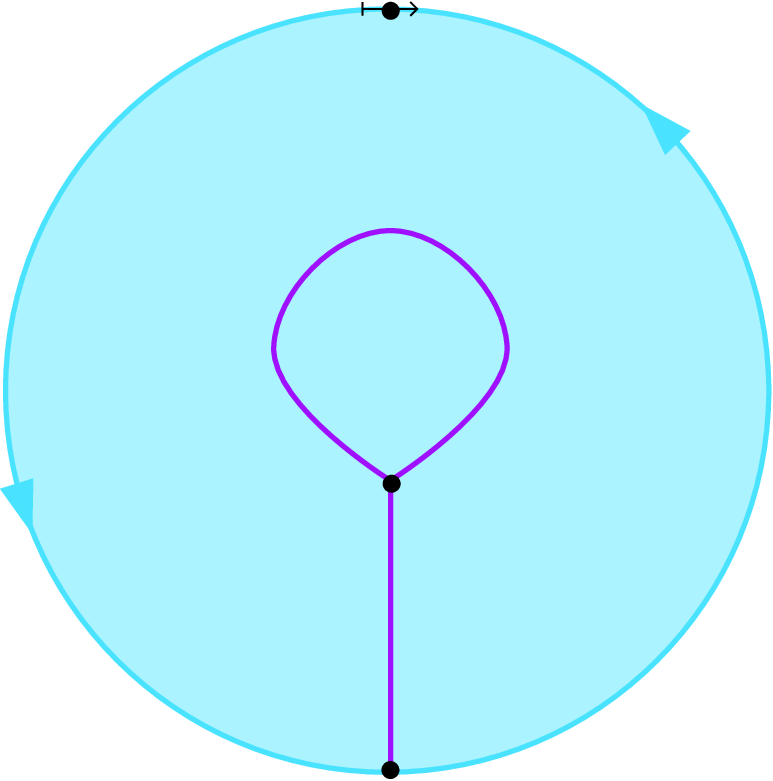}}
  \put(185,233)	{\footnotesize $\psi$}
  \put(185,130)	{\footnotesize $\chi$}
  }
  \put(262,179)	{$\longmapsto$}
  \put(20,0){
  \put(280,135)	{\Includepicfj{25}{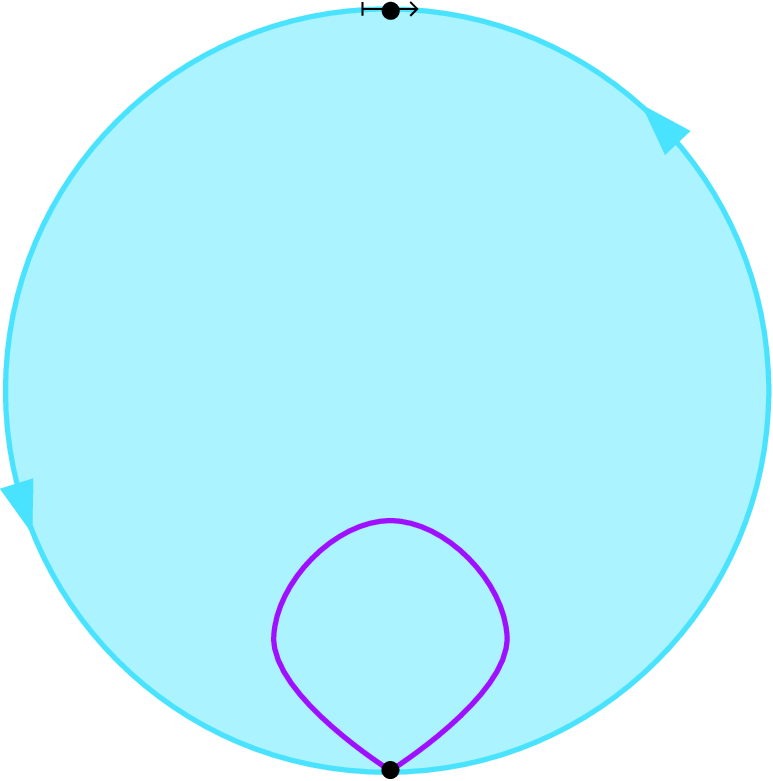}}
  \put(325,233)	{\footnotesize $\psi$}
  \put(325,130)	{\footnotesize $\chi$}
  }
    }
  \put(112,44)	{$\longmapsto$}
  \put(150,0){
  \put(0,5)	{\Includepicfj{25}{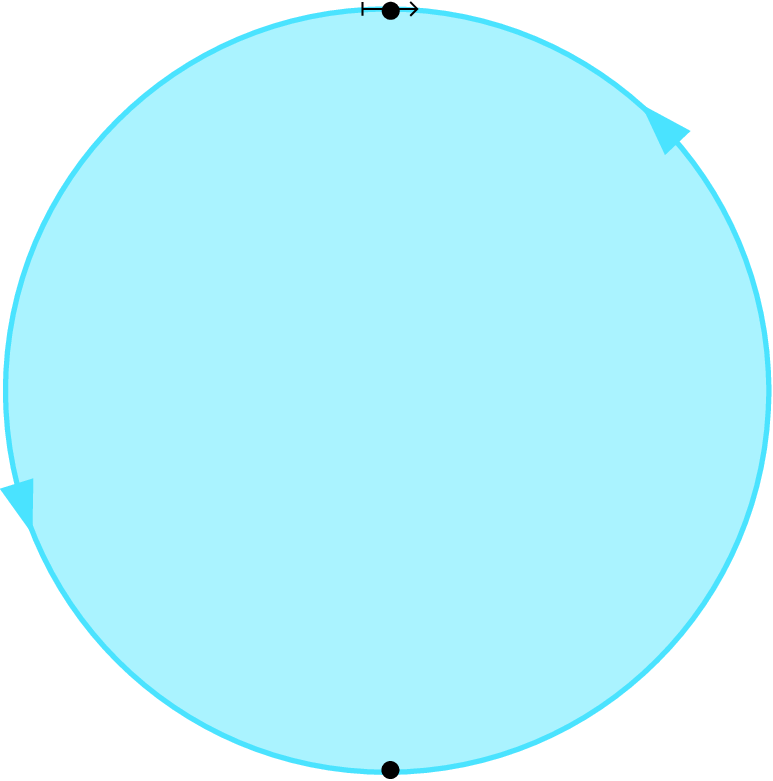}}
  \put(45,103)	{\footnotesize $\psi$}
  \put(45,-1){\footnotesize $\chi'$}
  }
  \put(262,44)	{$\longmapsto$}
  \put(160,0){
  \put(140,5)	{\Includepicfj{25}{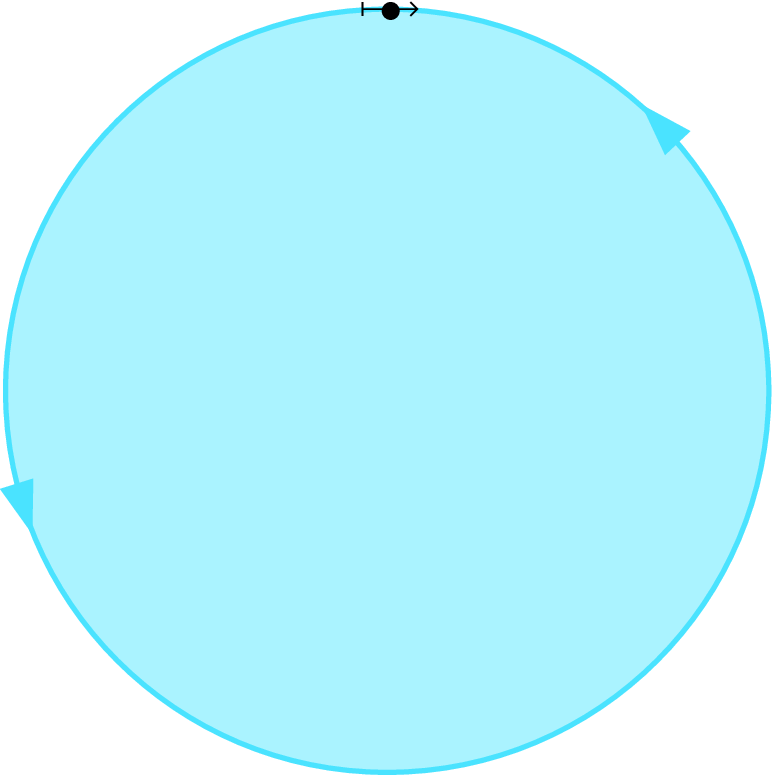}}
  \put(185,103)	{\footnotesize $\psi'$}
  } }
In conclusion, every one-point boundary correlator with defects on the disk is given by
some one-point boundary correlator without any defects.

%%%%%%%%%%%%%%%%%%%%%%%%%%%%%%%%%%%%%%%%%%%%%

\iline{One field on the sphere}

\noindent
Next take a world sheet $\ws$ that is a sphere, with defect graph $\D$ being
arbitrary apart from the constraint $|\P|=1$. By applying a suitable isotopy
it can be assumed that the image $\assd(\D)$ is contained in a
contractible region of $\ws$. Furthermore, one may slide a planar graph around
the sphere, such that one deals with the situation shown in the second picture of
\eqref{Sphere1pt} below. We may, without loss of generality, also assume that the
graph contained in the \nregion\ is connected (if necessary, by replacing it with
an equivalent graph for which
appropriate edges of the disconnected components have been fused). It is therefore equivalent to a graph with a single vertex. After removing all loops we then
end up with a graph having a single two-valent vertex, as shown in the third
picture in \eqref{Sphere1pt}.
At this point we may introduce, by invoking insertion of transparent tadpoles and
independence of transparent subgraph, two new structure vertices, one on each edge,
and connected by a new transparent edge; the resulting situation is shown in the
fourth picture. Finally, the so obtained bubble containing the insertion vertex can
be replaced by a new insertion vertex. Altogether we thus have:
  \eqpic{Sphere1pt}{440}{111}{ \put(0,104){
  \put(0,25) {\Includepicfj{30}{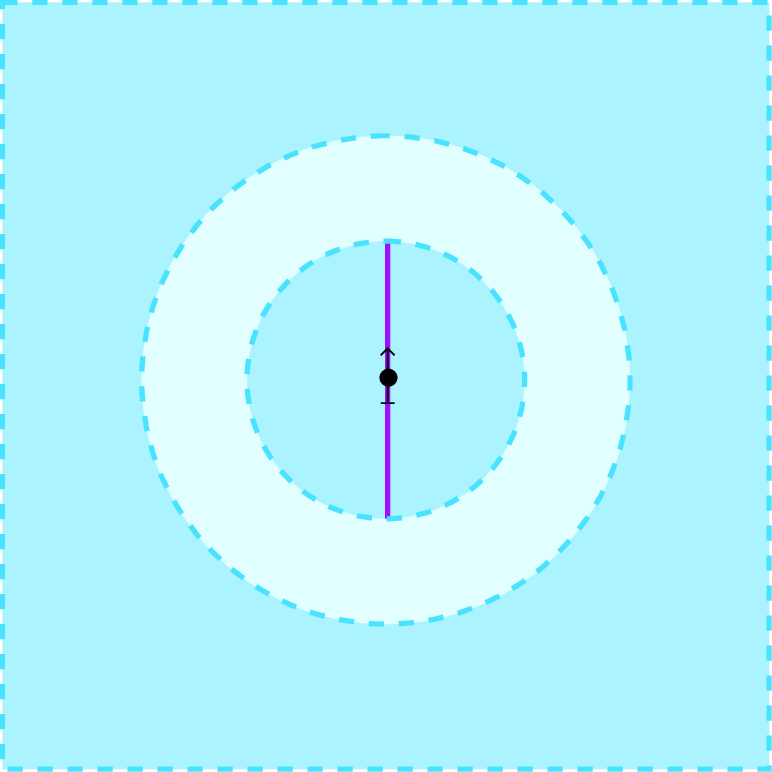}}
  \put(45,80)	{\footnotesize $\phi$}
  \put(125,79)  {\footnotesize $\longmapsto$}
  \put(155,25){\Includepicfj{30}{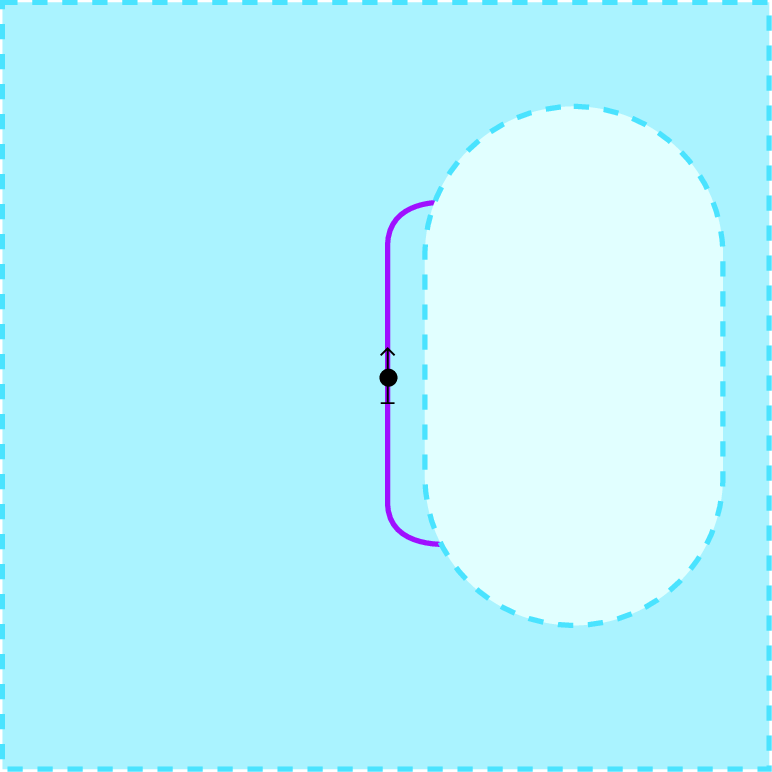}}
  \put(200,80)	{\footnotesize $\phi$}
  \put(280,79)  {\footnotesize $\longmapsto$}
  \put(310,25){\Includepicfj{30}{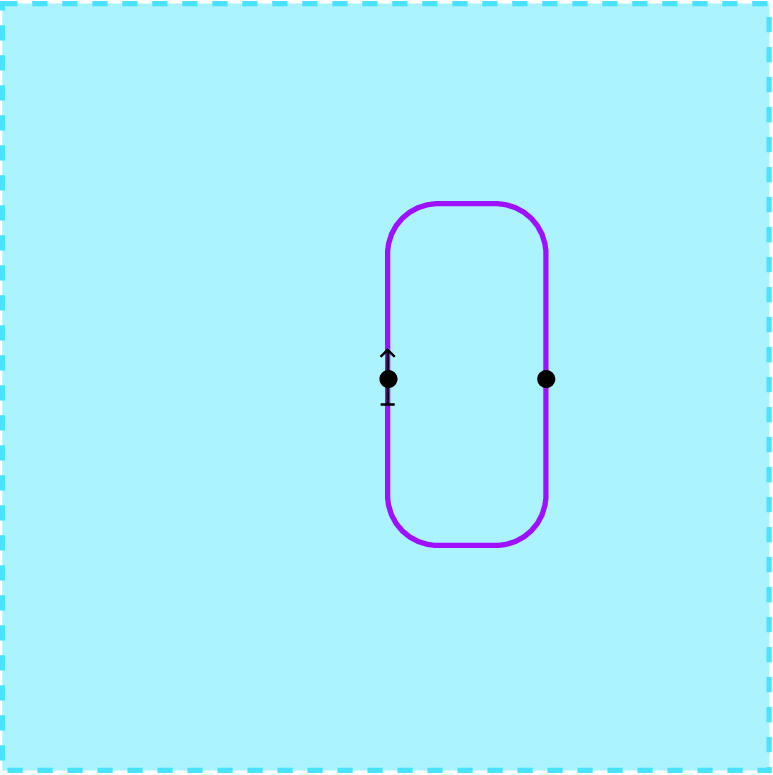}}
  \put(355,80)	{\footnotesize $\phi$}
  \put(390,80)	{\footnotesize $\varkappa$}
     }
  \put(125,54)  {\footnotesize $\longmapsto$}
  \put(155,0){\Includepicfj{30}{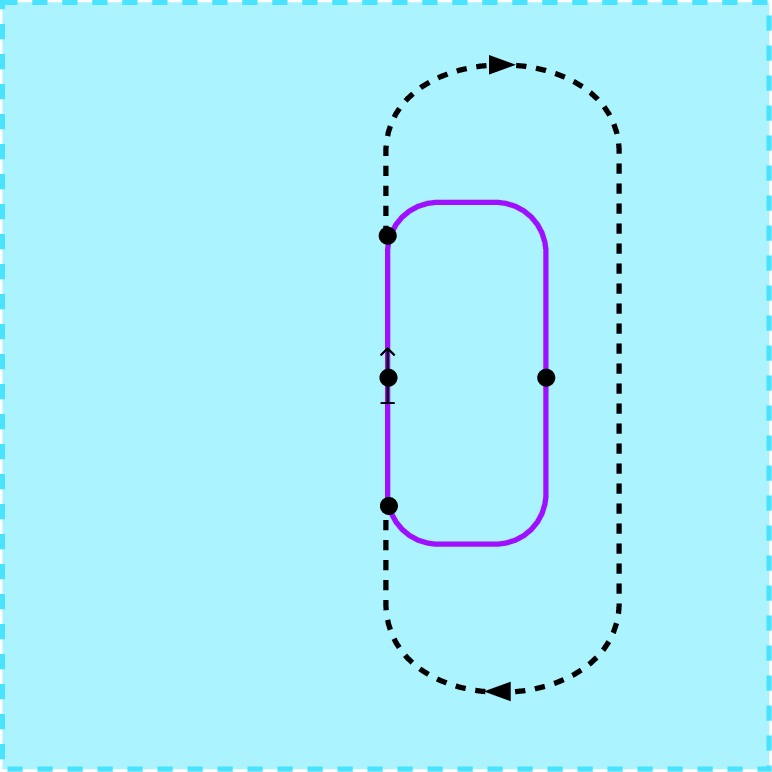}}
  \put(200,55)	{\footnotesize $\phi$}
  \put(235,55)	{\footnotesize $\varkappa$}
  \put(280,54)  {\footnotesize $\longmapsto$}
  \put(310,0){\Includepicfj{30}{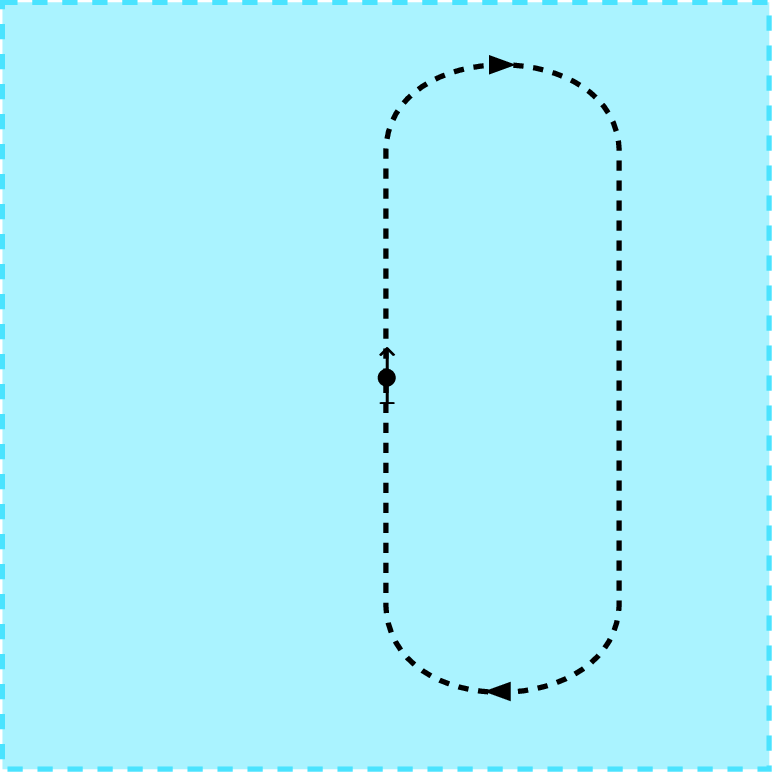}}
  \put(355,55)	{\footnotesize $\phi'$}
  }
In conclusion, every one-point correlator with defects on the sphere is given
by a one-point correlator on the sphere without defects.

%%%%%%%%%%%%%%%%%%%%%%%%%%%%%%%%%%%%%%%%%%%%%

\iline{Three boundary fields on the disk}

\noindent
As already mentioned we can restrict to one trivalent \fusion\ vertex between
each pair of boundary insertions. Thus for a disk with three boundary fields
we may without loss of generality assume that there are (at most) three \fusion\
vertices on the boundary, each connected by a single edge to an arbitrary planar
graph in the \nregion\ of the disk, as indicated in the first picture of
\eqref{Disk3pt} below. Using, if necessary, further fusion of defects we
may also assume that the interior graph is itself connected. Removing \fusion\
vertices as well as loops, we end up with a single trivalent \fusion\ vertex in
the interior that is directly connected with each boundary \fusion\ vertex, as
shown in the second picture. The interior \fusion\ vertex may then be removed as
well, resulting in the world sheet shown in the third picture. Further use of
fusion and of removal of \fusion\ vertices relates the latter world sheet to one
for which two of the boundary insertions are located on a boundary defect bubble
(fourth picture). After removing those bubbles no internal
edges are left at all and we end up with three boundary fields on the disk:
  \eqpic{Disk3pt}{430}{104}{ \put(0,100){
  \put(0,25) {\Includepicfj{25}{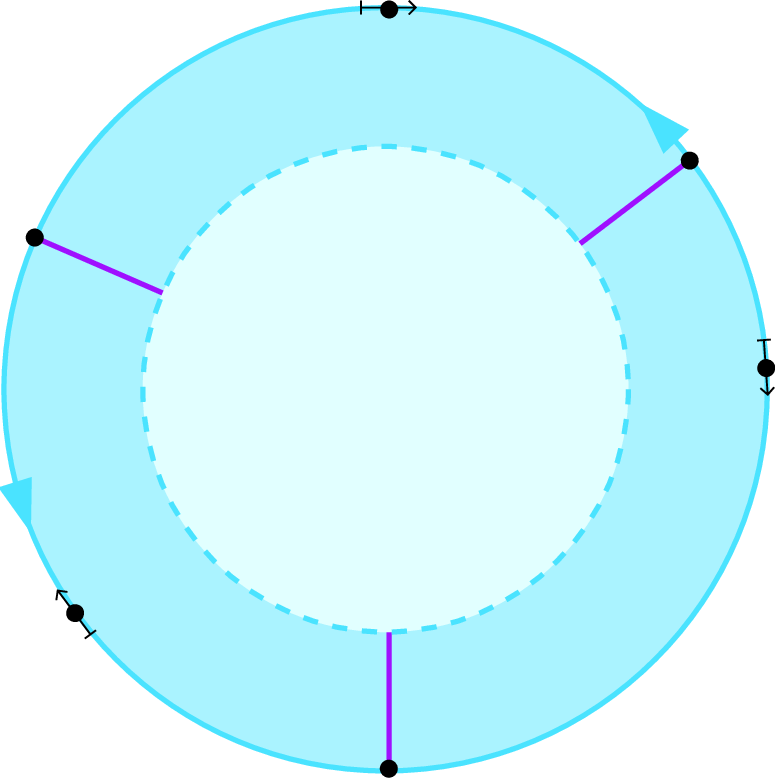}}
  \put(45,123)	{\footnotesize $\psi_1$}
  \put(95,73)	{\footnotesize $\psi_2$}
  \put(0,35)	{\footnotesize $\psi_3$}
  \put(85,100)	{\footnotesize $\chi_1^{}$}
  \put(45,20)	{\footnotesize $\chi_2^{}$}
  \put(-5,95)	{\footnotesize $\chi_3^{}$}
  \put(121,71)  {\footnotesize $\longmapsto$}
      } \put(120,100){
  \put(40,25)	{\Includepicfj{25}{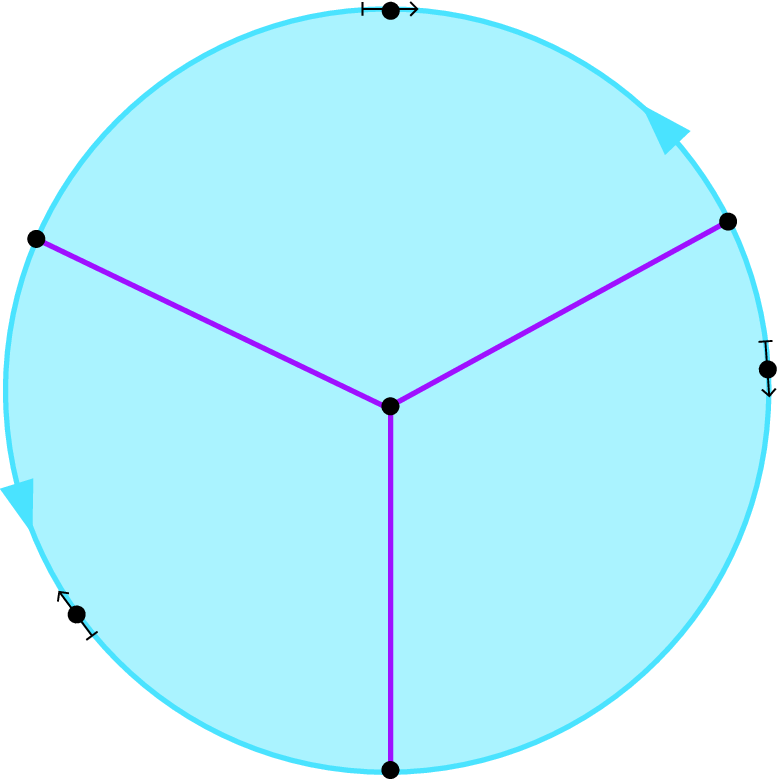}}
  \put(85,123)	{\footnotesize $\psi_1$}
  \put(135,73)	{\footnotesize $\psi_2$}
  \put(40,35)	{\footnotesize $\psi_3$}
  \put(125,100)	{\footnotesize $\chi_1^{}$}
  \put(85,20)	{\footnotesize $\chi_2^{}$}
  \put(35,95)	{\footnotesize $\chi_3^{}$}
  \put(83,75)	{\footnotesize $\varkappa$}
  \put(161,71)  {\footnotesize $\longmapsto$}
      } \put(240,100){
  \put(80,25)	{\Includepicfj{25}{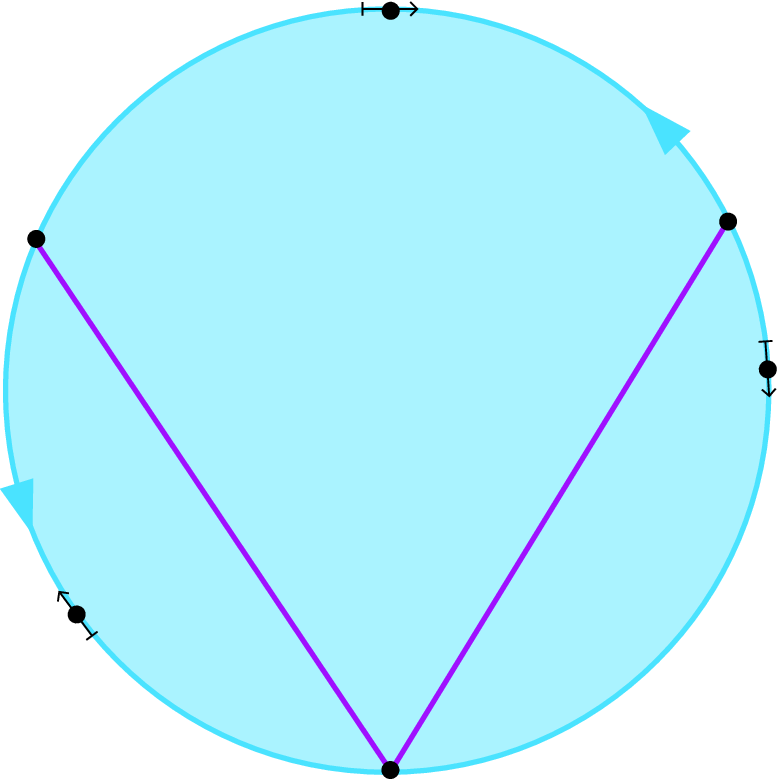}}
  \put(125,123)	{\footnotesize $\psi_1$}
  \put(175,73)	{\footnotesize $\psi_2$}
  \put(80,35)	{\footnotesize $\psi_3$}
  \put(165,100)	{\footnotesize $\chi_1^{}$}
  \put(125,20)	{\footnotesize $\chi_2'$}
  \put(75,95)	{\footnotesize $\chi_3^{}$}
      } \put(160,0){
  \put(-39,46)  {\footnotesize $\longmapsto$}
  \put(0,0)	{\Includepicfj{25}{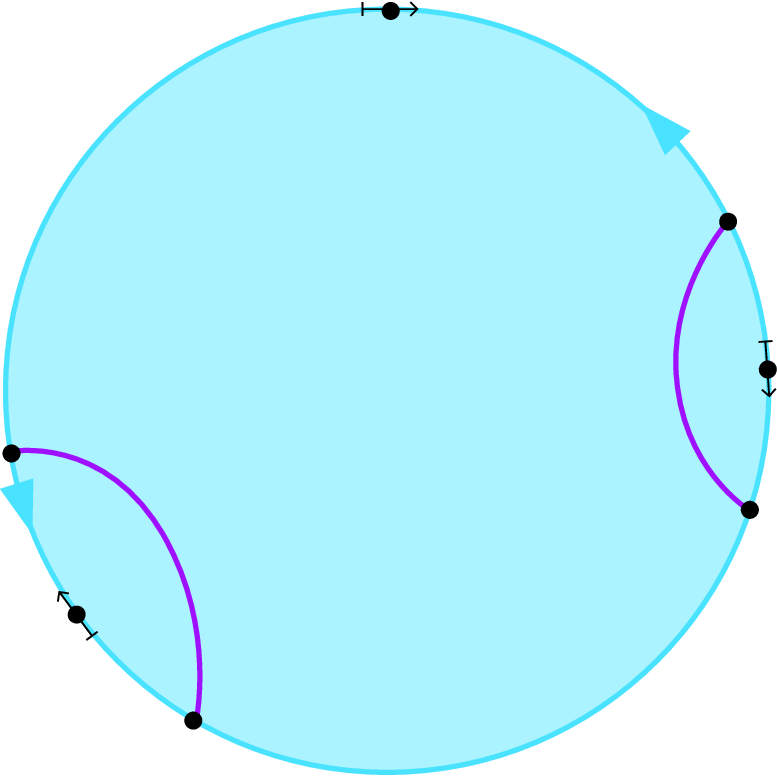}}
  \put(45,98)	{\footnotesize $\psi_1$}
  \put(95,48)	{\footnotesize $\psi_2$}
  \put(0,10)	{\footnotesize $\psi_3$}
  \put(90,70)	{\footnotesize $\chi_1^{}$}
  \put(93,30)	{\footnotesize $\chi_2''$}
  \put(19,0)	{\footnotesize $\chi_3'$}
  \put(-10,35)	{\footnotesize $\chi_4$}
  \put(124,46)  {\footnotesize $\longmapsto$}
      } \put(280,0){
  \put(40,0)	{\Includepicfj{25}{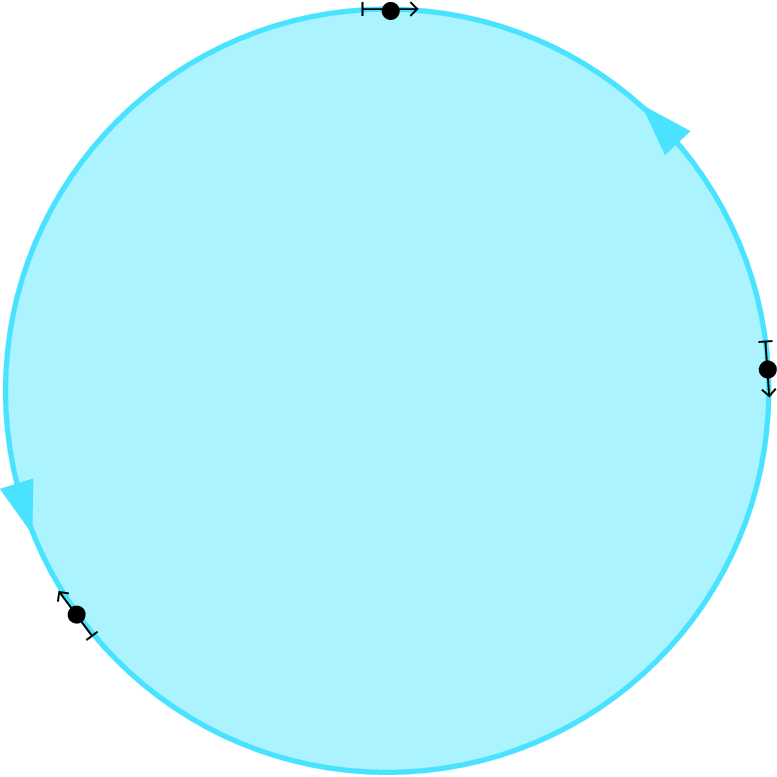}}
  \put(85,98)	{\footnotesize $\psi_1$}
  \put(135,48)	{\footnotesize $\psi_2'$}
  \put(40,10)	{\footnotesize $\psi_3'$}
  } }
We thus see that every correlator of three boundary fields on the disk with defects
is given by a correlator of three boundary fields on the disk without defects.

%%%%%%%%%%%%%%%%%%%%%%%%%%%%%%%%%%%%%%%%%%%%%

\iline{One boundary field and one disorder field on the disk}

\noindent
Next consider the disk with one boundary field and one bulk field. Again we may
without loss of generality assume that there is exactly one trivalent
\fusion\ vertex on the boundary. A world sheet $\ws$ of this type is depicted
on the left side of the picture \eqref{Disk21pt1} below. The annular \nregion\
of this world sheet $\ws$ contains an arbitrary defect graph without insertion
vertices. Fusing all parallel edges of that graph that cross a given line segment
(say, the one indicated in the picture by a dashed line)
in the \nregion\ results in the picture in the middle part of \eqref{Disk21pt1}.
Fusing further the two parallel edges and replacing the resulting bubble with a
new field insertion then results in the figure to the right.
  \eqpic{Disk21pt1}{430}{39}{ \put(0,-5){ \put(0,-5){
  \put(10,7)	{\Includepicfj{25}{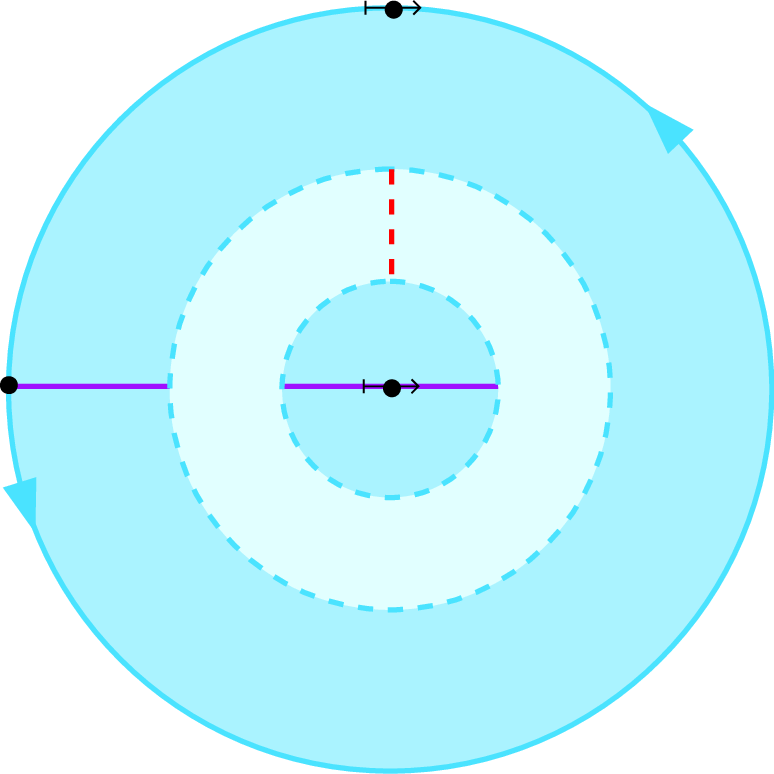}}
  \put(55,105)	{\footnotesize $\psi$}
  \put(0,53)	{\footnotesize $\chi$}
  \put(55,59)	{\footnotesize $\phi$}
  \put(123,51)  {\footnotesize $\longmapsto$}
  \put(165,7)	{\Includepicfj{25}{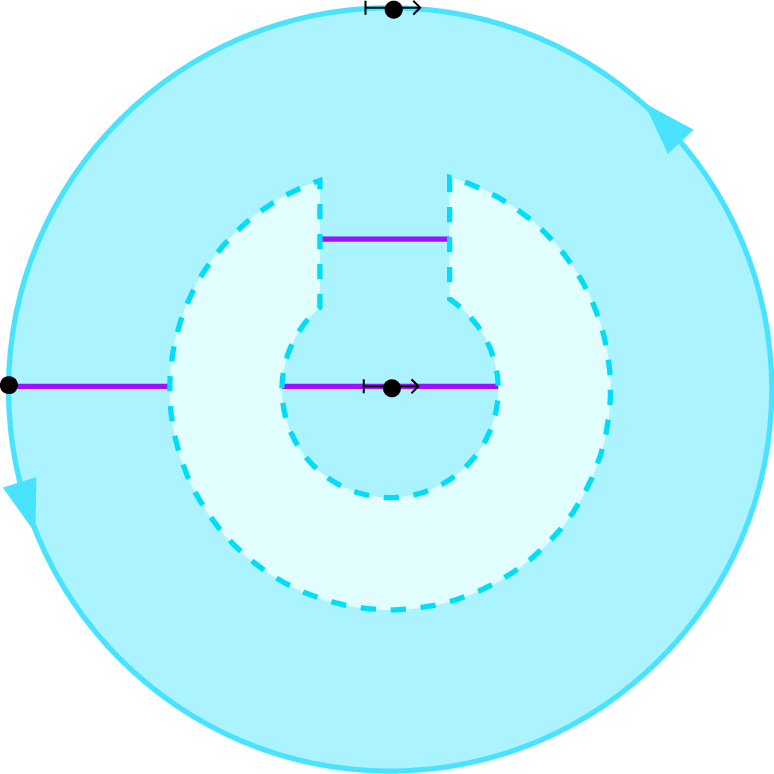}}
  \put(210,105)	{\footnotesize $\psi$}
  \put(155,53)	{\footnotesize $\chi$}
  \put(210,60)	{\footnotesize $\phi$}
  \put(279,51)  {\footnotesize $\longmapsto$}
  \put(320,7)	{\Includepicfj{25}{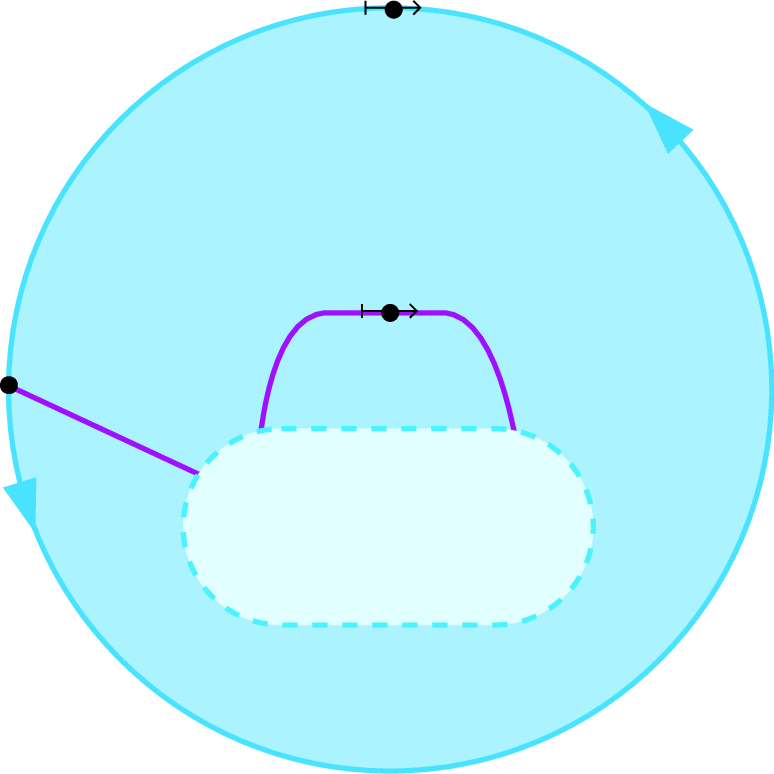}}
  \put(365,105)	{\footnotesize $\psi$}
  \put(310,53)	{\footnotesize $\chi$}
  \put(365,70)	{\footnotesize $\phi'$}
  } } }
Yet again we may assume that the graph contained in the \nregion\
is connected and thus can be replaced with
a graph having a single vertex. Removing loops and inserting a transparent edge
we then arrive at the world sheet shown on the left side of \eqref{Disk21pt2}.
Finally, collapsing a defect bubble we arrive at a world sheet with one boundary
field and one disorder field:
  \eqpic{Disk21pt2}{274}{42}{
  \put(10,7)	{\Includepicfj{25}{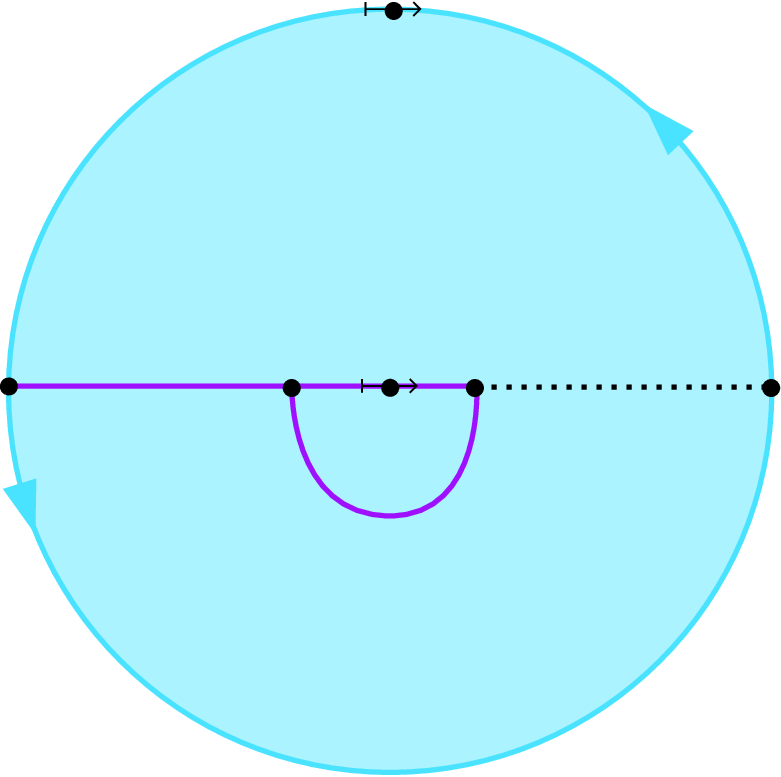}}
  \put(55,105)	{\footnotesize $\psi$}
  \put(0,53)	{\footnotesize $\chi$}
  \put(55,60)	{\footnotesize $\phi'$}
  \put(40,60)	{\footnotesize $\varkappa$}
  \put(139,51)  {\footnotesize $\longmapsto$}
  \put(190,7)	{\Includepicfj{25}{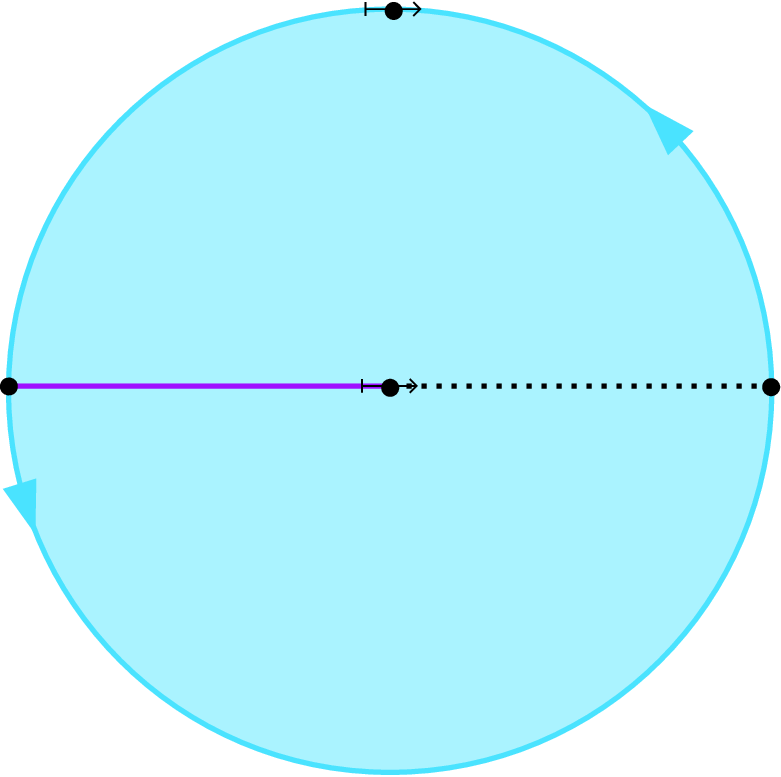}}
  \put(235,105)	{\footnotesize $\psi$}
  \put(181,53)	{\footnotesize $\chi$}
  \put(235,60)	{\footnotesize $\phi''$}
  }
%%%%%%%%%%%%%%%%%%%%%%%%%%%%%%%%%%%%%%%%%%%%%

\iline{Three fields on the sphere}

\vskip .4em

\noindent
We finally consider a correlator on the sphere with three field insertions.
This situation could be analyzed in a way very similar to the previous one.
Instead we follow a different, shorter, route.

Select a neighborhood of each insertion point, each containing a single
vertex and (segments of) the incident edges, along with a contractible region
intersecting neither $\assd(\D)$ nor any of
the three neighborhoods. The complement of the union of these domains
is the \nregion\ on the left hand side of \eqref{Sphere3pt1} below, which
contains the rest of $\assd(\D)$. Again we can assume that the latter part
of the graph is connected (otherwise we replace the world sheet with an
equivalent one for which edges in disconnected components have been fused).

By further using multiple fusions and removing bubbles and loops we end up
with the graph in the middle picture of \eqref{Sphere3pt1}. By removing
internal vertices, as well as
removing bubbles and loops, we end up with the picture on the right:
  \eqpic{Sphere3pt1}{440}{59}{ \put(0,-3){
  \put(0,0)     {\Includepicfj{30}{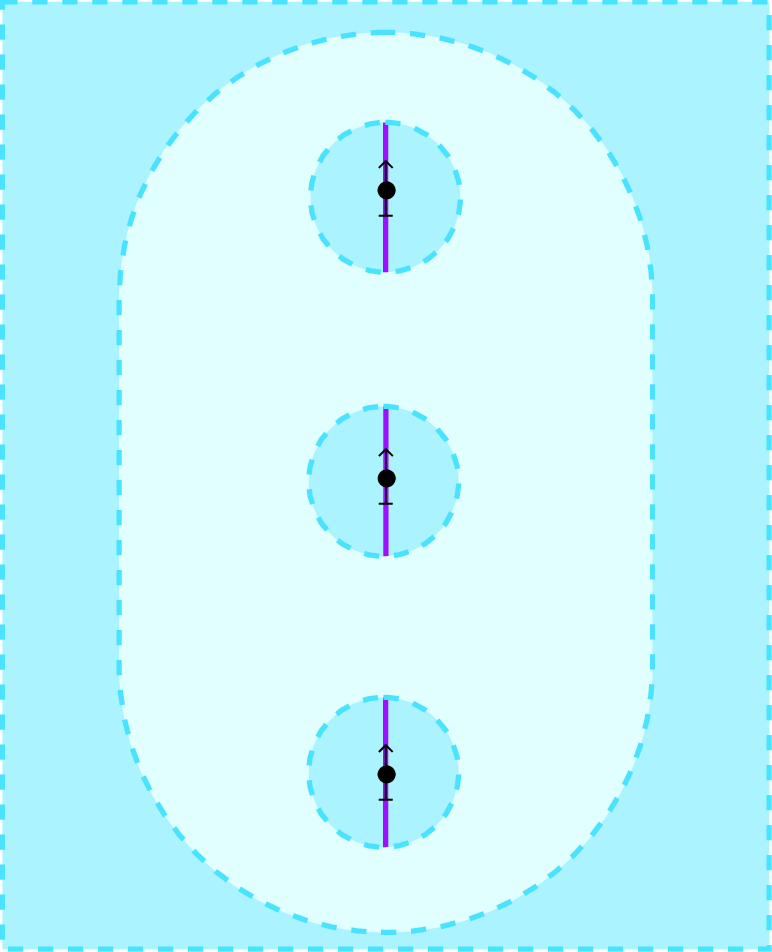}}
  \put(45,105)	{\footnotesize $\phi_1$}
  \put(45,62)	{\footnotesize $\phi_2$}
  \put(45,21)	{\footnotesize $\phi_3$}
  \put(126,66)  {\footnotesize $\longmapsto$}
   \put(5,0){
  \put(150,0) {\Includepicfj{30}{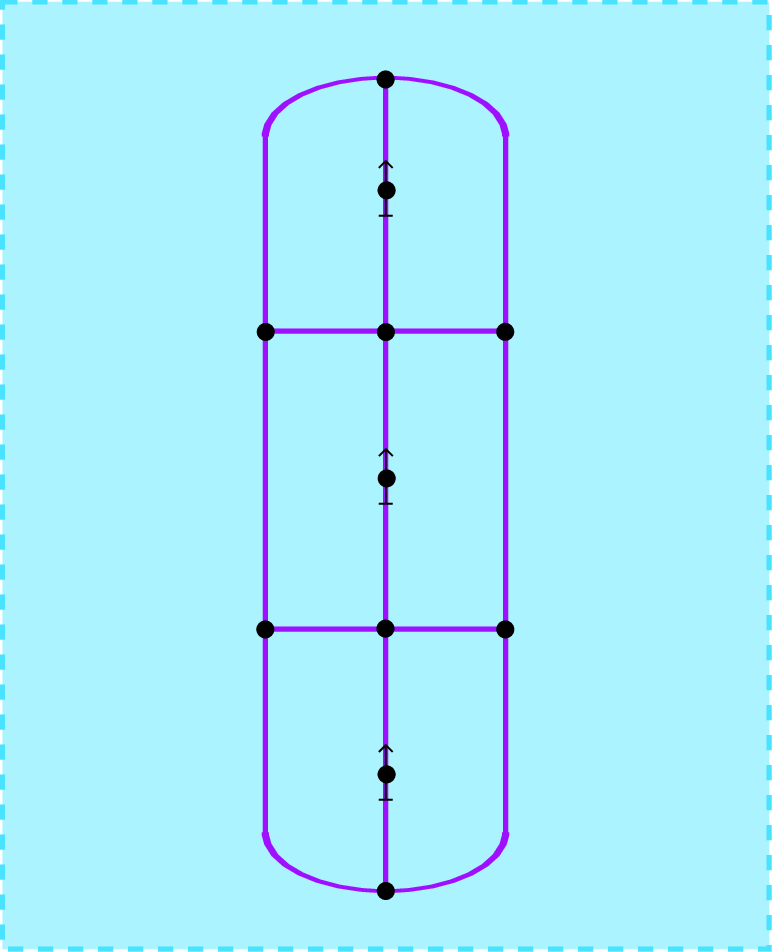}}
  \put(193,110)	{\footnotesize $\phi_1$}
  \put(193,65)	{\footnotesize $\phi_2$}
  \put(193,22)	{\footnotesize $\phi_3$}
  \put(201,2)   {\footnotesize $\varkappa_1$}
  \put(175,46)  {\footnotesize $\varkappa_2$}
  \put(207,40)  {\footnotesize $\varkappa_3$}
  \put(225,44)  {\footnotesize $\varkappa_4$}
  \put(175,89)  {\footnotesize $\varkappa_5$}
  \put(207,83)  {\footnotesize $\varkappa_6$}
  \put(225,87)  {\footnotesize $\varkappa_7$}
  \put(201,130) {\footnotesize $\varkappa_8$}
  }
  \put(280,66)  {\footnotesize $\longmapsto$}
   \put(10,0){
  \put(300,0) {\Includepicfj{30}{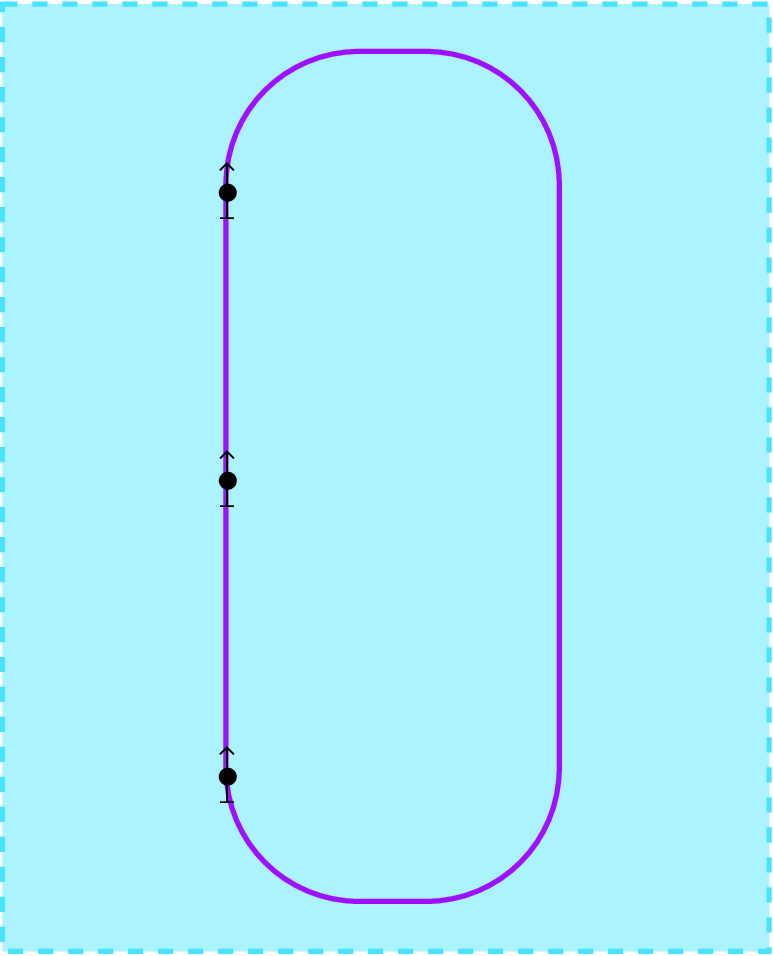}}
  \put(320,109)	{\footnotesize $\phi_1'$}
  \put(320,67)	{\footnotesize $\phi_2'$}
  \put(320,25)	{\footnotesize $\phi_3'$}
  } } }
This concludes the proof of Theorem \ref{thm:cS}.\endofproof
\end{proof}

\vskip .2em

\begin{remark}\label{rem_algo}
In the absence of non-trivial defects, for any sequence of arrows relating
a world sheet to the fundamental world sheets in $\wcS$ the factorization
and covariance formulas give the precise relation between the corresponding
correlators. One may worry that by invoking equivalences, this
algorithmic relation is lost. Comparing the description of equivalences
in Section \ref{sec:eqWS}, as well as the related comments in Remark
\ref{rem:equiv_moves}, with the use of the equivalences in the proof of
Theorem \ref{thm:cS} one can, however, convince oneself that any replacement
of a world sheet $\ws$ with an equivalent $\ws'$ is accompanied by a precise
prescription for the decoration data of $\ws'$ in terms of the decoration
data of $\ws$. In this way, any sequence of arrows \emph{and equivalences}
relating a world sheet with defects to the fundamental world sheets in $\cS$
provides again the precise relation between the corresponding correlators.
\end{remark}

\newpage
%%%%%%%%%%%%%%%%%%%%%%%%%%%%%%%%%%%%%%%%%%%%%%%%%%%%%%%%%%%%%%%%%%%%%%%%%
\appendix
\section{Appendix}

\subsection{Independence of transparent subgraph}\label{app:iots}

As an illustration of the various equivalences of world sheets presented in
Section \ref{sec:eqWS} we establish the \emph{independence of transparent
subgraph}, which is asserted after formula \eqref{remstrvert}.

First note that this assertion is the analogue of what in the absence of
non-trivial defects is the independence of the correlators from
the choice of triangulation of the world sheet.
The latter, which is shown in Proposition 3.2 of \cite{fjfrs}, is the basic ingredient of
the proof of covariance of the correlators in the defect-free case. To
arrive at the present generalization we first establish some basic properties
of morphisms that are composed of the structural morphisms of Frobenius
algebras. More specifically, the relevant morphisms are \emph{connected}
morphisms, that is, in the graphical notation they are described by connected
graphs.

\begin{lemma}\label{frob_proof_lemma}
Let $A$ be a symmetric special Frobenius algebra. Any connected morphism $\phi$ in
$\Hom(A^{\oti p},A^{\oti q})$ with $p,q \,{\geq}\, 1$ that is entirely composed
of the structural morphisms $m,\,\eta,\,\Delta$ and $\eps$ of $A$ can be written as
  \be\label{frob_p1}
  \phi = \psi\circ\varphi
  \ee
with
  \eqpic{frob_p5}{360}{38}{
  \put(4,40)  {$\varphi:=$}
  \put(50,0){
  \put(10,10)	{\Includepicfj{35}{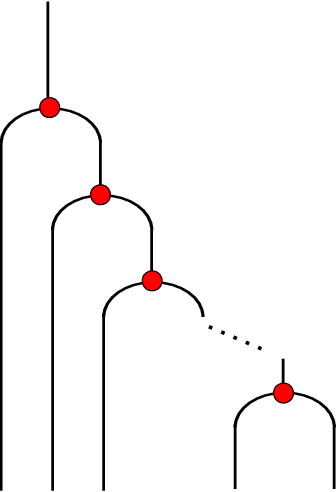}}
  \put(5,0)	{\footnotesize $A$}
  \put(14,0)	{\footnotesize $A$}
  \put(23,0)	{\footnotesize $A$}
  \put(32,0)	{\footnotesize $\cdots$}
  \put(45,0)	{\footnotesize $A$}
  \put(61,0)	{\footnotesize $A$}
  \put(14,98)	{\footnotesize $A$}
  }
  \put(170,40)    {and}
  \put(236,40)  {$\psi:=$}
  \put(282,0){
  \put(10,0)	{\Includepicfj{35}{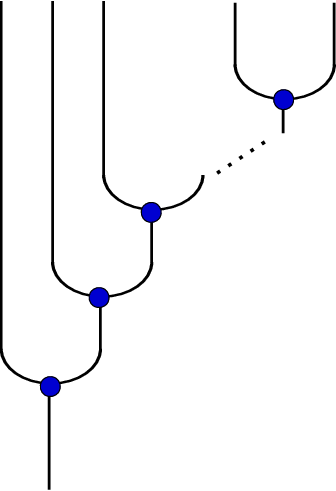}}
  \put(5,89)	{\footnotesize $A$}
  \put(14,89)	{\footnotesize $A$}
  \put(23,89)	{\footnotesize $A$}
  \put(32,89)	{\footnotesize $\cdots$}
  \put(45,89)	{\footnotesize $A$}
  \put(61,89)	{\footnotesize $A$}
  \put(13,-11)	{\footnotesize $A$}
  } }
In particular, for each pair $p,q$ of positive integers there is a unique such morphism.
\end{lemma}

\proof
(i) For a connected morphism consisting only of a combination of products and coproducts,
we can use associativity, coassociativity and the Frobenius property to rearrange the
order of their composition arbitrarily. In particular, this way the ordering can be
arranged in such a way that all coproducts come after all products, except possibly for
the appearance of `bubbles' $m\cir\Delta$. The latter can be removed owing to specialness,
and afterwards again associativity and coassociativity can be invoked, such that one ends
up with the morphism \eqref{frob_p1}.
\\[2pt]
(ii) Next consider a connected morphism consisting of one unit morphism $\eta$ besides
products and coproducts. Then owing to $p\,{\ge}\,1$ there must be at least one
occurrence of the product $m$. By using associativity, coassociativity and the Frobenius
property in an analogous manner as in (i), the ordering of products and coproducts
can then be changed in such a way that the unit morphism is directly followed by
a product, i.e.\ appears in the form of either $m\cir(\id_A\,\oti\,\eta)$ or
$m\cir(\eta\,\oti\,\id_A)$; by the unit property, both of these combinations
equal $\id_A$. In short, $\eta$ can be removed from the expression for the morphism.
By iterating the argument, this applies in fact to any number of occurrences of
$\eta$. Hereby this case is reduced to case (i).
\\[2pt]
(iii) When also counit morphisms $\eps$ are present, the same type of argument as in (ii)
applies, just with the role of $m$ now being taken over by the coproduct $\Delta$,
and the unit property replaced by the defining property of the counit.
\endofproof

    A similar statement applies when instead of $A$-lines also lines labeled by the
    dual $A^\vee$ are admitted:

\begin{lemma}\label{lemma:eqAmorph}
Let $H$ be the morphism space $\Hom(A^\pm\oti A^\pm\oti\cdots\oti A^\pm,
A^\pm\oti A^\pm\oti\cdots\oti A^\pm)$, with a definite choice of either $A$ or $A^\vee$
for each of the tensor factors $A^\pm$. Any connected morphism $\phi\iN H$ that is
entirely composed of structure morphisms of $A$ and $A^\vee$ can be written as
  \be\label{frob_p4}
  \phi = (g^\pm\,\oti\, g^\pm\,\oti\cdots\oti\, g^\pm) \circ \psi \circ \varphi
  \circ (h^\pm\,\oti\, h^\pm\,\oti\cdots\oti\, h^\pm) \,,
  \ee
where $\psi$ and $\varphi$ are the morphisms \eqref{frob_p5}, while
$g^+ \eq \id_{\!A} \eq h^+$ as well as $g^- \eq \Fiso \iN \Hom(A,A^\vee)$
   $($as defined in \eqref{Fiso}$)$
and $h^- \eq \Fisoi \iN \Hom(A^\vee,A)$.
 \\
In particular, for any given numbers $p,\,q$ of factors $A^\pm$ in the source
and target, there is a unique such morphism $\phi$ in the morphism space $H$.
\end{lemma}

\proof
When composing the structure morphisms, all duality morphisms that are not directly
attached to an ingoing or outgoing line of the graph get composed in pairs.
Two types of compositions can occur: Either a right (left) evaluation and a
right (left) coevaluation are composed such that they yield a `zig-zag' morphism
which according to the defining properties of the dualities equals an identity
morphism. Or else  they form the combination shown on the left hand side of
the following equality (or its mirrored version):
  \eqpic{frob_p5}{70}{24}{ \put(0,-10){
  \put(0,13)	{\Includepicfj{35}{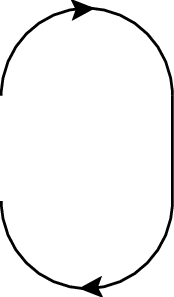}}
  \put(46,35)   {$=$}
  \put(70,0)	{\Includepicfj{35}{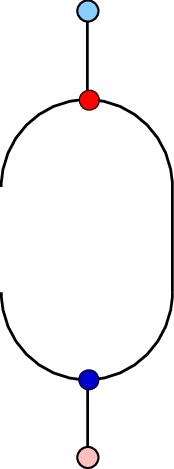}}
  } }
This equality, which can be established by inserting
$\id_{A^\vee} \,{=}\, \Fiso\cir\Fisoi$
and then using duality, allows us to remove such combinations of duality morphisms
as well. The remaining dualities directly attached to ingoing or outgoing lines
necessarily combine with structure morphism of $A$ to produce either $\Fiso$ or
$\Fisoi$. The assertion now follows from Lemma \ref{frob_proof_lemma}.
\endofproof

The same arguments as in the proof of the lemma show that a connected morphism,
entirely composed of structure morphisms of $A$ and $A^\vee$, in
$\Hom(A^\pm\oti A^\pm\oti\cdots\oti A^\pm,\one)$ equals
$\eps \cir \varphi \cir (h^\pm\,\oti\, h^\pm\,\oti\cdots\oti\, h^\pm)$,
while any such morphism in $\Hom(\one,A^\pm\oti A^\pm\oti\cdots\oti A^\pm)$
equals $(g^\pm\,\oti\, g^\pm\,\oti\cdots\oti\, g^\pm) \cir \psi \cir \eta$,
with $\psi$ and $\varphi$ as in \eqref{frob_p5}.
In a non-braided setting similar results have been obtained in \cite{coDu2}
(e.g.\ Theorem 6.11 in \cite{coDu2} corresponds to Lemma \ref{frob_proof_lemma}).

\medskip
The proof of independence of transparent subgraph is now straightforward.

\vskip .3em

\proof
The correlator $C(\ws)$ of a world sheet $\ws$ only depends on the morphism
in the category $\cC$ that is obtained by flattening the ribbon graph in the
connecting manifold. Thus when covering $\tws$ with contractible subsets,
$C(\ws)$ only depends on the collection of morphisms obtained from the ribbon
graphs in each of the contractible patches. Now assume that $\ws$ and $\ws'$
differ only in their transparent subgraphs $\TD$ and $\TD'$.
In particular $\assd(\DTD)$ differs from $\assd(\DTD')$ at most by an isotopy
on half-edges incident to insertion points. As a consequence we can assume
that $\assd(\DTD) \eq \assd(\DTD')$ by replacing $\ws'$ with an equivalent
world sheet differing only by such an isotopy. Let us focus on one connected
component of $\tws\,{\backslash}\,\assd(\DTD) \eq \tws\,{\backslash}\,\assd(\DTD')$
and choose a good open cover $\Xi$ of this component. Using equivalences given by
isotopies of $\assd|_{\TD}$ and removals of transparent tadpoles we may assume
that the images of $\assd$ and $\assd'$ coincide on the boundary of each open
set in $\Xi$, including orientation of each edge.
 \\
Next note that we may add a transparent tadpole to an edge, use specialness to
expand a bubble, and sliding of structure vertices to move one leg of the bubble
to the neighboring parallel edge (which is possible because $\D$ is connected
and each connected component of $\tws\backslash\dD$ is contractible).
As a consequence we may freely introduce a transparent edge between any two
parallel transparent edges, and hence we may assume that the graph contained
in any open set in $\Xi$ whose boundary crosses at least two edges is connected.
Now for each open set in $\Xi$, partition the edges crossing its boundary into
incoming and outgoing edges. By Lemma \ref{lemma:eqAmorph}, each open set in
$\Xi$ then contains a graph describing the same morphism in both $\ws$ and
$\ws'$. It follows that, however the transparent graphs are flattened in the
two connecting manifolds, the corresponding morphisms are the same.
 \\
We have thus shown that $C(\ws) \eq C(\ws')$.
\endofproof

%%%%%%%%%%%%%%%%%%%%%%%%%%%%%%%%%%%%%%%%%%%%%%%%%%%%%%%%%%%%%%%%%%%%%%%%%

\subsection{Fusion of defect lines}\label{app:fusion}

The tensor product of bimodules over algebras is defined as follows. Let
$A_i$, $i\eq 0,1,2,...\,,m \,{\ge}\, 2$, be algebras and
$X_i$, $i\eq 1,2,...\,,m$, a collection of $A_{i-1}$-$A_i$-bimodules with
left and right representation morphisms $\rho^\text l_i$ and $\rho^\text r_i$.
For $m \eq 2$ the tensor product $\otimes_{A_1}$ is defined as the coequalizer
  \be\label{algtensor}
  \xymatrix{
  X_1\otimes {A_1}\otimes X_2
  \ar@<.6ex>[rr]^{~~\rho^\text r_1\otimes\id_{X_2}}
  \ar@<-.6ex>[rr]_{~~\id_{X_1}\otimes\rho^\text l_2}
  && X_1\otimes X_2\ar[r]^{\otimes_{A_1}~}
  &X_1\otimes_{A_1} X_2
  } \ee
in $\cC_{A_0|A_2}$. For $m\,{>}\,2$ the multiple tensor product is defined
by invoking \eqref{algtensor} recursively.

In case the algebras $A_i$, $i\eq 0,1,2,...\, m$ are special Frobenius,
there is an alternative characterization of the tensor product, given by

\begin{lemma}\label{prop_defretr}
The morphism
  \eqpic{bnd_idem}{185}{34}{ \put(0,-3){
  \put(80,15)  {\Includepicfj{3}{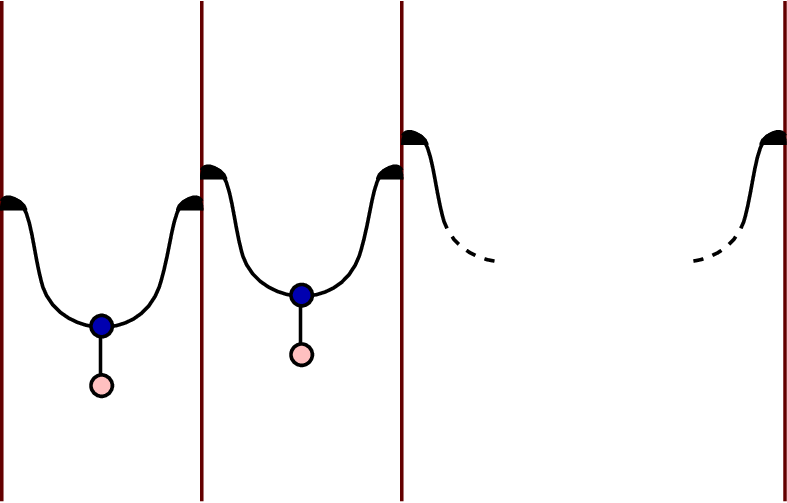}}
  \put(-12,48)	  {$P_{X_1\otimes\cdots\otimes X_{m}}~ := $}
  \put(75,5.5)	  {\scriptsize $X_1$}
  \put(75,91)	  {\scriptsize $X_1$}
  \put(85,71)	  {\small \fbox{$A_1$}}
  \put(105,5.5)	  {\scriptsize $X_2$}
  \put(105,91)	  {\scriptsize $X_2$}
  \put(114,73)	  {\small \fbox{$A_2$}}
  \put(133,5.5)	  {\scriptsize $X_3$}
  \put(133.5,91)  {\scriptsize $X_3$}
  \put(145,75)	  {\small \fbox{$A_3$}}
  \put(188,5.5)	  {\scriptsize $X_m$}
  \put(188,91)	  {\scriptsize $X_m$}
  \put(160,25)	  {\small \fbox{$A_{m\! -\! 1}$}}
  \put(160,47)	  {$\cdots$}
  } }
in $\mathrm{End}_{A_0|A_m}(X_1\,\oti\,\cdots\,\oti\,X_{m})$ is an idempotent
and has a retract $(X_1\oti_{A_1}X_2\oti_{A_2}\cdots\,\oti_{A_{m-1}}X_m,
   $\linebreak[0]$
e,r)$ in $\cC_{A|B}$. In particular $\im(P_{X_1\otimes\cdots\otimes X_{m}})
\,{\cong}\, X_1\oti_{A_1}\cdots\oti_{A_{m-1}} X_{m}$.
\end{lemma}

\begin{proof}
That \eqref{bnd_idem} is an idempotent follows in the same way as in
\cite[Eq.\,(5.127)]{fuRs4}. Consider first the case $m \eq 2$ and write
$A\,{\defas}\, A_0$, $B\,{\defas}\, A_2$, $C \,{\defas A}\,_1$,
$X_1 \,{\defas}\, X\eq (X,\rho^\text l,\rho^\text r)$ as well as
$X_2 \,{\defas}\, Y \eq (Y,\sigma^\text l,\sigma^\text r)$. By Proposition 2.4
of \cite{ffrs5} the category $\cC_{A|B}$ is idempotent-complete,
so there exists a retract $(\im(P_{X\oti Y}),e,r)$ such that
$P_{X\oti Y} \eq e\cir r$. By the representation property of $\rho^\text r$
and $\sigma^\text l$ together with the Frobenius and unit properties of $C$
it follows that $r$ indeed coequalizes $\rho^\text r\,\oti\,\id_Y$
and $\id_X\,\oti\,\sigma^\text l$:
  \be
  \xymatrix{
  X\otimes C\otimes Y
  \ar@<.6ex>[rr]^{~\rho^\text r\otimes\id_Y}
  \ar@<-.6ex>[rr]_{~\id_X\otimes\sigma^\text l}
  && X\otimes Y\ar[r]^{r~~}
  &\im(P_{X\oti Y}) \,.}
  \ee
It remains to show that $r$ has the universal property that is required for a
coequalizer. Assume that $f\colon X\,\oti\, Y\To Z$ coequalizes
$\rho^\text r\,{\oti}\,\id_Y$ and $\id_X\,{\oti}\,\sigma^\text l$ as well.
Using that $C$ is special Frobenius, it follows immediately that $f\cir
P_{X\oti Y} \eq f$, which furthermore implies the commutativity of the diagram
  \be
  \xymatrix{
  X\oti Y \ar[dr]_{f}\ar[r]^{r~~} & \im(P_{X\oti Y})\ar[d]^{f\circ e}\\
  & Z
  } \ee
Uniqueness of the morphism $f\cir e$ can be straightforwardly shown by making use
of the morphism $\otimes_C\cir e\colon \im(P_{X\oti Y})\To X\oti_C Y$; we omit the
details. This completes the proof for $m \eq 2$. The extension to $m\,{>}\,2$
is easy, and we refrain from spelling it out.
\endofproof
\end{proof}

%%%%%%%%%%%%%%%%%%%%%%%%%%%%%%%%%%%%%%%%%%%%%%%%%%%%%%%%%%%%%%%%%%%%%%%%%

\subsection{Dual bases for $\im(\proj)$}\label{app:B*}

Recall the basis $\mathcal B$ \eqref{basis_Im(P)} of the image $\im(\proj)$ of
the projector $\proj$ given by the invariant of the cobordism \eqref{CCinv}.

\begin{lemma}\label{Dual_bases}
The basis $\mathcal B^*$ dual to $\mathcal B$ is spanned by the vectors
$Z(\instord X pq\alpha\beta)$, with $\instord X pq\alpha\beta$ the cobordism
  \eqpic{instordual}{220}{100}{
  \put(-12,130) {$\mathcal N^{-1}\,\instord X pq\alpha\beta~:= $}
  \put(90,-1){ \Includepic{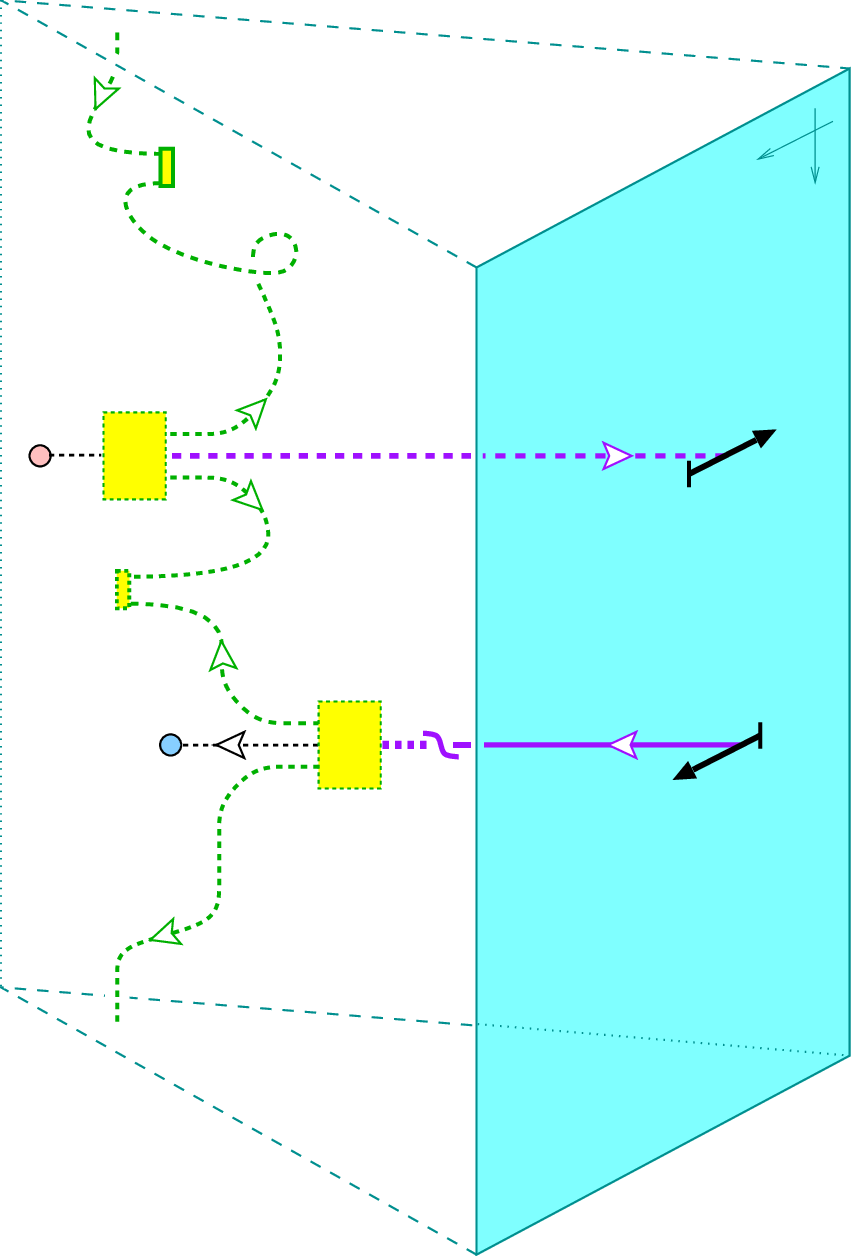}
  \put(33,194)    {\scriptsize $\lambda_{\qb q}$}
  \put(14,114.5)  {\scriptsize $\lambda_{\pb p}$}
  \put(136,199)   {\tiny $1$}
  \put(149,194)   {\tiny $2$}
  \put(24,221)    {\scriptsize $q$}
  \put(54,165)    {\scriptsize $\qb$}
  \put(42,72)     {\scriptsize $q$}
  \put(51,130)    {\scriptsize $\pb$}
  \put(45,107)    {\scriptsize $p$}
  \put(21.3,143.5){\small $\bar{\beta}$}
  \put(60.8,91)   {\small $\bar{\alpha}$}
  \put(122,78)    {\begin{turn}{27}\scriptsize $(X,+)$\end{turn}}
  \put(123,132.5) {\begin{turn}{27}\scriptsize $(X,-)$\end{turn}}
  } }
Here $\bar\alpha$ labels a basis of $\Homaa {X}{U_p\,\oti^+A\oti^-U_q}$, $\bar\beta$
labels a basis of $\Homaa A{U_\pb\,\oti^+X\oti^-U_\qb}$ and
$\,\mathcal N \eq {(\dim(U_p)\, \dim(U_q))^2} /\, {\dim(A)\; \dim(X)}$.
\end{lemma}

\begin{proof}
The assertion can be established by a simple generalization of the proof
of Eq.\ (5.20) of \cite{fjfrs} by just replacing relevant $A$-ribbons by
$X$-ribbons. The only non-obvious modification is that the equality
in (5.25) of \cite{fjfrs} needs to be replaced by
  \eqpic{dimdual}{220}{45}{
  \put(0,0){ \Includepic{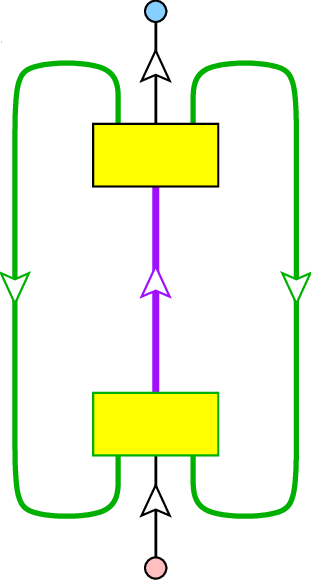}
  \put(25,75)   {\footnotesize $\bar\alpha$}
  \put(25,26)   {\footnotesize $\gamma$}
  \put(-6,80)   {\footnotesize $p$}
  \put(58,80)   {\footnotesize $q$}
  \put(31,41)   {\footnotesize $X$}
  }
  \put(80,50) {$=$}
  \put(110,6){ \Includepic{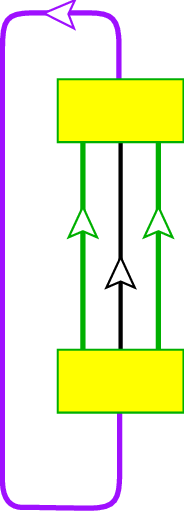}
  \put(19,71)   {\footnotesize $\bar\gamma$}
  \put(19,22)   {\footnotesize $\alpha$}
  \put(8.5,37)  {\footnotesize $p$}
  \put(31,37)   {\footnotesize $q$}
  \put(-13,80)  {\footnotesize $X$}
  }
  \put(170,50) {$=~ \dim(X)\,\delta_{\gamma,\alpha}\,.$}
  }
The first of these equalities is obtained by combining that $A$ is special and
symmetric with a deformation of the graph. The second equality follows from
the duality \eqref{DF_dual}.
\endofproof
\end{proof}

%%%%%%%%%%%%%%%%%%%%%%%%%%%%%%%%%%%%%%%%%%%%%%%%%%%%%%%%%%%%%%%%%%%%%%%%%

\subsection{Structure constants of the defect two-point function}\label{app:cdef}

The structure constants of the defect two-point function described in
\eqref{WS_twofieldsdefect} are given by (see Section 4.5 of \cite{fuRs10})
  \eqpic{2pdef}{170}{46}{
  \put(1,51) {$\cdef {X}pq\alpha\beta ~=~ \displaystyle\frac1{S_{0,0}} $}
  \put(87,-3){ \Includepic{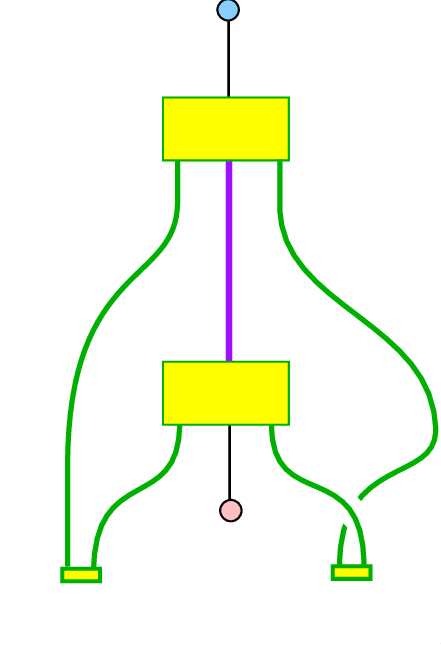}
  \put(24,81)      {\scriptsize $p $}
  \put(24,32)      {\scriptsize $\pb$}
  \put(53,81)      {\scriptsize $q $}
  \put(54,32)      {\scriptsize $\qb$}
  \put(37,42.5)    {\small $\beta$}
  \put(37,91)      {\small $\alpha$}
  \put(43.4,58)    {\small $X$}
  } }

\begin{lemma}\label{coeff_calc}
The matrix \eqref{2pdef} is non-degenerate. The inverse of
$c^{\mathrm{def}}_{X,pq}$ is given by
  \eqpic{2pinv}{260}{45}{
  \put(-3,54) {$\displaystyle \cdefinv {X}pq\alpha\beta ~=~
               S_{0,0}\,\frac{\dim(U_p)\;\dim(U_q)}{\dim(A)\;\dim(X)}$}
  \put(180,0){ \Includepic{5}
  \put(22,81)      {\scriptsize $\pb $}
  \put(24,34)      {\scriptsize $p$}
  \put(53,81)      {\scriptsize $\qb$}
  \put(49.5,34)    {\scriptsize $q$}
  \put(35.5,19)    {\small $\bar{\beta}$}
  \put(36,68.9)    {\small $\bar{\alpha}$}
  \put(41.7,52.8)  {\small $X$}
  } }
\end{lemma}

\begin{proof}
For any bimodule $X$, $\phi\iN\Homaa {U_p\,\oti^+X\oti^-U_q}A$ and
$\overline\phi\iN\Homaa X{U_\pb\,\oti^+A\oti^-U_\qb}$, define
  \eqpic{f(phi)}{420}{39}{
  \put(4,45)       {$f(\phi)~:=$}
  \put(67,0){ \Includepic{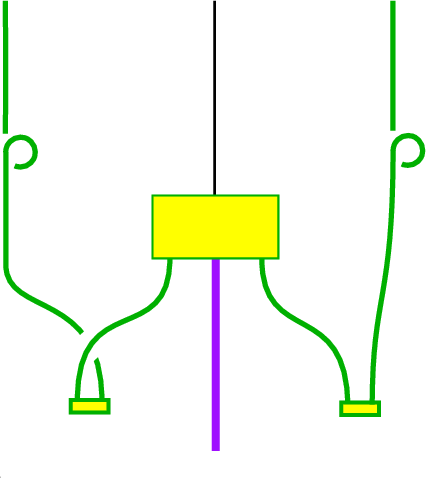}
  \put(26,24)      {\scriptsize $p $}
  \put(0,92)       {\scriptsize $\pb$}
  \put(49,24)      {\scriptsize $q $}
  \put(69.5,92)    {\scriptsize $\qb$}
  \put(35.5,43.6)  {\small ${\phi}$}
  \put(33.7,-6.7)  {\small $X$}
  \put(35,92.2)    {\small $A$}
  }
  \put(191,45)     {and}
  \put(246,45)     {$g(\overline\phi)~:=$}
  \put(306,0){ \Includepic{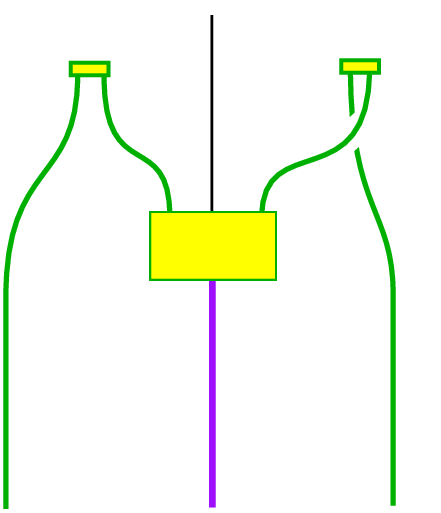}
  \put(26,69)      {\scriptsize $\pb $}
  \put(-2,-7)      {\scriptsize $p$}
  \put(53,69)      {\scriptsize $\qb$}
  \put(69.5,-7)    {\scriptsize $q$}
  \put(35.5,44.7)  {\small $\overline{\phi}$}
  \put(33.5,-9.8)  {\small $X$}
  \put(34.5,95.2)  {\small $A$}
  } }
Straightforward calculation shows that
  \be\label{fog}
  f\circ g = \frac{1}{\dim(U_p)\;\dim(U_q)}\; \id_{\Homaa X{U_\pb\oti^+A\oti^-U_\qb}} \,,
  \ee
and
  \be\label{gof}
  g\circ f = \frac{1}{\dim(U_p)\;\dim(U_q)}\; \id_{\Homaa {U_p\oti^+X\oti^-U_q}A} \,.
  \ee
Expanding
  \be
  f(\phi_\alpha^{pq})=\sum_{\gamma}\Delta_{\alpha\gamma}\, \bar\phi^{\pb\qb}_\gamma
  \qquad{\rm and} \qquad
  g(\bar\phi_\alpha^{\pb\qb})=\sum_{\gamma}\Omega_{\alpha\gamma}\, \phi^{pq}_\gamma \,,
  \ee
the equalities \eqref{fog}  and \eqref{gof} can be rewritten as
  \be\label{fogmat}
  \sum_\gamma\Delta_{\alpha\gamma}\,\Omega_{\gamma\beta}
  = \frac{1}{\dim(U_p)\,\dim(U_q)}\; \delta_{\alpha,\beta}
  = \sum_\gamma\Omega_{\alpha\gamma}\,\Delta_{\gamma\beta} \,.
  \ee
Also, combining the formulas \eqref{2pdef}, \eqref{2pinv} and \eqref{f(phi)}
it follows that
  \eqpic{CCinvx}{385}{40}{
  \put(0,47)  {$\dsty\sum_\gamma\cdefinv {X}pq\alpha\gamma \,\cdef {X}pq\gamma\beta
                ~=~\frac{\dim(U_p)\,\dim(U_q)}{\dim(A)\,\dim(X)}$}
  \put(248,-2){ \Includepic{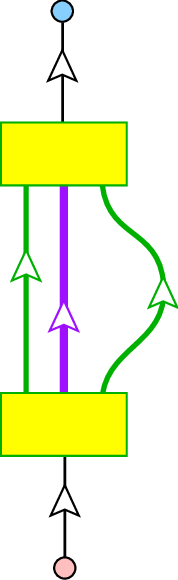}
  \put(4,75)       {\footnotesize $g(\bar\alpha)$}
  \put(8,26)       {\footnotesize $\bar\gamma$}
  \put(-3,42)      {\footnotesize $p$}
  \put(30,42)      {\footnotesize $q$}
  \put(15,50)      {\footnotesize $X$}
  }
  \put(307,-2){ \Includepic{6a}
  \put(22,75)      {\footnotesize $f(\gamma)$}
  \put(25,26)      {\footnotesize $\beta$}
  \put(-6,80)      {\footnotesize $\pb$}
  \put(58,80)      {\footnotesize $\qb$}
  \put(31,41)      {\footnotesize $X$}
  } }
This equality is most easily checked backwards. The rewriting of the left component
of the ribbon graph is immediate from the second equality in \eqref{f(phi)},
while the rewriting of the right component involves a deformation.
We then finally have
  \be
  \bearll
  \displaystyle \sum_\gamma \cdefinv {X_\mu}pq\alpha\gamma\,
  \cdef {X_\mu}pq\gamma\beta \!\!\!&\displaystyle
  = \dim(U_p)\;\dim(U_q)
  \,\sum_{\gamma,\delta,\eps} \Omega_{\alpha\delta}\,
  \Delta_{\gamma\eps}\, \delta_{\delta,\gamma}\, \delta_{\beta,\eps}
  \\&\displaystyle
  = \dim(U_p)\;\dim(U_q)\,\sum_{\gamma,\eps}\Omega_{\alpha\gamma}\,\Delta_{\gamma\beta}
  =\delta_{\alpha,\beta} \,,
  \eear
  \ee
where the first equality holds by the duality \eqref{DF_dual} of bases and the
third by \eqref{fogmat}.
\endofproof
\end{proof}

\vskip 1.5em

\noindent{\sc Acknowledgments:}
 \\
The authors were supported by the Research Links Programme of the Swedish Research
Council (VR) under project no.~348-2008-6049.
JFj is supported by the ESF network ``Interactions of Low-Dimensional Topology and
Geometry with Mathematical Physics (ITGP)'', the China Science Postdoc grant
no.~020400383, the Priority Academic Program Development
of Jiangsu Higher Education Institutions (PAPD), NSFC grant no.~10775067, and the
Chinese Central Government's 985 Project grants for Nanjing University.
JFu is largely supported by VR under project no.\ 621-2009-3993.
\\
JFj and JFu are grateful to Hamburg University, and in particular to
Christoph Schweigert and Astrid D\"orh\"ofer, for their hospitality when
parts of this paper were written.

\newpage

%%%%%%%%%%%%%%%%%%%%%%%%%%%%%%%%%%%%%%%%%%%%%%%%%%%%%%%%%%%%%%%%%%%%%%%%

 \newcommand\wb{\,\linebreak[0]} \def\wB {$\,$\wb}
 \newcommand\Bi[1]    {\bibitem{#1}}
 \newcommand\Erra[3]  {\,[{\em ibid.}\ {#1} ({#2}) {#3}, {\em Erratum}]}
 \newcommand\J[5]     {{\em #5}, {#1} {#2} ({#3}) {#4} }
 \renewcommand\K[6]   {{\em #6}, {#1} {#2} ({#3}) {#4} {\tt[#5]} }
 \newcommand\Prep[2]  {{\em #2}, preprint {\tt #1} }
 \newcommand\PhD[2]   {{\em #2}, Ph.D.\ thesis (#1)}
 \newcommand\BOOK[4]  {{\em #1\/} ({#2}, {#3} {#4})}
 \newcommand\inBO[7]  {{\em #7}, in:\ {\em #1}, {#2}\ ({#3}, {#4} {#5}), p.\ {#6}}
 \def\adma  {Adv.\wb in Math.}
 \def\apcs  {Applied\wB Cate\-go\-rical\wB Struc\-tures}
 \def\atmp  {Adv.\wb Theor.\wb Math.\wb Phys.}
 \def\comp  {Com\-mun.\wb Math.\wb Phys.}
 \def\fiic  {Fields\wB Institute\wB Commun.}
 \def\jhep  {J.\wb High\wB Energy\wB Phys.}
 \def\njop  {New\wB J.\wb Phys.}
 \def\nuci  {Nuovo\wB Cim.}
 \def\nupb  {Nucl.\wb Phys.\ B}
 \def\phlb  {Phys.\wb Lett.\ B}
 \def\remp  {Rev.\wb Mod.\wb Phys.}
 \def\taac  {Theo\-ry\wB and\wB Appl.\wb Cat.}
 \def\cntp  {Com\-mun.\wB Number\wB Theory\wB Phys.}
 \def\pspm  {Proc.\wb Symp.\wB Pure\wB Math.}

\end{document}